\documentclass[a4paper,11pt,openright]{article} 
\usepackage[english]{babel} 
\usepackage[utf8]{inputenc} 
\usepackage{amsmath, amssymb, stmaryrd,amsthm,amsfonts, amscd,color,enumerate,eucal,latexsym,mathrsfs,mathtools,cancel,tikz}
\usepackage{epsfig, hieroglf, protosem}  
\usepackage{simplewick, bm}
\usepackage{verbatim}
\usepackage{slashed}
\usepackage{hyperref}
\usepackage[all,cmtip]{xy}
\usetikzlibrary{matrix,calc}

\numberwithin{equation}{section}

\definecolor{PaColor}{RGB}{77,7,255}
\definecolor{ClaColor}{RGB}{46,0,128}
\definecolor{NiColoRed}{RGB}{255,77,77}
\definecolor{NiCitation}{RGB}{0,181,26}
\definecolor{AlColor}{RGB}{77,0,154}
\definecolor{BeaColor}{rgb}{0.0,0.44,1.0}

\newtheoremstyle{TheoremStyle}
{3pt}
{3pt}
{\slshape}
{}
{\bf}
{:}
{.5em}
{}

\theoremstyle{TheoremStyle}
\newtheorem{theorem}{Theorem}[section]
\newtheorem{corollary}[theorem]{Corollary}
\newtheorem{proposition}[theorem]{Proposition}
\newtheorem{lemma}[theorem]{Lemma}
\newtheorem{definition}[theorem]{Definition}
\newtheorem{remark}[theorem]{Remark}
\newtheorem{example}[theorem]{Example} 

\newcommand{\fish}{\begin{tikzpicture}[thick,scale=1.2]
		\draw[red] (0,0) edge [out=30,in=-30] node[above] {} (0,.35);
		\draw (0,0) edge [out=150,in=210] node[above] {} (0,.35);
\end{tikzpicture}}
\newcommand{\fiammifero}{\begin{tikzpicture}[thick,scale=1.2]
		\draw (0,0) -- (0,0.3);
		\filldraw (0,0.3)circle (1pt);
\end{tikzpicture}}
\newcommand{\propagatore}{\begin{tikzpicture}[thick,scale=1.2]
		\draw (0,0) -- (0,0.3);
\end{tikzpicture}}

\newcommand{\fiammiferoCC}{\begin{tikzpicture}[thick,scale=1.2]
		\draw[red] (0,0) -- (0,0.3);
		\filldraw[red] (0,0.3)circle (1pt);
\end{tikzpicture}}

\newcommand{\propagatoreCC}{\begin{tikzpicture}[thick,scale=1.2]
		\draw[red] (0,0) -- (0,0.3);
\end{tikzpicture}}

\title{A microlocal investigation of stochastic partial differential equations for spinors with an application to the Thirring model}

\author{
 Alberto Bonicelli\thanks{AB:
			Dipartimento di Fisica,
		Universit\`a degli Studi di Pavia \& INFN, Sezione di Pavia, 
		Via Bassi 6,
		I-27100 Pavia,
		Italia;
		\mbox{alberto.bonicelli01@universitadipavia.it}}
\and
	Beatrice Costeri\thanks{BC:
		Dipartimento di Fisica,
		Universit\`a degli Studi di Pavia,
		Via Bassi 6,
		I-27100 Pavia,
		Italia;
		\mbox{beatrice.costeri01@universitadipavia.it}}
	\and
	\underline{Claudio Dappiaggi}\thanks{CD: \underline{Corresponding Author},
		Dipartimento di Fisica,
		Universit\`a degli Studi di Pavia \& INFN, Sezione di Pavia, 
		Via Bassi 6,
		I-27100 Pavia,
		Italia;
		\mbox{claudio.dappiaggi@unipv.it}}
\and
Paolo Rinaldi\thanks{PR: Institute for Applied Mathematics, Universit\"at Bonn, 
	Endenicher Allee 60,
	D-53115 Bonn,
	Germany;
	\mbox{rinaldi@iam.uni-bonn.de}
}}

\date{\today}

\begin{document}
\maketitle
\begin{abstract}
	 On a $d$-dimensional Riemannian, spin manifold $(M,g)$ we consider non-linear, stochastic partial differential equations for spinor fields, driven by a Dirac operator and coupled to an additive Gaussian, vector-valued white noise. We extend to the case in hand a procedure, introduced in \cite{Dappiaggi:2020gge} for the scalar counterpart, which allows to compute at a perturbative level the expectation value of the solutions as well as the associated correlation functions accounting intrinsically for the underlying renormalization freedoms. This framework relies strongly on tools proper of microlocal analysis and it is inspired by the algebraic approach to quantum field theory. As a concrete example we apply it to a stochastic version of the Thirring model proving in particular that it lies in the subcritical regime if $d\leq 2$. 
\end{abstract}

\paragraph*{Keywords: stochastic partial differential equations, microlocal analysis, algebraic quantum field theory, Thirring model}
\small{}
\paragraph*{MSC 2020:} 81T05, 60H17.

\tableofcontents

\section{Introduction}\label{Sec: introduction}

Stochastic partial differential equations (SPDEs) play a prominent r\^{o}le in many models, since they allow to encode in a systematic and mathematically efficient way the random behaviour which is inherent to several physical phenomena. They occur especially in the realm of complex systems, such as interface dynamics and turbulence, but also in quantum field theory as the building block of a procedure known as stochastic quantization \cite{PW81,DGR23}. In addition to these considerations, in the past years, we have witnessed several remarkable leaps forward in our understanding of the solution theory of non-linear SPDEs, thanks to new successful approaches, among which noteworthy are the theory of regularity structures \cite{Hairer14, Hairer15} and paracontrolled calculus \cite{Gubinelli}. On the one hand, all these methods allow to establish local existence and uniqueness of the solutions of the underlying equations for a large class of scenarios, many of which are at the heart of important physical models, {\it e.g.}, the KPZ equation \cite{Hairer13}. On the other hand, the non-linear character of the equations under investigation, combined with the presence of a random source, entails that all the existing frameworks encompass a renormalization procedure, necessary to tame otherwise ill-defined products of distributions. 

In addition to these considerations, it is worth emphasizing that, for physical applications, it is important also to have as much information as possible on the explicit form of the solutions of the non-linear SPDE under scrutiny, as well as on the associated correlation functions. This necessity prompted the development of an approach complementary to those highlighted above and first outlined in \cite{Dappiaggi:2020gge}. The inspiration, as well as the starting point, for such work comes from the algebraic approach to quantum field theory, see \cite{BFDY15, Rejzner:2016hdj}. This has allowed to devise a versatile framework to analyze at a perturbative level a large class of non-linear models, accounting intrinsically for the underlying renormalization freedoms thanks to an approach, first advocated by Epstein and Glaser \cite{EpsteinGlaser}. 

In \cite{Dappiaggi:2020gge} the will to translate the algebraic approach to the realm of SPDEs has lead to consider as starting point scalar semi-linear SPDEs of the form $E\widehat{\Phi}=\xi+F[\widehat\Phi]$, where  $\widehat{\Phi}$ is a random distribution on an underlying manifold $M$, $\xi$ denotes the standard  Gaussian, real white noise centered at $0$ whose covariance is $\mathbb{E}[\xi(x)\xi(y)]=\delta(x-y)$. Furthermore $F:\mathbb{R}\to\mathbb{R}$ is a non-linear potential which can be considered for simplicity of polynomial type, while $E$ is a linear operator either of elliptic or of parabolic type -- see \cite{Dappiaggi:2020gge} for more details and comments. The rationale of the algebraic approach is to replace $\widehat{\Phi}$ with a suitable distribution on $M$ with values in the polynomial functionals over $C^\infty(M)$ which acts as the generator of a commutative algebra with respect to the pointwise product. The stochastic nature codified by the centered Gaussian white noise $\xi$ is encoded subsequently by means of a two-step procedure. First of all, one deforms the product of the algebra of functionals in such a way to encode the information of the two-point correlation function. The r\^{o}le of the expectation value is taken by the evaluation of the algebra elements at the zero configuration. Without entering into technical details at this stage, for which we refer to \cite{Dappiaggi:2020gge}, we stress that, not only one can encompass at an algebraic level the data brought by the white noise, but it is also possible to encode in the deformation map all the information on the underlying renormalization freedoms and ambiguities. More precisely we adapt to the case in hand the approach of \cite{Brunetti-Fredenhagen-00} which is in turn based on techniques proper of microlocal analysis. 

Among the limitations of \cite{Dappiaggi:2020gge} we highlight that the algebraic framework has been devised so to be applicable only to real scalar SPDEs, although, recently in \cite{BDR23}, it has been extended to the analysis of the stochastic non linear Schr\"odinger equation. Yet, among the plethora of physically relevant models, a notable class is the one whose underlying kinematic configurations are codified by spinors. More precisely, in the realm of differential geometry a distinguished class of manifolds is that admitting a spin structure, see Section \ref{Sec: geometric setting} for a succinct review. Among the data which it encompasses, notable is $SM$ the spin bundle which is a principal bundle whose underlying structure group is the double cover either of $SO(d)$ or of $SO(d-1,1)$ depending whether the underlying $d$-dimensional manifold is endowed with a Riemannian or with a Lorentzian structure. On top of $SM$, one can construct a natural, associated vector bundle, denoted in this paper by $DM$. Its sections, namely elements of $\Gamma(DM)$, are known as {\em spinor fields}, while those of the dual vector bundle $D^*M$ are called {\em cospinor fields}. In addition, on $\Gamma(DM)$ one can introduce a natural, one parameter, family of linear differential operators which are known in the physics literature as the massive Dirac operators, see Equation \eqref{Eq: Dirac Operator}, out of which one can construct the renown Dirac equation. 

In many quantum field theoretical models, this rules the dynamics of Fermionic particles, whose behaviour is notably different from their Bosonic counterpart. This does not stem only from the underlying, different, spin-statics, but it is also manifest at the level of model building. Dirac fields are never considered in combination with a self-interaction term, since, when the underlying dimension is $d\geq 4$, these contributions to the dynamics fail to be renormalizable. Yet, this distinguished feature is no longer present when one considers lower dimensional manifolds and specific self-interacting models can be built. Among these, the most renown are the Gross-Neveu \cite{Gross:1974jv} and the Thirring models, \cite{Thirring}. The former is especially noteworthy as a toy model for QCD while the latter is one of the building blocks of Coleman's equivalence, which is in turn a special instance of a phenomenon known as Bosonization \cite{Col75}. In this paper we will be particularly interested in the Euclidean version of the Thirring model as a prototypical example of a self-interacting theory for Fermions. We will be considering a scenario where the spinor field is kept classical, but the underlying model is modified by the addition of an additive vector valued Gaussian noise. On the one hand, this takes into account that spinors and cospinors are vector valued fields, but, on the other hand, it is devised in such a way to induce non vanishing correlations only between spinor and cospinor fields. Via the random contribution to the model we introduce a behaviour typical in quantum field theories of the two-point correlation functions of Fermionic fields. 

In particular, we shall prove that the algebraic approach to stochastic partial differential equations can be extended to this class of models. This allows not only to construct the expectation value and the correlation functions of the solutions of the underlying equations at every order of perturbation theory, but it accounts intrinsically for the renormalization freedoms, extending to the case in hand the techniques devised in \cite{Dappiaggi:2020gge}. Furthermore, we are able to give a sharp estimate on the spacetime dimension for which the non linear SPDE under investigation lies in the subcritical regime. It is worth mentioning that our method strongly enjoys from the analysis in \cite{Rej11} concerning the algebra of functionals for Fermionic field theories, although, contrary to this paper, we feel unnecessary to make an explicit use of anticommuting variables of Grassmann type to account for the specific behaviour of the underlying fields. To conclude, we also stress that our investigation on SPDEs associated to Fermions is not an unicum in the literature. Recent works on these models have highlighted their relevance, especially in connection to the stochastic quantization program, see \cite{Albeverio, DeVecchi}, although the model considered in these works is different from the one analyzed in this paper.

\vskip .2cm

The paper is organized as follows: In Section \ref{Sec: geometric setting} we give an overview of the key geometric structures necessary to define spinor and cospinor fields, while in Section \ref{Sec: on the model} we outline the stochastic Thirring model. Section \ref{Sec: algebras} is devoted to introducing the algebra of vector valued distributions and its main structural properties. In Section \ref{Sec: wavefront set} we discuss suitable constraints on the singular structure of the functional derivatives of the admissible algebra elements. Section \ref{Sec: local deformation} contains the first main result of this work, namely Definition \ref{Def: deformed product} which introduces the formal structure of the deformed product of the algebra under investigation, while Theorem \ref{Prop: deformation map} is one of the main results of this work, namely it proves the existence of the sought deformation. The multi-local counterpart of this construction is discussed in Section \ref{Sec: nonlocal deformation algebra}, while in Section \ref{Sec: Thirring} we investigate the application to the stochastic Thirring model, at the level both of expectation values in Section \ref{Sec: expectation values} and of correlation functions, see Section \ref{Sec: Two-point correlation functions}. In Section \ref{Sec: Renormalized Equation} we discuss the ensuing renormalized equation while in Section \ref{Sec: Renormalizability} we prove that the Euclidean stochastic Thirring model lies in the subcritical regime if the manifold has dimension $d\leq 2$. Eventually, in the appendices, we recollect some basic facts about vector valued white noises and about Clifford algebras.

\subsection{Geometric Setting}\label{Sec: geometric setting}

In the following, we introduce the key ingredients, which are at the heart of our analysis. Henceforth, with $(M,g)$ we denote a smooth and oriented Riemannian manifold of dimension $\dim M=d\geq 1$. In addition, since we are interested in working with Dirac fields, we assume that $M$ is {\em spin}, namely it admits a spin structure $SM\equiv (SM,\pi)$ which amounts to the underlying second Stiefel-Whitney class $w_2(M):=\dim H^2(M;\mathbb{Z}_2)$ being trivial -- see \cite[Chap. II \S 2]{Lawson}. Here with $(SM,\pi)$ we denote a principal $Spin(d)$-bundle, where $Spin(d)$ is the {\em spin group} associated to the Euclidean space $\mathbb{R}^d$, together with a morphism of principal bundles $\pi:SM\to\mathfrak{F}(M)$ covering the identity. Here $\mathfrak{F}(M)$ is the bundle of oriented, orthonormal frames associated to $(M,g)$ and, for all $(x,s)\in SM\times Spin(d)$, it holds that
\begin{equation*}
	\pi(r_{SM}(x,s))=r_{\mathfrak{F}(M)}(\pi(x),\Lambda(s)),\quad\forall x\in SM, s\in Spin(d),
\end{equation*}
where $r_{SM}$ and $r_{\mathfrak{F}(M)}$ are the right action of $Spin(d)$ on $SM$ and of $SO(d)$ on $\mathfrak{F}(M)$ respectively, while $\Lambda:Spin(d)\to SO(d)$ is the double covering map. Although $M$ can admit more than one non-equivalent spin structure, depending on the dimension of $H^1(M;\mathbb{Z}_2)$, we assume that this is arbitrary, but fixed. 

On top of $(M,g)$ we are interested in introducing spinor fields which shall represent the kinematic configurations of our models. To this end we assume that the reader is familiar with the basic notions concerning Clifford algebras, see {\it e.g.}, \cite[Chap. I]{Lawson} and we outline only the concepts which are strictly necessary to keep this paper self-contained. 
\begin{definition}\label{Def: spinor space}
	Given a spin manifold $(M,g)$ such that $\dim M=d$, we call \textbf{spinor space} $\Sigma_d=\mathbb{C}^{N_d}$, where $N_d=2^{\lfloor\frac{d}{2}\rfloor}$ where $\lfloor, \rfloor$ is the floor function. Denoting by $Cl(d)$ the Clifford algebra on the Euclidean space $\mathbb{R}^d$ and by $\mathbb{C}l(d):=Cl(d)\otimes\mathbb{C}$ its complexification, a \textbf{spinor representation} of $\mathbb{C}l(d)$ is an isomorphism
	\begin{equation*}
		\sigma_d:\mathbb{C}l(d)\rightarrow End(\Sigma_d),
	\end{equation*}
where $End(\Sigma_d)$ denotes the collection of all endomorphisms on $\Sigma_d$.
\end{definition}

\noindent As a byproduct of Definition \ref{Def: spinor space} we can introduce the main geometric structures of this work.

\begin{definition}[Spinor bundle]\label{Def: Dirac bundle}
	Given a $d$-dimensional Riemannian Spin manifold $(M,g)$ with spin structure $SM$ and given a spinor representation $\sigma_d$ as per Definition \ref{Def: spinor space}, the (Dirac) spinor bundle $DM$ is the associated vector bundle
	\begin{equation*}
		DM:=SM\times_{\sigma_d}\Sigma_d.
	\end{equation*}
	Consequently, a \textbf{spinor field} on $M$ is a smooth section of $DM$, {\it i.e.}, $\psi\in\Gamma(DM)$. At the same time, a \textbf{cospinor field} is a section of $D^\ast M$, the vector bundled dual to $DM$, {\it i.e.}, $\bar{\psi}\in\Gamma(D^\ast M)$.
\end{definition}

Having introduced the kinematic configurations in terms of spinor and cospinor fields, we are interested in associating to them a dynamical model, such as the Thirring one \cite{Thirring}. To this end, one must individuate suitable differential operators which are naturally tied to the underlying spin structure. This is a topic which has been thoroughly studied in the literature and we summarize here the main concepts, while a reader can still consult \cite{Jost} for additional information. Let us fix an oriented, orthochronous, orthonormal co-frame $e=(e^\mu)_{\mu=1,\dots,d}$ on $M$ once and for all. In other words all $e^\mu$ are no-where vanishing one-forms on $M$ such that
$g=\delta_{\mu\nu}e^\mu\otimes e^\nu$, where $g$ is the underlying Riemannian metric. Fixing $e$ is completely equivalent to choosing a frame $\epsilon=(\epsilon_\mu)$, 
namely a section of the frame bundle $\mathfrak{F}(M)$ which can be obtained from $e$ setting $\epsilon_\mu=\delta_{\mu\nu} e^\nu$. 
Using a fixed co-frame $e$ of $M$, one can specify a one-form $\gamma$ over $M$  taking values in the bundle of endomorphisms of the Dirac bundle $DM$: 
\begin{align}
	\gamma:TM\to\textit{End}(DM),\quad v\mapsto e^\mu(v)\gamma_\mu,
\end{align}
where $\{e^\mu(v)\in\mathbb{R}\}_{\mu=1,\dots,d}$ are the components of $v\in T_xM$ with respect to the frame $\epsilon$ obtained raising the indices of the fixed co-frame $e$, 
namely $v=e^\mu(v)\epsilon_\mu$. Here, the collection $\{\gamma_\mu\}_{\mu=1,\dots, d}$ are the gamma-matrices on $\mathbb{R}^d$ as discussed in Appendix \ref{App: Clifford Algebras}.

The second ingredient that we need is a suitable covariant derivative on the spinor and cospinor bundles. Considering $DM$ for definiteness, this can be defined as follows
\begin{align}\label{Eq: spin connection}
	\nabla:\Gamma(TM)\otimes\Gamma(DM)\to\Gamma(DM),\quad (X,\sigma)\mapsto \nabla_X\sigma
	=\partial_X\sigma+\frac{1}{4}X^\mu\Gamma_{\mu\nu}^\rho\gamma_\rho\gamma^\nu\sigma,
\end{align}
where, considering a local trivialization of $DM$, $\sigma$ is here regarded as a smooth $\Sigma_d$-valued function on a subset of $M$, 
$X^\mu=e^\mu(X)$ are the components of $X$ in the fixed frame and $\Gamma_{\mu\nu}^\rho=e^\rho(\nabla_{\epsilon_\mu}\epsilon_\nu)$ 
are the Christoffel symbols of the Levi-Civita connection with respect to the given frame. 
The covariant derivative is naturally extended to cospinors by imposing the identity 
\begin{equation*}
	\partial_X(\omega(\sigma))=(\nabla_X\omega)(\sigma)+\omega(\nabla_X\sigma), 
\end{equation*}
where $\omega\in\Gamma(D^*M)$.

\noindent Starting from these premises we can introduce eventually the (massive) Dirac operator $\slashed{D}:\Gamma(DM)\to\Gamma(DM)$ and its dual counterpart $\slashed{D}^*:\Gamma(D^\ast M)\to\Gamma(D^\ast M)$
  \begin{equation}\label{Eq: Dirac Operator}
  	\slashed{D}:=i\slashed{\nabla}-m,\hspace{2cm} \slashed{D}^\ast:=i\slashed{\nabla}+m,
  \end{equation}
where $\nabla$ represent the covariant derivative as per Equation \eqref{Eq: spin connection} while $\slashed{\nabla}:=\gamma^\mu\nabla_\mu$. Here $m$ represents a mass parameter which is assumed to lie in the interval $[0,\infty)$. 

\begin{remark}\label{Rem: D and D*}
	Observe that one can introduce the counterparts of $\slashed{D}$ and $\slashed{D}^\ast$ acting respectively on $\Gamma(D^\ast M)$ and on $\Gamma(DM)$. Since their formal expression is the same as in Equation \eqref{Eq: Dirac Operator}, with a slight abuse of notation, we shall denote them by the same symbol. In each instance, it will be clear from the context which is the specific operator to which we are referring to.
\end{remark}

\noindent A direct application of the Lichnerowicz-Weizenb\"ock formula, see \cite[Thm 8.8]{Lawson}, entails that 

\begin{equation}\label{Eq: Dirac to Laplacian}
	\slashed{D}^\ast \slashed{D}=\slashed{D}\,\slashed{D}^\ast=\left(-\Delta+\frac{R}{4}+m^2\right)\mathbb{I}_{\Sigma_d}, 
\end{equation}
where $\mathbb{I}_{\Sigma_d}$ is the identity matrix acting on $\Sigma_d$, $\Delta=g^{\mu\nu}\nabla_\mu\nabla_\nu$ is the Laplacian built out of the metric $g$ while $R$ is the associated scalar curvature. To conclude the section and for future convenience we highlight a few structural properties concerning the Dirac operator and its dual as per Equation \eqref{Eq: Dirac Operator}. In particular we observe that Equation \eqref{Eq: Dirac to Laplacian} entails that the operators $\slashed{D}^*\slashed{D}$ and $\slashed{D}\slashed{D}^*$ are manifestly elliptic. In the following we shall only consider Riemmanian manifolds $(M,g)$ such that the operator $-\Delta+\frac{R}{4}+m^2$ admits a fundamental solution, namely $G_{\Delta}:\Gamma_0(DM)\to\Gamma (DM)$ such that 
$$(\slashed{D}^*\slashed{D})G_\Delta=(\slashed{D}\slashed{D}^*)G_{\Delta}=\mathbb{I},$$
where $\mathbb{I}$ is the identity operator acting on scalar functions. A direct inspection once more of Equation \eqref{Eq: Dirac to Laplacian} entails that
\begin{equation}\label{Eq: fundamental solution Dirac}
	G_\psi:=\slashed{D}^* G_\Delta,
\end{equation}
is a fundamental solution for the Dirac operator acting on spinors. As a matter of fact the subscript $\psi$ serves as a visual remainder of the space on which $G_\psi$ is acting. Therefore we shall also make an extensive use of $G_{\bar{\psi}}:=\slashed{D} G_\Delta:\Gamma_0(D^\ast M)\to\Gamma (D^\ast M)$ which is here read as an operator acting on cospinors. 

\begin{remark}\label{Rem: Parametrix}
From a structural viewpoint, we could drop the hypothesis that we consider Riemannian manifolds on which the Dirac operator admits a fundamental solution since, in the worst case scenario, the operator $-\Delta+\frac{R}{4}+m^2$ admits a parametrix \cite[Thm 4.4]{Wells}, namely $P:\Gamma_0(DM)\to\Gamma (DM)$ such that 
$$(\slashed{D}^*\slashed{D})P=(\slashed{D}\slashed{D}^*)P=\mathbb{I}+K,$$
where $\mathbb{I}$ is the identity operator while $K$ is a smoothing operator. In turn $S_\psi:=\slashed{D}^*P$ and $S^*_{\bar{\psi}}:=\slashed{D}P$ are parametrices of $\slashed{D}$ and $\slashed{D}^\ast$ acting respectively on spinors and on cospinors.
\end{remark}

\begin{remark}\label{Rem: Trivial Bundle}
	In order to make the main results and concepts at the heart of our analysis more accessible to the readers, henceforth we shall only consider {\em trivial} Dirac bundles, namely $DM$ is isomorphic as a vector bundle to $M\times\Sigma_d$, see Definition \ref{Def: spinor space}. Consequently $D^\ast M$ is isomorphic to $M\times\Sigma^*_d$, while $\Gamma_{(0)}(DM)\simeq C^\infty_{(0)}(M;\Sigma_d)\simeq C^\infty_{(0)}(M)\otimes\Sigma_d$. We highlight that this assumption does not entail a reduction in generality of our results. Using a suitable and standard partition of unity argument, subordinated to a local trivialization of both $DM$ and $D^\ast M$, we can extend all the following discussions also to this scenario. Yet, the price to pay is a heavy notation and we feel that it might be discomforting for a reader with no benefit in clarity.
\end{remark}

\subsection{The Thirring Model}\label{Sec: on the model}
As mentioned in the introduction, in the following we shall devise a framework to analyze stochastic partial differential equations whose kinematic configurations are spinor fields. As a specific example, we shall consider a stochastic variation of the Thirring model, which we will introduce succinctly in the following. More precisely, we consider the Euclidean space $\mathbb{R}^d$ and we denote by $\psi,\bar{\psi}$ a section respectively of the Dirac bundle $D\mathbb{R}^d$ and of its dual counterpart $D^*\mathbb{R}^d$, where $D\mathbb{R}^d\equiv\mathbb{R}^d\times \Sigma_d$ and, on account of Definition \ref{Def: spinor space}, $\Sigma_d=\mathbb{C}^{\lfloor\frac{d}{2}\rfloor}$. Consequently $\psi,\bar{\psi}:\mathbb{R}^d\to\Sigma_d$ and, working at the level of components, the dynamics is ruled by
\begin{equation}
	\label{Dirac}
	\begin{cases}
		\slashed{D}^{\rho}_{\sigma} \psi^{\sigma} + \lambda (\bar{\psi} \gamma^{\mu} \psi)^n(\gamma_{\mu})_{\sigma}^{\rho}
		\psi^{\sigma} = \xi^{\rho}, \\
		\slashed{D}^{*\, \sigma}_{\rho} \bar{\psi}_{\sigma} + \lambda (\bar{\psi} \gamma^{\mu} \psi)^n \bar{\psi}_{\sigma} (\gamma_{\mu})_{\rho}^{\sigma} = \bar{\xi}_{\rho}, \\
	\end{cases}
\end{equation}
where $\lambda \in \mathbb{R}$ is the coupling constant. In the following it might be necessary to multiply $\lambda$ by a function $g\in C^{\infty}_0(\mathbb{R}^d)$ which denotes an arbitrary spacetime cutoff, necessary to avoid infrared divergences. Furthermore, we denote by
\begin{flalign*}
	& \slashed{D}= (i \slashed{\nabla}-m) := (i \gamma^{\mu} \nabla_{\mu} - m) \\
	& \slashed{D}^* = (i \slashed{\nabla}+m) := (i \gamma^{\mu} \nabla_{\mu} + m)
\end{flalign*}
the Dirac operator and its formal adjoint as per Equation \eqref{Eq: Dirac Operator}. We recall that $m\geq 0$ is here an arbitrary, but fixed mass parameter. The source term is represented by a vector-valued spinor white noise, see Definition \eqref{Def: fermionic white noise}, denoted by 
\begin{equation*}
	\underline{\xi} := \begin{pmatrix}
		\xi \\
		\bar{\xi}
	\end{pmatrix}
	\in \mathcal{D}'(\mathbb{R}^d; L^2((\mathfrak{W}, \mathcal{F}, \mathbb{P}); \Sigma_d \times \Sigma_d^*)),
\end{equation*}
which is completely characterized by its covariance as per Equation \eqref{Eq: noncommutative white noise}.

In absence of the white noise source, Equation \eqref{Dirac} is known as \emph{Thirring model} \cite{Thirring}, which describes the dynamics of a massive, interacting Fermionic model in $d = 2$ Euclidean dimensions. A rigorous study of this model can be found in \cite{FMRS86}. As mentioned in the introduction its relevance in physics resides in its notable applications within the realm of statistical mechanics, especially when discussing the phenomenon of Bosonization, which is widely known in the physics literature \cite{BFM09, Col75} and it plays a prominent r\^ole in theoretical condensed matter physics as well as in quantum field theory \cite{FrSo93}. 

\section{Analytical and Algebraic Tools}
Since the goal of this work is to extend the algebraic and microlocal approach to SPDEs advocated in \cite{CDDR20} and \cite{BDR23} to models having spinors as underlying configurations, we devote this section to introducing the key analytic tools, fixing in the meanwhile all notations and conventions. We recall that, as highlighted in Remark \ref{Rem: Trivial Bundle}, we shall only consider trivial Dirac bundles.

\subsection{Algebras of Functional Valued Distributions}\label{Sec: algebras}
In the following we build upon the geometric structures outlined in Section \ref{Sec: geometric setting} to extend to the case in hand the definitions of functional-valued distributions introduced in \cite{Dappiaggi:2020gge} for scalar theories and in \cite{BDR23} for complex-valued fields. We assume that the reader is familiar with the basic notions of microlocal analysis referring to \cite{Hormander-I-03} for additional details. In addition we shall enjoy from \cite[App. A \& App. B]{Dappiaggi:2020gge}, which is a succinct account of the key results, to which we shall refer often.

Having in mind an application to the Thirring model, see Section \ref{Sec: Thirring}, in the following we consider $\Sigma_d$ as per Definition \ref{Def: spinor space}, although, as far as the content of this section is concerned, it could be replaced by an arbitrary, finite dimensional, complex vector space. Furthermore we define for later convenience
\begin{equation}\label{Eq: W space}
W:= \bigoplus_{p, q=0}^\infty W^{(p,q)}:=\bigoplus_{p,q=0}^\infty ((\Sigma_d)^{\otimes p} \oplus (\Sigma_d^\ast)^{\otimes q}),
\end{equation}
where $\Sigma_d^\ast$ is the vector space, dual to $\Sigma_d$. 

In the following we introduce the basic kinematic structures at the heart of our investigation and, as already mentioned in the introduction, we stress that we rely on an approach which is based on considering a class of suitable functionals. This is an adaptation to the analysis of SPDEs of a framework which has been first introduced in the algebraic approach to quantum field theory since it allows to encode efficiently the existence of non-trivial, underlying correlation functions in terms of a deformation of the pointwise product between functionals -- see {\it e.g.}, \cite{BFDY15,Fredenhagen:2014lda,Rejzner:2016hdj}. Observe that, throughout the paper, we shall employ the notation $\mathcal{D}(M):=C^\infty_0(M)$ to denote the space of scalar test functions. 

\begin{definition}\label{Def: functional-valued vector distribution} 
	We call \textbf{functional-valued vector distribution} $u\in\mathcal{D}'(M; \mathsf{Fun}_W)$ a map
	\begin{align}
		u:\mathcal{D}(M)\times \Gamma(DM)\times&\Gamma(D^\ast M)\rightarrow W\,,\\
		&(f;\psi,\bar{\psi})\mapsto u(f;\psi,\bar{\psi})\nonumber\,,
	\end{align}
which is linear in $\mathcal{D}(M)$ and continuous in the locally convex topology of $\mathcal{D}(M)\times \Gamma(DM)\times\Gamma(D^\ast M)$.
\end{definition}

\begin{remark}
	With reference to Definition \ref{Def: functional-valued vector distribution} two comments are in due course. On the one hand we stress that the notion of vector-valued distribution that we adopt is mildly different from that of distributional sections of a vector bundle, see \textit{e.g.} \cite{BGP08}. On the other hand, we highlight that an alternative option consists of introducing Grassman variables to discuss spinor fields. This has been considered elsewhere, see \cite{Rej11}, but, in our work, it does not seem to offer any specific advantage and, therefore, we refrain from introducing an additional structure.
\end{remark}

On top of $\mathcal{D}'(M; \mathsf{Fun}_W)$ we introduce a notion of derivation with respect to smooth configurations. 
\begin{definition}\label{Def: functional derivatives}
Let $(k,k')\in\mathbb{N}$, we define the $(k,k')$-th functional derivative of $u\in\mathcal{D}'(M; \mathsf{Fun}_W)$ as the map
$$u^{(k,k^\prime)}:	\mathcal{D}(M)\otimes\underbrace{\Gamma(DM)\otimes\dots\otimes\Gamma(DM)}_k\otimes\underbrace{\Gamma(D^*M)\otimes\dots\Gamma(D^*M)}_{k^\prime}\times \Gamma(DM)\times\Gamma(D^\ast M)\rightarrow W,$$	
such that
	\begin{align*}
		&u^{(k,k')}(f\otimes\eta_1\otimes\ldots\otimes\eta_k\otimes\bar{\eta}_1\otimes\ldots\otimes\bar{\eta}_{k'};\eta,\bar{\eta})\\
		&:=\frac{\partial^{k+k'}}{\partial s_1\ldots\partial s_k\partial \bar{s}_1\ldots\partial\bar{s}_{k'}} u(f;\eta+s_1\eta_1+\ldots+s_k\eta_k,\bar{\eta}+\bar{s}_1\bar{\eta_1}+\ldots+\bar{s}_{k'}\bar{\eta}_{k'})\Big\vert_{s_1=\ldots=\bar{s}_{k'}=0},
	\end{align*}
for all $\eta,\eta_1,\ldots,\eta_k\in\Gamma(DM)$, $\bar{\eta},\bar{\eta}_1,\ldots,\bar{\eta}_{k'}\in\Gamma(D^\ast M)$, $f\in\mathcal{D}(M)$. For conciseness we shall denote it by $u^{(k,k')}\in\mathcal{D}'(M^{k+k'+1};\mathsf{Fun}_W)$. In addition, we define the \textbf{directional derivatives} along the directions $\varphi\in\Gamma(DM),\bar{\varphi}\in\Gamma(D^\ast M)$ as
\begin{align*}
	&\delta_\varphi:\mathcal{D}'(M;\mathsf{Fun}_W)\rightarrow\mathcal{D}'(M \times M;\mathsf{Fun}_W),\qquad[\delta_\varphi u](f;\eta,\bar{\eta}):=u^{(1,0)}(f\otimes\varphi;\eta,\bar{\eta}),\\
	&\delta_{\bar{\varphi}}:\mathcal{D}'(M;\mathsf{Fun}_W)\rightarrow\mathcal{D}'(M \times M;\mathsf{Fun}_W),\qquad[\delta_{\bar{\varphi}} u](f;\eta,\bar{\eta}):=u^{(0,1)}(f\otimes\bar{\varphi};\eta,\bar{\eta}).
\end{align*}
 A functional-valued distribution $u\in\mathcal{D}'(M; \mathsf{Fun}_W)$ is said to be \textbf{polynomial}, denoted by $u\in\mathcal{D}^\prime(M;\mathsf{Pol}_W)$, if there exists  $(n,n^\prime)\in\mathbb{N}_0\times\mathbb{N}_0$ such that $u^{(k,k^\prime)}=0$ in $\{k\geq n\}\cup\{k^\prime\geq n^\prime\}$.
\end{definition}


\begin{remark}\label{Rem: Components of a functional}
	Let $u\in\mathcal{D}'(M; \mathsf{Fun}_W)$ be as per Definition \ref{Def: functional-valued vector distribution} and let $\{e_\lambda\}_{\lambda=1,\dots, N_d}$ be an arbitrary basis of $\Sigma_d$ whose dual counterpart is denoted as $\{e^\lambda\}_{\lambda=1,\dots, N_d}$. For any $(f;\psi,\bar{\psi})\in\mathcal{D}(M)\times \Gamma(DM)\times\Gamma(D^\ast M)$ we can define the {\bf component} of $u$ with values in $W^{(p,q)}=\Sigma_d^{\otimes p}\otimes(\Sigma^\ast_d)^{\otimes q}$ as 
	\begin{equation}\label{Eq: component of a functional}
		u^{\lambda_1\dots\lambda_p}_{\rho_1\dots\rho_q}(f;\psi,\bar{\psi}):=\langle u(f;\psi,\bar{\psi}), e^{\lambda_1}\otimes\dots\otimes e^{\lambda_p}\otimes e_{\rho_1}\otimes\dots\otimes e_{\rho_q}\rangle,
	\end{equation}
where $\langle,\rangle$ denotes the extension to the tensor product of the duality pairing between $\Sigma_d$ and $\Sigma^*_d$. In the following we shall identify implicitly a functional with one of its components, {\it e.g.}, $u\equiv u^{\lambda_1\dots\lambda_p}_{\rho_1\dots\rho_q}$ if and only if for any $(f;\psi,\bar{\psi})\in\mathcal{D}(M)\times \Gamma(DM)\times\Gamma(D^\ast M)$ 
$$u(f;\psi,\bar{\psi})=u^{\lambda_1\dots\lambda_p}_{\rho_1\dots\rho_q}(f;\psi,\bar{\psi})e_{\lambda_1}\otimes\dots\otimes e_{\lambda_p}\otimes e^{\rho_1}\otimes\dots\otimes e^{\rho_q}.$$
More generally, making use of the natural embedding $W^{(p,q)}\hookrightarrow W$, $p,q\geq 0$, which will be therefore left implicit, for any $u\in\mathcal{D}^\prime(M;\mathsf{Pol}_W)$, we can decompose it as 
\begin{equation}\label{Eq: decomposing functional}
	u=\sum\limits_{p,q\geq 0}u^{\lambda_1\dots\lambda_p}_{\rho_1\dots\rho_q} e_{\lambda_1}\otimes\dots\otimes e_{\lambda_p}\otimes e^{\rho_1}\otimes\dots\otimes e^{\rho_q},
\end{equation}
where the right hand side is defined as per Equation \eqref{Eq: component of a functional}.
\end{remark}

\begin{example}\label{Ex: basic functionals}
	In the following we give some notable examples of polynomial functional-valued vector distributions which will constitute the building blocks of our construction. Bearing in mind Remark \ref{Rem: Components of a functional} and the notation introduced therein, we set
	\begin{align*}
		&\mathbf{1}(f;\psi,\bar{\psi})=\int_Md\mu_g(x)\,f(x)\,,\\
		&\hspace{2cm}\Phi^\lambda(f;\psi,\bar{\psi}):=\int_Md\mu_g(x)\,f(x)\psi^\lambda(x)\,,\\
		&\hspace{4cm}\bar{\Phi}_\lambda(f;\psi,\bar{\psi}):=\int_Md\mu_g(x)\,f(x)\bar{\psi}_\lambda(x)\,,\quad f\in\mathcal{D}(M),\psi\in\Gamma(DM),\bar{\psi}\in\Gamma(D^\ast M),
	\end{align*}
where $d\mu_g$ is the measure induced by the metric $g$.

Furthermore, following the same line of reasoning and using $\Phi^\lambda$ and $\bar{\Phi}_\lambda$ as building blocks, one can introduce functionals built out of single components with values in $(\Sigma_d^{\otimes k})\otimes(\Sigma^\ast_d)^{\otimes k^\prime}$ with either or both $k\geq 1$ and $k^\prime\geq 1$, such as
\begin{equation*}
	\bar{\Phi}_\alpha\Phi^{\beta\gamma}(f;\psi,\bar{\psi}):=\int_M d\mu_g(x)\,f(x)\bar{\psi}_\alpha(x)\psi^\beta(x)\psi^\gamma(x).
\end{equation*}
It is important to observe that all the above examples lie in $\mathcal{D}'(M;\mathsf{Pol}_W)$. As a matter of fact, a direct calculation yields for all $(f;\psi,\bar{\psi})\in\mathcal{D}(M)\times \Gamma(DM)\times\Gamma(D^\ast M)$
\begin{align*}
	&{\Phi^\lambda}^{(1,0)}(f\otimes\psi_1;\psi,\bar{\psi})=\int_Md\mu_g(x)f(x)\psi_1^\lambda(x),\quad \psi_1\in\Gamma(DM),\\
	&{\bar{\Phi}_{\lambda'}}^{(0,1)}(f\otimes\bar{\psi}_1;\psi,\bar{\psi})=\int_Md\mu_g(x)f(x)\bar{\psi}_{1\lambda'}(x),\quad \bar{\psi}_1\in\Gamma(D^\ast M),
\end{align*}
while all the other functional derivatives vanish. A similar result holds true for $\bar{\Phi}_\alpha\Phi^{\beta\gamma}$.
\end{example}

\begin{remark}	
Observe that the functionals $\Phi^\lambda$ and $\bar{\Phi}_{\lambda}$ in Example \ref{Ex: basic functionals} can be interpreted as encoding the information of the underlying field configuration. Yet, mirroring \cite{BDR23}, at this level we do not impose any constitutive relation connecting $\bar{\Phi}$ to $\Phi$, treating them as independent objects. Only at a later stage, when constructing the perturbative solutions to the stochastic Thirring model, we shall relate them by means of a suitable conjugation matrix $\beta\in Spin(d)$, see \cite{DHP09}. 
\end{remark}

For later convenience, we highlight that we can introduce a tensor product between functional-valued distributions, namely, for all $u,v\in\mathcal{D}'(M;\mathsf{Pol}_W)$, we denote by $u\otimes v$ the complex valued functional on $\mathcal{D}(M)\otimes\mathcal{D}(M)\times \Gamma(DM)\times\Gamma(D^\ast M)$ such that 
	\begin{equation}
		(u\otimes v)\,(f_1\otimes f_2;\psi,\bar{\psi})=u(f_1;\psi,\bar{\psi})\otimes_W v( f_2;\psi,\bar{\psi}),\quad\forall f_1,f_2\in\mathcal{D}(M),\psi\in\Gamma(DM),\bar{\psi}\in\Gamma(D^\ast M),
	\end{equation}
where the symbol $\otimes_W$ highlights that the tensor product is taken between elements lying in the finite dimensional vector space $W$. Furthermore we shall need a further generalization which is often referred to as \textbf{pointwise tensor product}: for al  $u,v\in\mathcal{D}'(M;\mathsf{Pol}_W)$ 
\begin{equation}\label{Eq: pointwise product}
	uv(f;\psi,\bar{\psi})=\Delta^*[u(\psi,\bar{\psi})\otimes v(\psi,\bar{\psi})](f),\qquad\forall f\in\mathcal{D}(M),\psi\in\Gamma(DM),\bar{\psi}\in\Gamma(D^\ast M),		
\end{equation}
where $\Delta^*$ denotes the pull-back along the diagonal map $\Delta:M\to M\times M$, $x\mapsto\Delta(x):=(x,x)$. Hence $\Delta^*[u(\psi,\bar{\psi})\otimes v(\psi,\bar{\psi})]:\mathcal{D}(M)\to W\otimes_W W$. We observe that Equation \eqref{Eq: pointwise product} is a priori ill-defined unless the distributions $u(\psi,\bar{\psi}), v(\psi,\bar{\psi})$ abide by suitable constraints. In the following we shall use techniques proper of microlocal analysis to restrict the class of functionals in such a way that this obstruction does not occur.

An additional, relevant step in the approach, that we advocate, consists of encoding at the algebraic level the information that the linear component of the underlying dynamics is ruled by the Dirac operator $\slashed{D}$ and by its adjoint as per Equation \eqref{Eq: Dirac Operator} and per Remark \ref{Rem: D and D*}. Denoting by $G_\psi\in\mathcal{D}^\prime\left(M\times M;\Sigma_d\otimes\Sigma^*_d\right)$ the fundamental solution of $\slashed{D}$ as per Equation \eqref{Eq: fundamental solution Dirac} and, working at the level of components, see Remark \ref{Rem: Components of a functional}, given $u\equiv u^{\lambda_1\dots\lambda_p}_{\rho_1\dots\rho_q}:\mathcal{D}(M)\times \Gamma(DM)\times\Gamma(D^\ast M)\to W^{p,q}$, we define 
\begin{gather}
(G_\psi\circledast u):\mathcal{D}(M)\times \Gamma(DM)\times\Gamma(D^\ast M)\to W^{p+1,q+1}\notag\\
(G_\psi\circledast u)(f;\psi,\bar{\psi})=u^{\lambda_1\dots\lambda_p}_{\rho_1\dots\rho_q}(G^{\lambda_{p+1}}_{\rho_{q+1}}\circledast f;\psi,\bar{\psi}) e_{\lambda_{p+1}}\otimes e_{\lambda_1}\otimes\dots\otimes e_{\lambda_p}\otimes e^{\rho_{q+1}}\otimes e^{\rho_1}\otimes\dots\otimes e^{\rho_q},\label{Eq: Convolution with a component}
\end{gather}
where we have written $G_\psi\equiv G^{\lambda^{p+1}}_{\rho_{q+1}} e^{\rho_{q+1}}\otimes e_{\lambda_{p+1}}$ and where $G^{\lambda^{p+1}}_{\rho_{q+1}}\in\mathcal{D}^\prime(M\times M)$. 
Here $G^{\lambda_{p+1}}_{\rho_{q+1}}\circledast f$ is such that, for all $h\in\mathcal{D}(M)$, $(G^{\lambda_{p+1}}_{\rho_{q+1}}\circledast f)(h):=G^{\lambda_{p+1}}_{\rho_{q+1}}(f\otimes h)$.

Combining Remark \ref{Rem: Components of a functional} and Equation \eqref{Eq: Convolution with a component}, we can extend per linearity the convolution of $G_\psi$ with a generic $u\in\mathcal{D}^\prime (M;\mathsf{Pol}_W)$ and we shall denote it by
\begin{equation}\label{Eq: functional convolution}
	(G_\psi\circledast u)(f;\eta,\bar{\eta}):=u(G_\psi\circledast f;\eta,\bar{\eta}).
\end{equation}
An analogous definition can be given for $G_{\bar{\psi}}$, the fundamental solution of $\slashed{D}^\ast$ acting on cospinors.

\begin{remark}\label{Rem: cut-off}
	If the underlying manifold $M$ is not compact, the convolution in \eqref{Eq: functional convolution} is not well-defined unless we introduce a cut-off function $\chi\in\ C^\infty_0(M)$. Focusing on Equation \eqref{Eq: Convolution with a component} and suppressing the subscript $\rho_{q+1}$ and the superscript $\lambda_{p+1}$, it holds that
	\begin{equation}
		G_{\psi}^\chi:=G_\psi\cdot(\chi\otimes\mathbf{1}),\qquad G_{\bar{\psi}}^\chi:=G_{\bar{\psi}}\cdot(\chi\otimes\mathbf{1}).
	\end{equation}
	Since $\chi$ does not affect the singular behaviour and the microlocal properties of the distribution involved, see \cite[Chap. 8]{Hormander-I-03}, for the sake of conciseness of the notation, we will continue to use the symbols $G_\psi$ and $G_{\bar{\psi}}$, leaving the superscript $\chi$ understood.
\end{remark}

We are in a position to define a distinguished algebra that encompasses all operations introduced. In what follows, given $F_1,\dots, F_n\in\mathcal{D}'(M;\mathsf{Pol}_W)$, $n\geq 1$, we denote by $\mathcal{E}[F_1,\ldots, F_n]$ the polynomial ring on $\mathcal{E}(M;W)$ generated by $F_1,\dots,F_n$ via the pointwise product of functional-valued vector distributions defined in Equation \eqref{Eq: pointwise product}. Adapting to the case in hand the construction of \cite{BDR23} we start from
\begin{equation}\label{Eq: A_0}
	\mathcal{A}_0^W:=\mathcal{E}[\mathbf{1},\Phi,\bar{\Phi}],
\end{equation}
and we encompass the information carried by the fundamental solutions of $\slashed{D}$ and $\slashed{D}^*$ via an inductive procedure, namely, for any $j\geq 1$ we set
\begin{align}\label{Eq: A_j}
	\mathcal{A}_j^W:=\mathcal{E}[\mathcal{A}_{j-1}^W\cup G_\psi\circledast\mathcal{A}_{j-1}^W\cup G_{\bar{\psi}}\circledast \mathcal{A}_{j-1}^W],
\end{align}
where
\begin{align*}
	& G_\psi\circledast\mathcal{A}_{j}^W:=\{u\in\mathcal{D}'(M; \mathsf{Pol}_W)\;\vert\; u=G_\psi\circledast v,\;\text{with }v\in\mathcal{A}_{j}^W\},\\
	& G_{\bar{\psi}}\circledast\mathcal{A}_{j}^W:=\{u\in\mathcal{D}'(M; \mathsf{Pol}_W)\;\vert\; u=G_{\bar{\psi}}\circledast v,\;\text{with }v\in\mathcal{A}_{j}^W\}.
\end{align*}
A direct inspection shows that the natural inclusion $\mathcal{A}_{j-1}^W\subset \mathcal{A}_{j}^W$ for all $ j\in\mathbb{N}$ entails that the following definition is meaningful.
\begin{definition}\label{Def: A}
	Let $\mathcal{A}_{j}^W$ be defined in Equations \eqref{Eq: A_0} and \eqref{Eq: A_j}. We denote with $\mathcal{A}^W$ the unital, commutative $\mathbb{C}$-algebra obtained as the direct limit
	\begin{equation}
		\mathcal{A}^W:=\lim_{\longrightarrow}\mathcal{A}_j^W.
	\end{equation}
\end{definition}

\begin{remark}\label{Rem: M}
	Due to its relevance in our construction, we point out that $\mathcal{A}^W$ is a graded algebra over $\mathcal{E}(M;W)$, where the grading accounts for the number of occurrences of $\Phi$ and $\bar{\Phi}$, see Example \ref{Ex: basic functionals}. In other words
	\begin{equation*}
		 \mathcal{A}^{W} = \bigoplus_{r,l,r',l' \in \mathbb{N}_0} \mathcal{M}_{r,r',l,l'}, \qquad\qquad \mathcal{A}^{W}_j = \bigoplus_{r,l,r',l' \in \mathbb{N}_0} \mathcal{M}^j_{r,r',l,l'},  \, \, \forall j \ge 0,
	\end{equation*}
where we denoted by $\mathcal{M}_{r,r',l,l'}$ the $\mathcal{E}(M;W)-$module generated by elements of the algebra with an overall polynomial degree of $r$ in $\psi^{\rho}$ and $r'$ in $\bar{\psi}_{\rho'}$ in which the fundamental solution $G_{\psi}$ acts $l$ times whilst $G_{\bar{\psi}}$ acts $l'$ times. Moreover, we have that $\mathcal{M}^j_{r,r',l,l'} := \mathcal{M}_{r,r',l,l'} \cap \mathcal{A}^{W}_j$. 
 For later convenience, let us introduce the following notation, where we ignore the occurrence of $G_{\psi}$ and $G_{\bar{\psi}}$ and focus exclusively on the overall polynomial degree in $\psi^{\rho}$ and $\bar{\psi}_{\rho'}$
\begin{equation*}
    \mathcal{M}_{r, r'} := \bigoplus_{\substack{l, l' \in \mathbb{N}_0 \\ p \le r \\ q \le r'}} \mathcal{M}_{p,q,l,l'}, \, \, \, \mathcal{M}_{r, r'}^j := \bigoplus_{\substack{l, l' \in \mathbb{N}_0 \\ p \le r \\ q \le r'}} \mathcal{M}_{p,q,l,l'}^j, \, \, \forall j \ge 0,
\end{equation*}
The direct limit is well-defined according to the natural inclusions 
\begin{equation*}
   \mathcal{M}_{r,r'} \subset \mathcal{M}_{r+1, r'}, \, \, \,  \mathcal{M}_{r,r'} \subset \mathcal{M}_{r, r'+1}, 
\end{equation*}
and one has that 
\begin{equation*}
    \mathcal{A}^W = \lim_{r, r' \rightarrow \infty} \mathcal{M}_{r, r'}.
\end{equation*}
\end{remark}

\subsubsection{Wavefront set}\label{Sec: wavefront set}
Before delving into the construction of an algebraic approach to singular SPDEs for spinor fields, it is mandatory to set suitable constraints to the singular structure of the derivatives of all functionals which enter the construction, using tools proper of microlocal analysis -- we refer to \cite[Ch. 8]{Hormander-I-03} for a detailed exposition of this topic. To set the nomenclature, observe that these derivatives give rise to multilocal distributions as per Definition \ref{Def: functional derivatives}.

\begin{remark}\label{Rem: vector WF}
	Since in our setting functional derivatives are vector-valued distributions, we recall that, given $u\in\mathcal{D}'(M,V)$, where $V$ is a finite dimensional vector space,
	\begin{equation}\label{Eq: vectorial wf}
		\mathrm{WF}(u):=\bigcup_{j=0}^m \mathrm{WF}(u_j),
	\end{equation}
	where $u_j$ denotes the $j$-th component of $u$ with respect to a chosen basis of $V$. Given the geometric nature of the wavefront set, this extension for distributions taking values in higher dimensional vector spaces turns out to be independent from the choice of a basis \cite{Kra00, SaVe01}. For completeness we remark that one might also consider a \textit{polarised wavefront set}, first introduced by \cite{Den82}, accounting for the possibility of different singular behaviors along the components of $u$. Yet, for our purposes this generalization brings no advantages and only additional technical hurdles. Therefore we shall not consider it further.
\end{remark} 
Since the algebra $\mathcal{A}^W$ encompasses the action of the fundamental solutions $G_\psi$ and $G_{\bar{\psi}}$ via Equation \eqref{Eq: convolution G}, it is natural to ask whether the pointwise product between functional derivatives of elements lying in $\mathcal{A}^W$ is well defined in a distributional sense. To this end a first step consists of characterizing the wavefront set of both $G_\psi$ and $G_{\bar{\psi}}$. Observing that both the Dirac operator $\slashed{D}$ and $\slashed{D}^\ast$ are elliptic, a standard computation using microlocal analysis techniques yields
\begin{equation}\label{Eq: WF Dirac}
	\mathrm{WF}(G_\psi)=\mathrm{WF}(G_{\bar{\psi}})=\mathrm{WF}(\delta_{\mathrm{Diag}_2})=\{(x,y;k_x,k_y)\in T^\ast (M\times M)\setminus\{0\}\;\vert\;x=y, k_x=-k_y\}.
\end{equation}
As it will become clear in the following, Equation \eqref{Eq: WF Dirac} suggests to consider only functional-valued vector distribution with a constraint on the wavefront set of their functional derivatives. To settle the notation, given a set of indices $I=\{1,\ldots,m\}$ consider its partition into disjoint non-empty subsets $I=I_1 \uplus\ldots\uplus I_p$. 

\noindent In what follows $\vert I_i\vert$ denotes the cardinality of the subset $I_i$ while $\mathrm{Diag}_{|I_i|}=\{(x,\ldots,x)\in M^{|I_i|}\;|\;x\in M\}$ is the diagonal of the Cartesian product $M^{|I_i|}$. 
The Dirac delta supported on the diagonal $\mathrm{Diag}_{|I_i|}$ acts as follows
\begin{equation*}
	\delta_{\mathrm{Diag}_{|I_i|}}(f):=\int_{M}d\mu_g(x)f(x,\ldots,x),\qquad\forall f\in\mathcal{D}(M^{|I_i|}).
\end{equation*}
Denoting $\widehat{x}_m=(x_1,\ldots, x_m)\in M^m$ as well as $\widehat{k}_m=(k_1,\ldots, k_m)\in T^\ast_{\widehat{x}_m}M^m$, it holds that
\begin{equation*}
	\mathrm{WF}(\delta_{\mathrm{Diag}_m})=\left\{(\widehat{x}_m;\widehat{k}_m)\in T^\ast M^m\setminus\{0\}\;|\;x_1=\ldots=x_m, \sum_{i=1}^mk_i=0\right\}.
\end{equation*}
\begin{definition}\label{Def: algebra WF}
	For any but fixed $m\in\mathbb{N}$ and $k\leq m$ let us consider the decomposition in disjoint partitions $\{1,\ldots,m\}=I_1\uplus \ldots \uplus I_k$ as well as the set
	\begin{align}\label{Eq: C_m}
		C_m:=\{(\widehat{x}_m,&\widehat{k}_m)\in T^\ast M^m\setminus\{0\}\;|\; \exists p\in\{1,\ldots, m-1\}, \nonumber\\
		&\{1,\ldots, m\}=I_1\uplus\ldots\uplus I_p \text{ such that } \forall i\neq j, \forall (a,b)\in I_i\times I_j\nonumber\\
		& \text{then }x_a\neq x_b \text{ and } \forall j\in\{1,\ldots, p\}, (\widehat{x}_{|I_j|},\widehat{k}_{|I_j|})\in \mathrm{WF}(\delta_{\mathrm{Diag}_{|I_j|}})\}.
	\end{align}
We define
\begin{equation}\label{Eq: constraint wf}
	\mathcal{D}'_C(M;\mathsf{Pol}_W):=\{u\in\mathcal{D}'(M;\mathsf{Pol}_W)\;|\;\forall k_1,k_2\in\mathbb{N}_0, \mathrm{WF}(u^{(k_1,k_2)})\subseteq C_{k_1+k_2+1}\}.
\end{equation}
\end{definition}
\begin{remark}
On account of Equation \eqref{Eq: C_m}, $C_1=\emptyset$.
Since $u\in\mathcal{D}'_C(M;\mathsf{Pol}_W)$ implies $\mathrm{WF}(u)\subseteq C_1=\emptyset$, it descends that the elements of $\mathcal{D}'_C(M;\mathsf{Pol}_W)$ are generated by smooth functions.
\end{remark}
\noindent In the following we show that on the class of functionals abiding by Definition \ref{Def: algebra WF} differentiation and convolution with fundamental solutions are well-defined operations. 
\begin{lemma}
	The space $\mathcal{D}'_C(M;\mathsf{Pol}_W)$ is stable under directional derivatives $\delta_\varphi$, $\delta_{\bar{\varphi}}$ for all $\varphi\in\Gamma(DM),\bar{\varphi}\in\Gamma(D^\ast M)$ , see Definition \ref{Def: functional derivatives}. Moreover, if $u\in\mathcal{D}'_C(M;\mathsf{Pol}_W)$, then $G_\psi\circledast u,G_{\bar\psi}\circledast u\in\mathcal{D}'_C(M;\mathsf{Pol}_W)$.
\end{lemma}
\begin{proof}
	Stability under directional derivatives immediately follows from the definition of $\mathcal{D}'_C(M;\mathsf{Pol}_W)$ and from Equation \eqref{Def: functional derivatives}. Let us focus on the action of $G_\psi$, the one for $G_{\bar{\psi}}$ following suit. Knowing that $\mathrm{WF}(G_\psi)=\mathrm{WF}(\delta_{\mathrm{Diag}_2})$ and that any $u\in\mathcal{D}'_C(M;\mathsf{Pol}_W)$ is generated by a smooth function, it descends that $G_\psi\circledast u\in\mathcal{D}'(M;\mathsf{Pol}_W)$. As a matter of facts, it holds that
	\begin{align*}
		(G_\psi\circledast u)^{(k_1,k_2)}(f\otimes &\eta^{\otimes k_1}\otimes\bar{\eta}^{\otimes k_2};\psi,\bar{\psi})=u^{(k_1,k_2)}((G_\psi\ast f)\otimes\eta^{\otimes k_1}\otimes\bar{\eta}^{\otimes k_2};\psi,\bar{\psi})\\
		&=\left[(G_\psi\otimes \delta_{\mathrm{Diag}_2}^{\otimes k_1+k_2})\circledast u^{(k_1,k_2)}\right](f\otimes\eta^{\otimes k_1}\otimes\bar{\eta}^{\otimes k_2};\psi,\bar{\psi}).
	\end{align*}
Invoking \cite[Eq. (A.1) \& (A.16)]{Dappiaggi:2020gge} and item 5. of \cite[Thm. A.4]{Dappiaggi:2020gge} 
\begin{equation*}
	\mathrm{WF}(G_\psi\otimes \delta_{\mathrm{Diag}_2}^{\otimes k_1+k_2})\subseteq\{(x,x,\widehat{x}_{k_1+k_2},\widehat{x}_{k_1+k_2};k,-k,\widehat{k}_{k_1+k_2},-\widehat{k}_{k_1+k_2})\in T^\ast M^{2+2(k_1+k_2)}\setminus\{0\}\}.
\end{equation*}
The projections of $\mathrm{WF}(G_\psi\otimes \delta_{\mathrm{Diag}_2}^{\otimes k_1+k_2})$ defined in item 5. of \cite[Thm. A.4]{Dappiaggi:2020gge} are empty, so that the sufficient conditions for $(G_\psi\otimes \delta_{\mathrm{Diag}_2}^{\otimes k_1+k_2})\circledast u^{(k_1,k_2)}$ to be well-defined are abode by. Furthermore we can give an estimate of the wavefront set
\begin{align*}
	\mathrm{WF}((G_\psi\otimes \delta_{\mathrm{Diag}_2}^{\otimes k_1+k_2})&\circledast u^{(k_1,k_2)})\subseteq\{(x,\widehat{x}_{k_1+k_2};k_1,\widehat{k}_{k_1+k_2})\in T^\ast M^{1+k_1+k_2}\;|\;\\
	& \exists\,(y,\widehat{y}_{k_1+k_2};p_1,\widehat{p}_{k_1+k_2})\in \mathrm{WF}(u^{(k_1,k_2)}),\\
	&(x,y,\widehat{x}_{k_1+k_2},\widehat{y}_{k_1+k_2};k,p,\widehat{k}_{k_1+k_2},\widehat{p}_{k_1+k_2})\in \mathrm{WF}(G_\psi\otimes \delta_{\mathrm{Diag}_2}^{\otimes k_1+k_2})\}\subseteq C_{k_1+k_2+1},
\end{align*}
which allows to conclude that $G_\psi\circledast u\in\mathcal{D}'_C(M;\mathsf{Pol}_W)$.
\end{proof}
\begin{theorem}\label{Thm: inclusion A}
	Denoting with $\mathcal{A}^W$ the unital, associative algebra as per Definition \ref{Def: A}, it holds that
	\begin{equation*}
		\mathcal{A}^W\subset \mathcal{D}'_C(M; \mathsf{Pol}_W)
	\end{equation*}
\end{theorem}
The proof of Theorem \ref{Thm: inclusion A} follows the same lines of \cite[Prop. 2.11]{BDR23}, thus we omit it. We highlight that, although functional derivatives of elements of $\mathcal{A}^W$ are distributions taking values in $W$, \emph{cf.} Equation 	\eqref{Eq: W space}, their wavefront set as per Equation \eqref{Eq: vectorial wf} satisfies Equation \eqref{Eq: constraint wf}.

\subsubsection{Scaling degree}\label{Sec: scaling degree}

A fundamental tool adopted in the renormalization of singular distributions is the \textbf{Steinmann scaling degree} \cite{Ste71,Brunetti-Fredenhagen-00, BaWr14}. We recall succinctly its definition. Given $u\in\mathcal{D}'(M)$, $f\in\mathcal{D}(M)$, given $\lambda>0$, we define $f_{x_0}^\lambda(x):=\lambda^{-d}f(\lambda^{-1}(x-x_0))$ and $u_{x_0}^\lambda$ to be the element of $\mathcal{D}^\prime(M)$ such that, for all $f\in\mathcal{D}(M)$, $u_{x_0}^\lambda(f)=u(f_{x_0}^\lambda)$. Consequently we call {\em scaling degree} of $u$ at $x_0$
\begin{equation}\label{Eq: scaling degree}
	\textrm{\textrm{sd}}_{x_0}(u)=\inf\{\omega\in\mathbb{R}\;\vert\;\lim_{\lambda\rightarrow 0^+}\lambda^\omega u^\lambda_{x_0}=0\}.
\end{equation}
Instead of delving into a detailed exposition of the properties of the scaling degree, we refer the reader to \cite[App. B]{Dappiaggi:2020gge}. 

In this work we deal with distributions $u\in\mathcal{D}^\prime(M)$, whose singular support is not an isolated point, rather a submanifold $N\subset M$, the prototypical case being $N=\mathrm{Diag}_2(M)\doteq\{(x,x)\;|\;x\in M\}\subset M\times M$. In this case the notion of scaling degree is the natural extension of Equation \eqref{Eq: scaling degree} and we shall denote it by $\textrm{sd}_{N} (u)$ -- refer to \cite{Brunetti-Fredenhagen-00} for details. As for the wavefront set of vector-valued distributions, see Remark \ref{Rem: vector WF}, we extend the notion scaling degree for vector-valued distributions, considering the worst case scenario, namely, given a submanifold $N\subseteq M$, an $m$-dimensional vector space $V$ and $u\in\mathcal{D'}(M,V)$, we set
\begin{equation}\label{Eq: vector scaling degree}
	\textrm{sd}^V_{N}(u):=\max_{i\in\{1,\ldots, m\}}\textrm{sd}_{N}(u_i)\,.
\end{equation}
Here $u_i$ is the $i$-th component of $u$ with respect to an arbitrary basis, the choice of which does not alter Equation \eqref{Eq: vector scaling degree}. Observe, that Equation \eqref{Eq: vector scaling degree} entails that, mutatis mutandis, all results outlined in \cite[App. B]{DDR20} remain valid. For later convenience, we introduce the \textit{degree of divergence} of a distribution $u\in\mathcal{D}'(M,V)$ with respect to a submanifold $N\subseteq M$
\begin{equation}\label{Eq: degree of divergence}
	\rho_N^V(u):=\textrm{sd}^V_{N}(u)-\mathrm{codim}(N),
\end{equation}
where $\mathrm{codim}(N)$ denotes the codimension of the smooth manifold. In the following we shall drop the superscript $V$ unless there is a chance of confusion for the reader. It is instructive to apply these notions to $G_\psi$, the fundamental solution of the massless Dirac operator. Denoting by $(G_\psi)^{\rho^\prime}_\rho$ its components, $\rho,\rho^\prime=0,1,\dots,d-1$, and using Equation \eqref{Eq: Dirac to Laplacian} in combination with the local expression for the fundamental solution of the operator $-\Delta+\frac{R}{4}+m^2$, it turns out that, working at the level of integral kernel,
\begin{equation}
	\lim_{\lambda\rightarrow 0^+}\lambda^\omega(G_\psi)^\rho_{\rho^\prime}(\lambda x,\lambda y)=0,
\end{equation}
only if $\omega>d-1$, which entails
\begin{equation}\label{Eq: scaling degree fundamental solutions}
	\textrm{sd}_{\mathrm{Diag}_2}(G_\psi)^\rho_{\rho^\prime}=d-1.
\end{equation}
An identical result holds true for all the components of $G_{\psi}$ and of $G_{\bar{\psi}}$. 

A key result that we will exploit when dealing with renormalization of a priori ill-defined analytical structures pertains the scaling degree of functionals of the form $G_\psi\circledast u$ and $G_{\bar{\psi}}\circledast u$, $u\in\mathcal{D}^\prime_C(M; \mathsf{Pol}_W)$ as defined in Equation \eqref{Eq: functional convolution}. Adapting almost slavishly the proof in \cite[Lem. 2.14, Lem. B.15]{DDR20} to the case in hand, the following result holds true.
\begin{lemma}\label{Lem: stability}
	Let $G_\psi,G_{\bar{\psi}}$ be the fundamental solutions introduced in Section \ref{Sec: geometric setting}. 
  For any $u\in\mathcal{D}_C'(M;\mathsf{Pol}_W)$ and denoting with $\circledast$ the operation as per Equation \eqref{Eq: functional convolution},
	\begin{equation*}
		\textrm{sd}_{\mathrm{Diag}_{1+k_1+k_2}}(G_\psi\circledast u)^{(k_1,k_2)},\quad	\textrm{sd}_{\mathrm{Diag}_{1+k_1+k_2}}(G_{\bar{\psi}}\circledast u)^{(k_1,k_2)}<\infty\,,
	\end{equation*}
	whenever $u^{(k_1,k_2)}$ has finite scaling degree with respect to
	\begin{equation*}
		\mathrm{Diag}_{k_1+k_2+1}=\{(x_1,\ldots,x_{k_1+k_2+1})\in M^{k_1+k_2+1}\;\vert\;x_1=\ldots = x_{k_1+k_2+1}\}.
	\end{equation*} 
\end{lemma}

\section{The Local Algebra $\mathcal{A}_{\cdot Q}^W$}\label{Sec: local deformation}
In order to encode in the algebra $\mathcal{A}^W$ as per Definition \ref{Def: A} the information concerning an underlying random field, such as the one which rules Equation \eqref{Dirac}, we deform the underlying algebraic structures, following a procedure inspired by the \emph{deformation quantization} approach to field theories \cite{Fredenhagen:2014lda, HaRe20, Rejzner:2016hdj}. In order to better grasp the underlying rationale, let us consider the solutions to the linear stochastic equations
\begin{equation}
	\label{Eq: linear Dirac}
	\begin{cases}
		\slashed{D}^{\rho}_{\sigma} \widehat{\varphi}^{\sigma}= \xi^{\rho}, \\
		\slashed{D}^{*\, \sigma}_{\rho} \widehat{\bar{\varphi}}_{\sigma} = \bar{\xi}_{\rho}, \\
	\end{cases}
\end{equation}
where the noise abides by Definition \ref{Def: fermionic white noise}. Employing the fundamental solutions of the operator $D$ and $D^*$, it descends that
\begin{align}\label{Eq: random fields}
	&\widehat{\varphi}^\lambda:=(G_\psi)^\lambda_{\lambda'}\circledast\xi^{\lambda'},\\
	&\widehat{\bar{\varphi}}_\lambda:=(G_{\bar{\psi}})^{\lambda'}_\lambda\circledast\bar{\xi}_{\lambda'}.
\end{align}
Using once more Definition \ref{Def: fermionic white noise}, for all $f_1,f_2\in\mathcal{D}(M)$ it holds
\begin{align*}
	\mathbb{E}[\widehat{\varphi}^\lambda&(f_1)\widehat{\varphi}^{\lambda'}(f_2)]\\
	&=\mathbb{E}\left[\int_{M^4}d\mu_g(x)d\mu_g(y)d\mu_g(x')d\mu_g(y')\,(G_\psi)^\lambda_{\lambda''}(x,x')\xi^{\lambda''}(x')(G_\psi)^{\lambda'}_{\lambda'''}(y,y')\xi^{\lambda'''}(y')f_1(x)f_2(y)\right]\\
	&=\int_{M^4}d\mu_g(x)d\mu_g(y)d\mu_g(x')d\mu_g(y')\,(G_\psi)^\lambda_{\lambda''}(x,x')(G_\psi)^{\lambda'}_{\lambda'''}(y,y')\underbrace{\mathbb{E}[\xi^{\lambda''}(x')\xi^{\lambda'''}(y')]}_{0}f_1(x)f_2(y)=0.
\end{align*}
Observe that, since we are considering random valued distributions, it is legit to exchange the operations of smearing against test-functions and of taking the expectation value. An analogous calculation yields
\begin{equation*}
	\mathbb{E}[\widehat{\bar{\varphi}}_\lambda(f_1)\widehat{\bar{\varphi}}_{\lambda'}(f_2)]=0\,.
\end{equation*}
On the contrary, it holds that
\begin{align}\label{Eq: 2-point function 1}
	\mathbb{E}[\widehat{\varphi}^\lambda(f_1)\,&\widehat{\bar{\varphi}}_{\lambda'}(f_2)]\nonumber\\
	&=\int_{M^4}d\mu_g(x)d\mu_g(y)d\mu_g(x')d\mu_g(y')\,(G_\psi)^\lambda_{\lambda''}(x,x')(G_{\bar{\psi}})_{\lambda'}^{\lambda'''}(y,y')\underbrace{\mathbb{E}[\xi^{\lambda''}(x')\bar{\xi}_{\lambda'''}(y')]}_{\delta^{\lambda''}_{\lambda'''}\delta(x'-y')}f_1(x)f_2(y)\nonumber\\
	&=\int_{M^3}d\mu_g(x)d\mu_g(y)d\mu_g(x')\,(G_\psi)^\lambda_{\lambda''}(x,x')(G_{\bar{\psi}})_{\lambda'}^{\lambda''}(y,x')f_1(x)f_2(y)=(Q_{\psi\bar{\psi}})^\lambda_{\lambda'}(f_1\otimes f_2),
\end{align}
where, working at the level of integral kernels, we defined 
\begin{equation}\label{Eq: Q}
	(Q_{\psi\bar{\psi}})^\lambda_{\lambda'}(x,y):=\int_M d\mu_g(x')\,(G_\psi)^\lambda_{\lambda''}(x,x')(G_{\bar{\psi}})_{\lambda'}^{\lambda''}(y,x').
\end{equation}
An identical procedure yields
\begin{equation}\label{Eq: 2-point function 2}
	\mathbb{E}[\widehat{\bar{\varphi}}_{\lambda}(f_1)\,\widehat{\varphi}^{\lambda'}(f_2)]=(Q_{\bar{\psi}\psi})^{\lambda'}_\lambda(f_1\otimes f_2),
\end{equation}
whose integral kernel is 
\begin{equation}\label{Eq: tildeQ}
	(Q_{\bar{\psi}\psi})^{\lambda'}_\lambda(x,y):=-\int_M d\mu_g(x')\,(G_{\bar{\psi}})^{\lambda''}_\lambda(x,x')(G_\psi)^{\lambda'}_{\lambda''}(y,x').
\end{equation}

\begin{remark}\label{Rem: All over Q}
	One can interpret $\widehat{\varphi}$ and $\widehat{\bar{\varphi}}$ as being the component of a centered vector-valued Gaussian distribution
	 \begin{equation*}
		\underline{\widehat{\varphi}} := \begin{pmatrix}
			\widehat{\varphi} \\
			\widehat{\bar{\varphi}}
		\end{pmatrix},
	\end{equation*}
whose covariance matrix $C$ is such that $C^{11}=C^{22}=0$ while $C^{12}=Q_{\psi\bar{\psi}}$ and $C^{21}=Q_{\bar{\psi}\psi}$ as per Equations \eqref{Eq: 2-point function 1} and \eqref{Eq: 2-point function 2} -- see Definitions \ref{cov} and \ref{Def: fermionic white noise}. For notational simplicity, henceforth we adopt the convention
	\begin{equation}\label{Eq: Q-notation}
		Q_\lambda^{\lambda'}=(Q_{\psi\bar{\psi}})_\lambda^{\lambda'},\qquad \widetilde{Q}_\lambda^{\lambda'}=(Q_{\bar{\psi}\psi})_\lambda^{\lambda'}, 
	\end{equation} 
where we drop the subscript representing the ordering of the fields involved since this should not constitute a source of confusion.  In addition, using \cite[Thm. 8.2.14]{Hormander-I-03} and Equation \eqref{Eq: WF Dirac} one can infer that
	\begin{align}\label{Eq: WF Q}
		&\mathrm{WF}(Q), \mathrm{WF}(\widetilde{Q})\subseteq \mathrm{WF}(G_\psi)\equiv \mathrm{WF}(G_{\bar{\psi}})\equiv \mathrm{WF}(\delta_{\mathrm{Diag}_2}).
	\end{align}
\end{remark}

We are now ready to introduce the sought deformation codifying the stochastic correlations of the vector-valued Gaussian white noise. 

\begin{definition}[Deformed product]\label{Def: deformed product}
	Given the algebra $\mathcal{A}^W$ as per Definition \ref{Def: A}, for all $u_1,u_2\in\mathcal{A}^W$, $f\in\mathcal{D}(M)$ and for all field configurations $(\psi,\bar{\psi})\in\Gamma(DM\oplus D^\ast M)$, we define
	\begin{equation}\label{Eq: deformed product}
		u_1\cdot_{\mathsf{Q}} u_2(f;\psi,\bar{\psi}):=\sum_{k\geq 0}\sum_{k_1+k_2=k}\frac{1}{k_1!k_2!}(\delta_{\mathrm{Diag}_2}\otimes Q^{\otimes k_1}\otimes \widetilde{Q}^{\otimes k_2})\cdot [u_1^{(k_1,k_2)}\tilde{\otimes} u_2^{(k_2,k_1)}](f\otimes \mathbf{1}_{1+2k};\psi,\bar{\psi}),
	\end{equation}
where $\delta_{\mathrm{Diag}_2}$ is the Dirac delta distribution supported at the diagonal of $M\times M$, while $Q$ and $\widetilde{Q}$ are the integral kernels introduced in Equations \eqref{Eq: Q} and \eqref{Eq: tildeQ}, respectively. Furthermore, $\tilde{\otimes}$ denotes a modified tensor product
characterized, at the level of integral kernel, by the following expression
\begin{align}\label{Eq: expression deformed product}
	(\delta_{\mathrm{Diag}_2}\otimes &Q^{\otimes k_1}\otimes \widetilde{Q}^{\otimes k_2})\cdot [u_1^{(k_1,k_2)}\tilde{\otimes} u_2^{(k_2,k_1)}]\nonumber\\
	&:=\delta(x_1,x_2)\prod_{j=1}^{k_1}\prod_{l=1}^{k_2}Q(z_j,z_j')\widetilde{Q}(y_l,y_l')u_1^{(k_1,k_2)}(x_1,\widehat{z}_{k_1},\widehat{y}_{k_2})u_2^{(k_2,k_1)}(x_2,\widehat{z'}_{k_1},\widehat{y'}_{k_2}),
\end{align}
being $\widehat{z}_{k_1}:=(z_1,\ldots, z_{k_1})$, $\widehat{y}_{k_2}:=(y_1,\ldots, y_{k_2})$ and analogously for primed variables.
\end{definition}

\begin{remark}
	Since $u_1,u_2\in\mathcal{A}^W\subset \mathcal{D}'_C(M;\mathsf{Pol}_W)$ are polynomial functionals, see Definition \ref{Def: functional derivatives}, the sum in Equation \eqref{Eq: deformed product} consists only of a finite number of terms. Let us stress that Equation \eqref{Eq: deformed product} is a mere formal expression, since it involves the vector-valued integral kernels $Q$ and $\widetilde{Q}$, which are ill-defined on the diagonal of $M\times M$. Such hurdle will be bypassed using a renormalization procedure.
\end{remark}

\begin{example}\label{Ex: correlation function}
As an application of Equation \eqref{Eq: deformed product}, let us consider $u_1^\lambda=\Phi^\lambda$ and $(u_2)_\lambda=\bar{\Phi}_\lambda$, see Remark \ref{Rem: Components of a functional} and Example \ref{Ex: basic functionals}. It holds that
\begin{align*}
	&\Phi^\lambda\cdot_{\mathsf{Q}}\bar{\Phi}_{\lambda'}(f;\psi,\bar{\psi})=\Phi^\lambda\bar{\Phi}_{\lambda'}(f;\psi,\bar{\psi})+Q^\lambda_{\lambda'}(f\delta_{\mathrm{Diag}_2}).
\end{align*}
Recalling that, in our formalism, taking the expectation value amounts to evaluating the underlying functional at $\psi=\bar{\psi}=0$, one recovers \underline{formally} the correlations as in Equation \eqref{Eq: 2-point function 1} with $f_1\otimes f_2$ replaced by $f\delta_{\mathrm{Diag}_2}$, namely
$$\Phi^\lambda\cdot_{\mathsf{Q}}\bar{\Phi}_{\lambda'}(f;0,0)=Q^\lambda_{\lambda'}(f\delta_{\mathrm{Diag}_2}).$$
As one can infer per direct inspection the right hand side is an ill-defined quantity, which needs to be tamed.
\end{example}

The following result establishes the existence of a deformation map which will play a prominent r\^ole in the definition of an algebra of  functionals encoding the information of the underlying random field.
\begin{theorem}\label{Prop: deformation map}
	Let us consider the algebra $\mathcal{A}^W$ introduced in Definition \ref{Def: A} and let $\mathcal{M}_{r,r'}$ be the $\mathcal{E}(M;W)$-moduli defined in Remark \ref{Rem: M}. Bearing in mind Definition \ref{Def: algebra WF}, there exists a linear map $\Gamma_{\cdot_{\mathsf{Q
}}}^W:\mathcal{A}^W\rightarrow\mathcal{D}'_C(M;\mathsf{Pol}_W)$ such that
	\begin{enumerate}\label{Eq: identity}
		\item for all $u:=u^\lambda\in\mathcal{M}_{1,0}$ or $u:=u_{\lambda'}\in\mathcal{M}_{0,1}$,
		\begin{equation}\label{Eq: identity}
			\Gamma_{\cdot_{\mathsf{Q}}}^W(u)=u,
		\end{equation}
	\item for all $u\in\mathcal{A}^W$, it holds that
	\begin{equation}\label{Eq: convolution G}
	\begin{cases}
		\Gamma_{\cdot_{\mathsf{Q}}}^W(G_\psi\circledast u)=G_\psi\circledast \Gamma_{\cdot_{\mathsf{Q}}}^W(u) \\
		\Gamma_{\cdot_{\mathsf{Q}}}^W(G_{\bar{\psi}}\circledast u)=G_{\bar{\psi}}\circledast \Gamma_{\cdot_{\mathsf{Q}}}^W(u)
	\end{cases}\,,
	\end{equation}
	\item for all $u\in\mathcal{A}^W$ and for every $\psi\in\Gamma(DM)$, $\bar{\psi}\in\Gamma(D^\ast M)$
	\begin{equation}\label{Eq: deformation derivatives}
		\begin{cases}
			\Gamma_{\cdot_{\mathsf{Q}}}^W\circ \delta_\psi=\delta_\psi\circ\Gamma_{\cdot_{\mathsf{Q}}}^W \\
			\Gamma_{\cdot_{\mathsf{Q}}}^W\circ\delta_{\bar{\psi}}=\delta_{\bar{\psi}}\circ \Gamma^W_{\cdot_{\mathsf{Q}}}
		\end{cases}\,,
	\end{equation}
	where we denoted by $\delta_\psi$ (resp. $\delta_{\bar{\psi}}$) the directional derivatives along $\psi$ (resp. $\bar{\psi}$) as per Definition \ref{Def: functional derivatives}. Furthermore, for every $\varphi\in C^\infty(M;W)$, it holds that
	\begin{equation}\label{Eq: product smooth function}
		\Gamma_{\cdot_{\mathsf{Q}}}^W(\varphi u)=\varphi\Gamma_{\cdot_{\mathsf{Q}}}^W(u),
	\end{equation}
	\item denoting by $\sigma_{(k,k')}(u):=\textrm{sd}_{\mathrm{Diag}_{k+k'+1}}(u^{(k,k')})$ the scaling degree with respect to the thin diagonal of $M^{k+k'+1}$ as per Equation \eqref{Eq: vector scaling degree}, it holds that
	\begin{equation}\label{Eq: finiteness scaling degree}
		\sigma_{(k,k')}(\Gamma_{\cdot_{\mathsf{Q}}}^W(u))<\infty,\quad \forall u\in\mathcal{A}^W.
	\end{equation}
	\end{enumerate}
\end{theorem}
\begin{proof}	
We discuss only the key steps of the proof, since the line of reasoning that we adopt mirrors that illustrated in \cite{BDR23} for the scalar field.

\noindent\underline{\emph{Strategy}}. Since it has been established in Remark \ref{Rem: M} that $\mathcal{A}^W = \lim_{r, r' \rightarrow \infty} \mathcal{M}_{r, r'}$
	the proof can proceed by induction with respect to the indices $r$ and $r'$. We define formally 
		\begin{equation}
			\label{eq1pr}
			\Gamma^W_{\cdot_\mathsf{Q}} (u_1 \cdot ... \cdot u_n) := \Gamma^W_{\cdot_\mathsf{Q}} (u_1) \cdot_{\mathsf{Q}} ... \cdot_{\mathsf{Q}} \Gamma^W_{\cdot_\mathsf{Q}}(u_n),
		\end{equation}
		for all $u_1, ..., u_n \in \mathcal{A}^W$, where $\cdot_{\mathsf{Q}}$ denotes the \emph{a priori} ill-defined product characterized in Equation \eqref{Eq: deformed product}.\\
		
		\noindent \underline{\emph{The Case $d=1$}}. In the massive one-dimensional scenario with $M \equiv \mathbb{R}$, the fundamental solution of the Dirac operator $\slashed{D}$ and that of its formal adjoint $\slashed{D}^*$ are 
		\begin{equation*}
			\begin{cases}
				G_{\psi}(x) = -i e^{-imx} \Theta(-x+m), \\
				G_{\bar{\psi}}(x) = i e^{imx} \Theta (x+m),
			\end{cases}
		\end{equation*}
		As a consequence, a direct computation of the the correlation functions $Q \delta_{\mathrm{Diag}_2} := Q_{\psi \bar{\psi}} \delta_{\mathrm{Diag}_2}$ and $\widetilde{Q} \delta_{\mathrm{Diag}_2} := Q_{\bar{\psi} \psi} \delta_{\mathrm{Diag}_2}$ at the coinciding point limit yields 
		\begin{flalign*}
			&Q (f\delta_{\mathrm{Diag}_2}) := \int_{\mathbb{R}} dx \,  G_{\psi}(x) G_{\bar{\psi}}(x) f(x) = \int_{\mathbb{R}} dx \, \Theta(-x+m) \Theta (x+m) f(x) = \\ &= \int_{\mathbb{R}} dx \, \chi_{[-m,m]} f(x) = \int_{-m}^m dx \, f(x), \, \, \, \forall f \in \mathcal{D}(\mathbb{R})
		\end{flalign*}
		where $\chi_{[-m,m]}$ is the characteristic function of the interval $[-m,m]$. An analogous result can be derived for $\widetilde{Q}(f\delta_{\mathrm{Diag}_2})$. Hence, the composition $G_{\psi} \circ G_{\bar{\psi}}$ lies in $L^1_{loc}(\mathbb{R})$ for any $m > 0$. Therefore, there is no need to invoke a renormalization procedure and the proof of the theorem follows immediately taking into account Equation \eqref{eq1pr}.\\
		
		\noindent \underline{\emph{The Case $d \ge 2$: Step 1.}} The first step of the proof consists in showing that, under the hypotheses that the action of the map $\Gamma^W_{\cdot_\mathsf{Q}}$ is well-defined for any $u \in \mathcal{A}^W$ and that Equation \eqref{Eq: convolution G} is satisfied, all the requirements in the statement of the theorem hold true also when considering algebra elements of the form $G_{\psi} \circledast u$ and $G_{\bar{\psi}} \circledast u$. Specifically, we shall focus only on $G_{\psi} \circledast u$, the reasoning translating slavishly to the case $G_{\bar{\psi}} \circledast u$ -- see Equation \eqref{Eq: functional convolution}. Yet, $\Gamma^W_{\cdot_\mathsf{Q}}(G_{\psi} \circledast u)$ is completely specified by means of Equation \eqref{Eq: convolution G}, while Equation \eqref{Eq: finiteness scaling degree} is a direct consequence of Lemma \ref{Lem: stability}. Thus, we are left to prove the validity of the relations appearing in item 3. of the theorem. Observe that, for all $f \in \mathcal{D}(M)$, $(\eta, \bar{\eta}) \in \Gamma(DM\oplus D^\ast M)$ and for any $(\zeta, \bar{\zeta}) \in \Gamma(DM\oplus D^\ast M)$, it holds that
		\begin{flalign*}
			\delta_{\zeta} \circ \Gamma^W_{\cdot_{\mathsf{Q}}} (G_{\psi} \circledast u) (f; \eta, \bar{\eta}) &= \delta_{\zeta} \circ \Gamma^W_{\cdot_{\mathsf{Q}}} (u) (G_{\psi} \circledast f; \eta, \bar{\eta}) = \Gamma^W_{\cdot_{\mathsf{Q}}} (\delta_{\zeta} u) (G_{\psi} \circledast f; \eta, \bar{\eta}) \\ 
			&= \Gamma^W_{\cdot_{\mathsf{Q}}} (G_{\psi} \circledast  \delta_{\zeta} u) (f; \eta, \bar{\eta}) = \Gamma^W_{\cdot_{\mathsf{Q}}} (\delta_{\zeta} G_{\psi} \circledast  u) (f; \eta, \bar{\eta}) \\ 
			&= \Gamma^W_{\cdot_{\mathsf{Q}}} \circ \delta_{\zeta} (G_{\psi} \circledast  u) (f; \eta, \bar{\eta}), 
		\end{flalign*}
	A similar result can be devised for $ \delta_{\bar{\zeta}} \circ \Gamma^W_{\cdot_{\mathsf{Q}}}$. 
		\noindent For what concerns Equation \eqref{Eq: product smooth function}, its validity for elements of the form $G_{\psi} \circledast u$ descends from 
		\begin{flalign*}
			\varphi \Gamma^W_{\cdot_{\mathsf{Q}}} (G_{\psi} \circledast u) (f; \eta, \bar{\eta}) &= \varphi \Gamma^W_{\cdot_{\mathsf{Q}}} (u) (G_{\psi} \circledast f; \eta, \bar{\eta}) = \Gamma^W_{\cdot_{\mathsf{Q}}} (\varphi u) (G_{\psi} \circledast f; \eta, \bar{\eta}) \\
			&= \Gamma^W_{\cdot_{\mathsf{Q}}} (\varphi G_{\psi} \circledast u) (f; \eta, \bar{\eta}),
		\end{flalign*}
		After these preliminary observations, we are ready to discuss the inductive argument.\\
		\noindent \underline{\emph{Step 2: $(r, r') = (1,1)$}}. First of all, let us observe that Equation \eqref{Eq: identity} completely characterizes the restriction of the map $\Gamma^W_{\cdot_\mathsf{Q}}$ to the moduli $\mathcal{M}_{1,0}$ and $\mathcal{M}_{0,1}$ introduced in Remark \ref{Rem: M} and that all other properties hold true by direct inspection. \\
		\noindent Therefore, we shall now devote our attention to proving the validity of the aforementioned properties for $\mathcal{M}_{1,1}$. To this avail, we will define the action of the map $\Gamma^W_{\cdot_{\mathsf{Q}}}$ inductively with respect to the index $j$ appearing in the decomposition
		\begin{equation*}
			\mathcal{M}_{1,1} = \bigoplus_{j \in \mathbb{N}_0} \mathcal{M}^j_{1,1}, \, \, \text{with} \, \, \mathcal{M}^j_{1,1} := \mathcal{M}_{1,1} \cap \mathcal{A}^W_j, \, \forall j \in \mathbb{N}_0,
		\end{equation*}
		as discussed in Remark \ref{Rem: M}. Considering $j=0$, one has that 
		\begin{equation*}
			\mathcal{M}^0_{1,1} := \mathcal{M}_{1,1} \cap \mathcal{A}^W_0 = \text{span}_{\mathcal{E}({M; W})} \{1, \Phi^{\rho}, \bar{\Phi}_{\rho'}, \Phi^{\rho} \bar{\Phi}_{\rho'}, \bar{\Phi}_{\rho'} \Phi^{\rho}\},
		\end{equation*}
		where $W$ is defined in Equation \eqref{Eq: W space}.
		Thus, the remaining task consists in defining the action of $\Gamma^W_{\cdot_{\mathsf{Q}}}$ on the generators $\Phi^{\rho} \bar{\Phi}_{\rho'}$ and $\bar{\Phi}_{\rho'} \Phi^{\rho}$. Let us focus on the former, the proof for the latter following suit. From Equation \eqref{Eq: pointwise product} and Equation \eqref{Eq: deformed product} it descends that we can define formally
		\begin{flalign}
			\label{eq2proof}
			\Gamma^W_{\cdot_{\mathsf{Q}}} (\Phi^{\rho} \bar{\Phi}_{\rho'}) (f; \eta,\bar{\eta}) &:= [ \Gamma^W_{\cdot_{\mathsf{Q}}} (\Phi^{\rho}) \cdot_{\mathsf{Q}} \Gamma^W_{\cdot_{\mathsf{Q}}}(\bar{\Phi}_{\rho'})] (f; \eta,\bar{\eta}) \\ \notag &= (\Phi^{\rho} \bar{\Phi}_{\rho'}) (f; \eta,\bar{\eta}) + \underbrace{((G_{\psi})^{\rho}_{\rho''} \cdot (G_{\bar{\psi}})^{\rho''}_{\rho'})(f \otimes 1_{M})}_{(1)}, 
		\end{flalign}
		where $1_M$ is the constant function identically equal to $1$ on $M$, whilst $(G_{\psi})^{\rho}_{\rho''} \cdot (G_{\bar{\psi}})^{\rho''}_{\rho'}$ denotes the pointwise product of the fundamental solution of the Dirac operator and of that of its adjoint. For the sake of notational conciseness, the spinor indexes will be omitted in the following. The expression involving a product of kernels in Equation \eqref{eq2proof} is \emph{a priori} ill-defined due to the singular behaviour of $G_{\psi}$ and $G_{\bar{\psi}}$ at the coinciding point limit. Being $G_{\psi}, G_{\bar{\psi}} \in \mathcal{D}'(M; \Sigma_d \otimes \Sigma^*_d)$, from the H\"ormander criterion, see \cite[Thm. 8.2.10]{Hormander-I-03}, we can only assert that the product $G_{\psi} \cdot G_{\bar{\psi}} \in \mathcal{D'}(M \times M \setminus \mathrm{Diag}_2; \Sigma_d\otimes\Sigma^*_d)$. According to \cite[Rem. B.7]{Dappiaggi:2020gge} generalized to the scaling degree at an embedded submanifold, one can conclude that
		\begin{equation*}
			\text{\textrm{sd}}_{\mathrm{Diag}_2} (G_{\psi} \cdot G_{\bar{\psi}}) \le \text{\textrm{sd}}_{\mathrm{Diag}_2}(G_{\psi}) + \text{\textrm{sd}}_{\mathrm{Diag}_2}(G_{\bar{\psi}}) < \infty,
		\end{equation*}
		due to the finiteness of the scaling degree of each term on the right hand-side - see Equation \ref{Eq: scaling degree fundamental solutions}. Observe that, in this last chain of inequalities, we are considering scaling degrees for vector valued distributions as per Equation \eqref{Eq: vector scaling degree}, but, to lighten the notation, we omit indicating explicitly the counterpart of the vector space $V$.
		\noindent Hence, \cite[Thm. 6.9]{Brunetti-Fredenhagen-00} entails that there exists at least one extension of the vector-valued bi-distribution $G_{\psi} \cdot G_{\bar{\psi}}$, preserving its wavefront set and its scaling degree. Let us select one of such extensions, which we will denote by $\widetilde{G_{\psi} \cdot G_{\bar{\psi}}}$, and let us set 
		\begin{equation}\label{eq3proof}
			\Gamma^W_{\cdot_{\mathsf{Q}}} (\Phi \bar{\Phi}) (f; \eta,\bar{\eta}) := (\Phi \bar{\Phi}) (f; \eta,\bar{\eta}) + (\widetilde{G_{\psi} \cdot G_{\bar{\psi}}}) (f \otimes 1_M), \, \, \forall f \in \mathcal{D}(M), \forall (\eta,\bar{\eta}) \in \Gamma(DM\oplus D^*M).
		\end{equation}
      
		We outline that, taking into account the singular structure of $G_{\psi}$ and $G_{\bar{\psi}}$ investigated in Section \ref{Sec: wavefront set}, one can conclude that $\Gamma^W_{\cdot_{\mathsf{Q}}} (\Phi \bar{\Phi}) \in \mathcal{D}'_C(M; \text{Pol}_W)$. By linearity, $\Gamma^W_{\cdot_{\mathsf{Q}}}$ can be extended to the whole $ \mathcal{M}^0_{1,1}$. \\
		\noindent At this point, being $\Gamma^W_{\cdot_{\mathsf{Q}}}$ coherently defined on the modulus $\mathcal{M}_{1,1}^0$, let us assume that the action of this map is well-defined also on $\mathcal{M}_{1,1}^j$ with $j \ge 0$. On account of Remark \ref{Rem: M}, it is immediate to deduce that any $u \in \mathcal{M}^{j+1}_{1,1}$ can be realized as a linear combination of functionals of the following form: $G_{\psi} \circledast u'$ or $G_{\bar{\psi}} \circledast u'$ for some $u' \in \mathcal{M}^j_{1,1}$ or products $u := u_1 u_2$, where $u_1 \in \mathcal{M}^j_{1,0} \cup G_{\psi} \circledast \mathcal{M}^j_{1,0} \cup G_{\bar{\psi}} \circledast \mathcal{M}^j_{1,0}$ whilst $u_2 \in \mathcal{M}^j_{0,1} \cup G_{\psi} \circledast \mathcal{M}^j_{0,1} \cup G_{\bar{\psi}} \circledast \mathcal{M}^j_{0,1}$. In the first instance, namely if $u:= G_{\psi} \circledast u'$ or $u:= G_{\bar{\psi}} \circledast u'$ with $u' \in \mathcal{M}^j_{1,1}$, the map $\Gamma^W_{\cdot_{\mathsf{Q}}}$ can be defined as in \emph{Step 1.}, at the same time exploiting the inductive hypothesis. In the second case, Equation \eqref{eq1pr} - which is still to be read in a purely formal sense - entails that
		\begin{equation*}
			\Gamma^W_{\cdot_{\mathsf{Q}}} (u) (f; \eta, \bar{\eta}) = [\Gamma^W_{\cdot_{\mathsf{Q}}} (u_1) \cdot_{\mathsf{Q}} \Gamma^W_{\cdot_{\mathsf{Q}}} (u_2)] (f; \eta, \bar{\eta}) = [\Gamma^W_{\cdot_{\mathsf{Q}}} (u_1) \cdot \Gamma^W_{\cdot_{\mathsf{Q}}} (u_2)] (f; \eta, \bar{\eta}) + \mathfrak{Q}(f \otimes 1_3), 
		\end{equation*}
		where 
		\begin{equation}
			\label{illprod}
			\mathfrak{Q}(f \otimes 1_3) := [(\delta_{\mathrm{Diag}_2} \otimes Q) \cdot ( \Gamma^W_{\cdot_{\mathsf{Q}}} (u_1)^{(1,0)} \tilde{\otimes} \Gamma^W_{\cdot_{\mathsf{Q}}} (u_2)^{(0,1)})] (f \otimes 1_3),
		\end{equation}
	Notwithstanding, Equation \eqref{illprod} involves the product of singular vector-valued distributions and, hence, it needs to undergo a renormalization procedure. As a preliminary step, we outline that the microlocal behaviour of the vector-valued bi-distribution $Q$ highlighted in Remark \ref{Rem: All over Q} combined with the inductive hypothesis, \textit{i.e.},
		\begin{equation}\label{Eq: WF estimate}
			\mathrm{WF}(\Gamma^W_{\cdot_{\mathsf{Q}}}(u_1)^{(1,0)}) \cup \mathrm{WF}(\Gamma^W_{\cdot_{\mathsf{Q}}}(u_2)^{(0,1)}) \subseteq C_2 = \mathrm{WF}(\delta_{\mathrm{Diag}_2}) = \mathrm{WF}(Q),
		\end{equation}
		implies that
		\begin{equation*}
			\mathfrak{Q} \in \mathcal{D}'(M^4 \setminus \mathrm{Diag}^{big}_4; W), 
		\end{equation*}
		being
		\begin{equation*}
			\mathrm{Diag}^{big}_4 := \{(x_1, ..., x_4) \in M^4 \, | \, \exists i,j \in \{1,2,3,4\}, x_i = x_j\}.
		\end{equation*}
		It is worth noting that Equation \eqref{Eq: WF estimate} can be further improved by observing that, for any $x \in \mathrm{Diag}^{big}_4 \setminus \mathrm{Diag}_4$, exactly one of the factors among $\delta_{\mathrm{Diag}_2} \otimes Q$, $\Gamma^W_{\cdot_{\mathsf{Q}}}(u_1)^{(1,0)}$ and $\Gamma^W_{\cdot_{\mathsf{Q}}}(u_2)^{(1,0)}$ is generated by a smooth function, whereas the product of the others is well-defined. This implies that $\mathfrak{Q} \in \mathcal{D}'(M^4 \setminus \mathrm{Diag}_4; W)$, hence it suffices to extend it to the thin diagonal of $M^4$. Moreover, for what concerns the scaling degree of $\mathfrak{Q}$, \cite[Rem. B.7]{Dappiaggi:2020gge} guarantees that 
		\begin{flalign*}
			\text{\textrm{sd}}_{\mathrm{Diag}_4} (\mathfrak{Q}) &\le \text{\textrm{sd}}_{\mathrm{Diag}_4} (\delta_{\mathrm{Diag}_2} \otimes Q) + \text{\textrm{sd}}_{\mathrm{Diag}_4}(\Gamma^W_{\cdot_{\mathsf{Q}}}(u_1)^{(1,0)} \otimes \Gamma^W_{\cdot_{\mathsf{Q}}}(u_2)^{(0,1)}) \\
			&\le \text{\textrm{sd}}_{\mathrm{Diag}_2} (\delta_{\mathrm{Diag}_2}) + \text{\textrm{sd}}_{\mathrm{Diag}_2} (Q) + \text{\textrm{sd}}_{\mathrm{Diag}_2}(\Gamma^W_{\cdot_{\mathsf{Q}}}(u_1)^{(1,0)}) + \text{\textrm{sd}}_{\mathrm{Diag}_2}(\Gamma^W_{\cdot_{\mathsf{Q}}}(u_2)^{(0,1)}) < \infty.
		\end{flalign*}
		being each of the addend on the right hand-side of the above estimate finite. Thus, applying again \cite[Thm. 6.9]{Brunetti-Fredenhagen-00}, we conclude the existence of at least one extension of $\mathfrak{Q}$ preserving both the scaling degree and the wavefront set. Choosing one among such extensions - denoted by $\widetilde{\mathfrak{Q}}$ - we can set eventually 
		\begin{equation*}
			\Gamma^W_{\cdot_{\mathsf{Q}}} (u) (f; \eta, \bar{\eta}) = [\Gamma^W_{\cdot_{\mathsf{Q}}}(u_1) \Gamma^W_{\cdot_{\mathsf{Q}}}(u_2)](f; \eta, \bar{\eta}) + \widetilde{\mathfrak{Q}}(f \otimes 1_3),
		\end{equation*}
		for all $f \in \mathcal{D}(M)$ and for any field configurations $(\eta, \bar{\eta}) \in \Gamma(DM\oplus D^*M)$. By direct inspection, it can be seen that the previous definition satisfies all the requirements of the theorem. Therefore, this concludes the proof in the case $(r,r') = (1,1)$. \\
		\noindent \underline{\emph{Step 3: Generic $(r,r')$.}} The remaining task consists in proving the inductive step with respect to the indexes $(r,r')$. Specifically, let us suppose that the map $\Gamma^W_{\cdot_{\mathsf{Q}}}$ has been coherently assigned on $\mathcal{M}_{r,r'}$ and let us show that it can be defined consistently on $\mathcal{M}_{r+1,r'}$. We outline that one should prove the sought result also for $\mathcal{M}_{r,r'+1}$. However, since the proof of this last case mirrors that of the former, we omit it. \\
		\noindent In the spirit of the proof of the case $(r,r')=(1,1)$ we consider the following decomposition 
		\begin{equation*}
			\mathcal{M}_{r+1, r'} = \bigoplus_{j \in \mathbb{N}_0} \mathcal{M}^j_{r+1,r'} \, \, \text{with} \, \, \mathcal{M}^j_{r+1,r'} := \mathcal{M}_{r+1, r'} \cap \mathcal{A}^W_j, \, \forall j \in \mathbb{N}_0,
		\end{equation*}
		and we work inductively over the index $j \ge 0$. For $j=0$, one has that 
		\begin{equation*}
			\mathcal{M}^0_{r+1,r'} := \text{span}_{\mathcal{E}({M; W})} \{1, \Phi, \bar{\Phi}, ..., \Phi^{r+1} \bar{\Phi}^{r'}, \bar{\Phi}^{r'} \Phi^{r+1}\}.
		\end{equation*}
		The inductive hypothesis entails that the action of $\Gamma^W_{\cdot_{\mathsf{Q}}}$ has been already established on all the generators of  $\mathcal{M}^0_{r+1,r'}$, except for $\Phi^{r+1} \bar{\Phi}^{r'}$ and $\bar{\Phi}^{r'} \Phi^{r+1}$. Once the action of the map $\Gamma^W_{\cdot_{\mathsf{Q}}}$ has been defined on these terms, it can be extended by linearity, thus obtaining the sought after result. \\
		\noindent Plugging into Equation \eqref{eq1pr} the explicit form of Equation \eqref{Eq: deformed product}, one obtains the following
		\begin{equation*}
			[\Gamma^W_{\cdot_{\mathsf{Q}}} (\Phi^{r+1} \bar{\Phi}^{r'})] (f; \eta, \bar{\eta}) = \sum_{k=0}^{\min \{r+1,r'\}} k! \begin{pmatrix}
				r+1 \\
				k
			\end{pmatrix} \begin{pmatrix}
				r'\\
				k
			\end{pmatrix} [Q_{2k} \cdot (\Gamma^W_{\cdot_{\mathsf{Q}}} (\Phi)^{r+1-k} \Gamma^W_{\cdot_{\mathsf{Q}}} (\bar{\Phi})^{r'-k})](f; \eta, \bar{\eta}), 
		\end{equation*}
		where 
		\begin{equation}
			\label{Q2k}
			Q_{2k}(f) := (G_{\psi} G_{\bar{\psi}})^{\otimes k} \cdot (\delta_{\mathrm{Diag}_k} \otimes 1_k) (f \otimes 1_{2k-1}), \, \forall f \in \mathcal{D}(M).
		\end{equation}
		Reasoning as in \emph{Step 2} of this proof, one can conclude that, despite being \textit{a priori} ill-defined, the bi-distribution $G_{\psi} G_{\bar{\psi}} \in \mathcal{D}'(M^2 \setminus \mathrm{Diag}_2; W)$ can be extended to the whole $M\times M$ in such a way to preserve the scaling degree. Let us choose one among such admissible extensions and let us denote it by $\widetilde{G_{\psi} G_{\bar{\psi}}}$. Thus, replacing $G_{\psi} G_{\bar{\psi}}$ with $\widetilde{G_{\psi} G_{\bar{\psi}}}$ in Equation \eqref{Q2k}, one obtains 
		\begin{equation*}
			\widetilde{Q_{2k}}(f) := (\widetilde{G_{\psi} G_{\bar{\psi}}})^{\otimes k} \cdot (\delta_{\mathrm{Diag}_k} \otimes 1_k) (f \otimes 1_{2k-1}), \, \forall f \in \mathcal{D}(M),
		\end{equation*}
		which is the renormalized version of $Q_{2k}$. 
		Therefore, for any $f\in\mathcal{D}(M)$ and for any $\eta, \bar{\eta}\in\Gamma(DM\oplus D^*M)$, we can set 
		\begin{equation*}
			[\Gamma^W_{\cdot_{\mathsf{Q}}} (\Phi^{r+1} \bar{\Phi}^{r'})] (f; \eta, \bar{\eta}) = \sum_{k=0}^{\min \{r+1,r'\}} k! \begin{pmatrix}
				r+1 \\
				k
			\end{pmatrix} \begin{pmatrix}
				r'\\
				k
			\end{pmatrix} [\widetilde{Q}_{2k} \cdot (\Gamma^W_{\cdot_{\mathsf{Q}}} (\Phi)^{r+1-k} \Gamma^W_{\cdot_{\mathsf{Q}}} (\bar{\Phi})^{r'-k})](f; \eta, \bar{\eta}). 
		\end{equation*}
		We note that, thanks to the renormalization of $Q_{2k}$, this expression contains only products of distributions generated by smooth functions. Hence, it is well-defined. In addition, it is manifest that it satisfies all the requirements in the statement of the theorem. \\
		\noindent We are left with one last task, namely proving through an inductive argument on the index $j$ that the map $\Gamma^W_{\cdot_{\mathsf{Q}}}$ can be coherently constructed on $\mathcal{M}^{j+1}_{r+1,r'}$, once it has been assigned on $\mathcal{M}^j_{r+1,r'}$. Given $u \in \mathcal{M}^{j+1}_{r+1,r'}$, similarly to \emph{Step 2.} of this proof, the linearity of $\Gamma^W_{\cdot_{\mathsf{Q}}}$ entails that it suffices to work with elements either of the form $u = G_{\psi} \circledast u'$ - or, similarly, $u = G_{\bar{\psi}} \circledast u'$ - with $u' \in \mathcal{M}^{j}_{r+1,r'}$ or of the form 
		\begin{equation}
			\label{eq4proof}
			u = u_1 ... u_l, \, \, \text{with} \, \, u_i \in \mathcal{M}^j_{r_i, r'_i} \cup G_{\psi} \circledast \mathcal{M}^j_{r_i, r'_i} \cup G_{\bar{\psi}} \circledast \mathcal{M}^j_{r_i, r'_i}, \, i \in \{1,...,l\}, l \in \mathbb{N},
		\end{equation}
		where $r_i, r'_i \in \mathbb{N}$ with $\sum_{i=1}^l r_i = r+1$, $\sum_{i=1}^l r'_i = r'$, for any $l \in \mathbb{N}$. \\
		\noindent In the first case, the sought conclusion can be drawn exploiting the inductive hypothesis combined with \emph{Step 1.} of this proof. Therefore, we can devote our attention to studying the case in which $u$ is realized as in Equation \eqref{eq4proof}. The inductive procedure entails that the map $\Gamma^W_{\cdot_{\mathsf{Q}}}$ is well-defined on each of the factors $u_i$, $\forall i \in \{1, ..., l\}$. Using once more Equation \eqref{eq1pr}, one can formally write
		\begin{flalign}
			\label{eq5proof}
			\Gamma^W_{\cdot_{\cdot_{\mathsf{Q}}}}(u) (f; \eta, \bar{\eta}) &= \sum_{N, \widetilde{N}\ge 0} \sum_{\substack{N_1 + ... + N_l = N+ \widetilde{N} \\ \widetilde{N}_1 + ... +\widetilde{N}_l = N + \widetilde{N}}} \mathsf{F} (N_1, ..., N_l; M_1, ..., M_l) \cdot [(\delta_{\mathrm{Diag}_l} \otimes Q^{\otimes N} \otimes \widetilde{Q}^{\otimes \widetilde{N}}) \\ \notag
			& \left( (\Gamma^W_{\cdot_{\mathsf{Q}}}(u_1))^{(N_1, \widetilde{N}_1)} \tilde{\otimes} ... \tilde{\otimes} (\Gamma^W_{\cdot_{\mathsf{Q}}}(u_l))^{(N_l, \widetilde{N}_l)} \right) ] (f \otimes 1_{l-1+2N+2M}; \eta, \bar{\eta}),     
		\end{flalign}
		for all $f \in \mathcal{D}(M)$, $(\eta, \bar{\eta}) \in \Gamma(DM\oplus D^\ast M)$, where $\mathsf{F} (N_1, ..., N_l; M_1, ..., M_l) \in \mathbb{C}$ are multiplicative coefficients emerging from the underlying combinatorics. It is worth highlighting that, working with polynomial functional-valued vector distributions, only a finite number of terms contributes to the sums on the right hand-side of Equation \eqref{eq5proof}. Let us define 
		\begin{equation*}
			\mathfrak{Q}_{N, \widetilde{N}} := [(\delta_{\mathrm{Diag}_l} \otimes Q^{\otimes N} \otimes \widetilde{Q}^{\otimes \widetilde{N}}) \left( (\Gamma^W_{\cdot_{\mathsf{Q}}}(u_1))^{(N_1, \widetilde{N}_1)} \tilde{\otimes} ... \tilde{\otimes} (\Gamma^W_{\cdot_{\mathsf{Q}}}(u_l))^{(N_l, \widetilde{N}_l)} \right)].
		\end{equation*}
		As a consequence, \cite[Thm. 8.2.9]{Hormander-I-03} implies that 
		\begin{flalign*}
			\mathrm{WF} (\delta_{\mathrm{Diag}_l} \otimes Q^{\otimes N} \otimes \widetilde{Q}^{\otimes \widetilde{N}}) &\subseteq \{ (\widetilde{x}_l, \widetilde{z}_{2(N + \widetilde{N})}; \widetilde{k}_l, \widetilde{q}_{2(N + \widetilde{N})} ) \in \mathring{T}^*(M)^{l + 2(N + \widetilde{N})} \, | \, \\ 
			& (\widetilde{x}_l, \widetilde{k}_l) \in \mathrm{WF}(\delta_{\mathrm{Diag}_l}), (\widetilde{z}_{2(N + \widetilde{N})}, \widetilde{q}_{2(N + \widetilde{N})}) \in \mathrm{WF}(Q^{\otimes N} \otimes \widetilde{Q}^{\otimes \widetilde{N}}) \}, 
		\end{flalign*}
		while by the inductive hypothesis, one has that $\mathrm{WF}(\Gamma^W_{\cdot_{\mathsf{Q}}}(u_i))^{(N_i, \widetilde{N}_i)} \subseteq C_{N_i + \widetilde{N}_i +1}$ for all $i \in \{1, ..., l\}$, where the the set $C_{N_i + \widetilde{N}_i +1}$ are in agreement with Equation \eqref{Eq: C_m}. Hence, on the account of the H\"ormander criterion, we can assert that 
		\begin{equation*}
			\mathfrak{Q}_{N, \widetilde{N}} \in \mathcal{D}'(M^{l + 2(N+ \widetilde{N})} \setminus \mathrm{Diag}^{big}_{l + 2(N + \widetilde{N})};W),
		\end{equation*}
		where 
		\begin{equation*}
			\mathrm{Diag}^{big}_{l + 2(N + \widetilde{N})} := \{ (x_1, ..., x_{l + 2(N + \widetilde{N})}) \in M^{l + 2(N + \widetilde{N})} \, | \, \exists i,j \in \{1, ..., l + 2(N + \widetilde{N})\}, x_i = x_j\}. 
		\end{equation*}
		Hence
		\begin{flalign*}
			\mathrm{WF} (\mathfrak{Q}_{N, \widetilde{N}}) &= \{ (\widetilde{x}_l, \widetilde{z}_{2(N + \widetilde{N})}; \widetilde{k}_l + \widetilde{k'}_l, \widetilde{q}_{N_1 + \widetilde{N}_1} + \widetilde{q'}_{N_1 + \widetilde{N}_1}, ..., \widetilde{q}_{N_l + \widetilde{N}_l} + \widetilde{q'}_{N_l + \widetilde{N}_l}) \in \\ & T^*(M)^{l + 2(N + \widetilde{N})}\setminus\{0\} \, | \, (\widetilde{x}_l, \widetilde{k}_l) \in \mathrm{WF}(\delta_{\mathrm{Diag}_l}),
			(\widetilde{z}_{2(N + \widetilde{N})}, \widetilde{q}_{2(N + \widetilde{N})}) \in \mathrm{WF}(Q^{\otimes N} \otimes \widetilde{Q}^{\otimes \widetilde{N}}), \\ &(x_i, \widetilde{z}_{N_i, \widetilde{N}_i)}, \widetilde{q'}_{N_i + \widetilde{N}_i)}) \in C_{1+N_i +\widetilde{N}_i}, \, \forall i \in \{1,...,l\}\}. 
		\end{flalign*}
		By a generalization of the strategy discussed in \emph{Step 2.} of this proof, the above estimate for the wavefront set of $\mathfrak{Q}_{N, \widetilde{N}}$ can be additionally improved. Consider a partition $\{I,J\}$ of $\{1, ..., l + N + \widetilde{N}\}$ into two disjoint subsets such that if $\{x_1, ..., x_{l +2(N + \widetilde{N})}\} = \{\widetilde{x}_I, \widetilde{x}_J\}$, then  $x_i \ne x_j$ for all $x_i \in \widetilde{x}_I$ and $x_j \in \widetilde{x}_J$. Hence, we can write 
		\begin{equation*}
			\mathfrak{Q}_{N, \widetilde{N}}(x_1, ..., x_{l +2(N + \widetilde{N})}) = K_{N, \widetilde{N}}^I (\widetilde{x}_I) S_{N, \widetilde{N}} (\widetilde{x}_I, \widetilde{x}_J) K_{N, \widetilde{N}}^J (\widetilde{x}_J),  
		\end{equation*}
		where $S_{N, \widetilde{N}} (\widetilde{x}_I, \widetilde{x}_J)$ is smooth on the aforementioned partition, whilst $K_{N, \widetilde{N}}^I$ and $K_{N, \widetilde{N}}^J$ are kernels of vector-valued distributions in the definition of $\Gamma^{W}_{\cdot_{\mathsf{Q}}}$ on $\mathcal{M}^n_{(k,k')}$ for every $k < r+1$, $k' < r'+1$ and $n < j+1$. Taking into account also the inductive hypothesis, it descends that the product $(K_{N, \widetilde{N}}^I \otimes K_{N, \widetilde{N}}^J) \cdot S_{N, \widetilde{N}}$ is well-defined. Being the above argument independent from the choice of the partition $\{I,J\}$, one has that $\mathfrak{Q}_{N, \widetilde{N}} \in \mathcal{D}'(M^{l + 2(N+ \widetilde{N})} \setminus \mathrm{Diag}_{l + 2(N + \widetilde{N})}; W)$. Eventually, we observe that 
		\begin{flalign*}
			\text{\textrm{sd}}_{\mathrm{Diag}_{l + 2(N + \widetilde{N})}} (\mathfrak{Q}_{N, \widetilde{N}}) &\le \text{\textrm{sd}}_{\mathrm{Diag}_{l + 2(N + \widetilde{N})}} (\delta_{\mathrm{Diag}_l} \otimes Q^{\otimes N} \otimes \widetilde{Q}^{\otimes \widetilde{N}}) \\ &+ \sum_{i=1}^l \text{\textrm{sd}}_{\mathrm{Diag}_{l + 2(N + \widetilde{N})}} ((\Gamma^W_{\cdot_{\mathsf{Q}}})^{(N_i, \widetilde{N}_i)}) < \infty.
		\end{flalign*}
		Thus, invoking once again \cite[Thm. 6.9]{Brunetti-Fredenhagen-00} there exists at least one extension of the singular vector-valued bi-distribution $\mathfrak{Q}_{N, \widetilde{N}}$, preserving both its wavefront set and its scaling degree. Denoting by $\widetilde{\mathfrak{Q}}_{N, \widetilde{N}} \in \mathcal{D}'(M^{l+2(N+\widetilde{N})};W)$ one among such extensions, we are able to set
		\begin{equation}
			\label{eq6proof}
			\Gamma^W_{\cdot_{\mathsf{Q}}}(u) (f; \eta, \bar{\eta}) = \sum_{N, \widetilde{N} \ge 0} \sum_{\substack{N_1 + ... + N_l = N+ \widetilde{N} \\ \widetilde{N}_1 + ... +\widetilde{N}_l = N + \widetilde{N}}} \mathsf{F} (N_1, ..., N_l; M_1, ..., M_l) \cdot \widetilde{\mathfrak{Q}}_{N, \widetilde{N}}. 
		\end{equation}
		As in the proof of \cite[Thm. 3.4]{BDR23}, it can be shown that Equation \eqref{eq6proof} satisfies all the defining properties required in the statement of this theorem. The thesis is thus proven. 
\end{proof}

The deformation map $\Gamma_{\cdot_{\mathsf{Q}}}^W$, whose existence has been proven in Proposition \ref{Prop: deformation map}, is a fundamental tool to encode the information on the underlying white noise as a suitable deformation of the algebra $\mathcal{A}^W$ introduced in Definition \ref{Def: A}.
\begin{theorem}\label{Thm: local algebra}
	Let $\Gamma_{\cdot_{\mathsf{Q}}}^W:\mathcal{A}^W\rightarrow \mathcal{D}'_C(M;\mathsf{Pol}_W)$ be the deformation map as per Proposition \ref{Prop: deformation map} and let $\mathcal{A}_{\cdot_{\mathsf{Q}}}^W:=\Gamma_{\cdot_{\mathsf{Q}}}^W(\mathcal{A}^W)$ be such that, for all $u_1,u_2\in\mathcal{A}^W_{\cdot_{\mathsf{Q}}}$,
	\begin{equation}
		u_1\cdot_{\Gamma_{\cdot_{\mathsf{Q}}}^W}u_2:=\Gamma_{\cdot_{\mathsf{Q}}}^W[{(\Gamma_{\cdot_{\mathsf{Q}}}^{W})}^{-1}(u_1) {(\Gamma_{\cdot_{\mathsf{Q}}}^{W})}^{-1}(u_2)].
	\end{equation}
Then, $(\mathcal{A}_{\cdot_{\mathsf{Q}}}^W,\cdot_{\Gamma_{\cdot_{\mathsf{Q}}}^W})$ is a unital, associative algebra over W.
\end{theorem} 
\begin{proof}
	As a preliminary step, let us observe that the map 	$\Gamma_{\cdot_{\mathsf{Q}}}^W:\mathcal{A}^W\rightarrow \mathcal{D}'_C(M;\mathsf{Pol}_W)$ of Proposition \ref{Prop: deformation map} is injective, namely that $\ker(\Gamma_{\cdot_{\mathsf{Q}}}^W)=\emptyset$. Let us consider $u\in\ker(\Gamma_{\cdot_{\mathsf{Q}}}^W)\setminus\{0\}$ a polynomial functional-valued vector distribution of degree $(k,k')\in\mathbb{N}_0\times\mathbb{N}_0$. 
	
	This entails that $\delta^{(k,k')}_{(\xi,\bar{\xi})}u\neq 0$, while $\delta^{(k+1,k')}_{(\xi,\bar{\xi})}u=\delta^{(k,k'+1)}_{(\xi,\bar{\xi})}u=0$ for all $(\xi,\bar{\xi})\in\Gamma(DM\oplus D^\ast M)$. Equation \eqref{Eq: identity} entails that $\mathbf{1}\notin\ker(\Gamma_{\cdot_{\mathsf{Q}}}^W)$, hence $k,k'>0$. Thus, for all $(\xi,\bar{\xi})\in\Gamma(DM\oplus D^\ast M)$, there must exist $0\neq f_{(\xi,\bar{\xi})}\in\mathcal{E}(M^{k+k'+1})$ such that
	\begin{equation}
		\delta^{(k,k')}_{(\xi,\bar{\xi})}u=f_{(\xi,\bar{\xi})}\mathbf{1}\,.
	\end{equation}
Furthermore, Equation \eqref{Eq: deformation derivatives} implies that if $u\in\ker(\Gamma_{\cdot_{\mathsf{Q}}}^W)$, then also $\delta_{(\xi,\bar{\xi})}u\in\ker(\Gamma_{\cdot_{\mathsf{Q}}}^W)$. We can infer that $f_{(\xi,\bar{\xi})}\mathbf{1}\in\ker(\Gamma_{\cdot_{\mathsf{Q}}}^W)$ and, as a consequence, that $f_{(\xi,\bar{\xi})}=0$, leading to a contradiction. Injectivity of $\Gamma_{\cdot_{\mathsf{Q}}}^W$ entails the
well-definiteness of the product $\cdot_{\Gamma_{\cdot_{\mathsf{Q}}}^W}$. The algebra is unital since $u\cdot_{\Gamma_{\cdot_{\mathsf{Q}}}^W}\mathbf{1}=\mathbf{1}\cdot_{\Gamma_{\cdot_{\mathsf{Q}}^W}} u=u$ for 
all $u\in\mathcal{A}_{\cdot_{\mathsf{Q}}}^W$, while associativity is a direct consequence of the defining properties of $\mathcal{A}^W$, see \cite{Dappiaggi:2020gge}. 
\end{proof}
 In Proposition \ref{Prop: deformation map} we proved existence of a deformation map with the desired properties. In its proof we faced the problem of renormalization of a priori ill-defined structures, dealing with it via microlocal techniques, which are often used in the algebraic approach of quantum field theory. Most notably, this procedure allows to account for renormalization ambiguities directly at the level of the algebra deformation, avoiding in particular the choice of a  regularization scheme. Yet, in the proof of Proposition \ref{Prop: deformation map} we highlighted how the extension of singular distributions over their singular support may be non-unique. This arbitrariness reflects at the level of the deformation map as follows.
 \begin{theorem}\label{Thm: uniqueness}
 	Given two linear maps $\Gamma_{\cdot_{\mathsf{Q}}}^W,\widetilde{\Gamma}_{\cdot_{\mathsf{Q}}}^W :\mathcal{A}^W\rightarrow \mathcal{D}'_C(M;\mathsf{Pol}_W)$ built as per Proposition \ref{Prop: deformation map} there exists a collection of linear maps $\{C_{l,l'}\}_{l,l'\in\mathbb{N}_0}$, $C_{l,l'}:\mathcal{A}^W\rightarrow\mathcal{M}_{l,l'}$, see Remark \ref{Rem: M}, such that
 	\begin{enumerate}
 		\item for any $r, r'\in\mathbb{N}$ such that either $r\leq j + 1$ or $r'\leq j'+1$ , one has that
 		\begin{equation}
 			C_{j,j'}[\mathcal{M}_{r,r'}]=0,
 		\end{equation}
 	\item for all $l,l'\in\mathbb{N}_0$ and for any $u\in\mathcal{A}^W$, it holds that
 	\begin{equation}
 		C_{l,l'}(G_\psi\circledast u)=G_\psi\circledast C_{l,l'}(u),\qquad C_{l,l'}(G_{\bar{\psi}}\circledast u)=G_{\bar{\psi}}\circledast C_{l,l'}(u),
 	\end{equation}
 \item for any $l,l'\in\mathbb{N}$ and for all $(\xi,\bar{\xi})\in\Gamma(DM\oplus D^\ast M)$, one has that
 \begin{equation}
 	\delta_\xi \circ C_{l,l'}=C_{l-1,l'}\circ\delta_{xi},\qquad\delta_{\bar{\xi}}\circ C_{l,l'}=C_{l,l'-1}\circ\delta_{\bar{\xi}},
 \end{equation}
	\item for all $u\in\mathcal{M}_{l,l'}$ the following identity holds true
	\begin{equation}\label{Eq: relation renormalization}
		\widetilde{\Gamma}_{\cdot_{\mathsf{Q}}}^W(u)=\Gamma_{\cdot_{\mathsf{Q}}}^W(u+C_{l-1,l'-1}(u)).
	\end{equation}
 	\end{enumerate}
 \end{theorem}

\noindent The proof of this theorem is the same as in \cite{DDR20} barring minor modifications and therefore we omit it.

\begin{remark}
	The relevance of Theorem \ref{Thm: uniqueness} resides in the characterization of the arbitrariness in the choice of a deformation map. In particular, Equation \eqref{Eq: relation renormalization} must be read as the relation that elapses between two different deformation maps, hence translating to the present setting the renormalization freedom of the field theory under analysis. Making an analogy with quantum field theory, the maps $C_{l,l'}$ coincide with the \textit{renormalization freedoms}. 
\end{remark}
\begin{corollary}
	Under the same hypotheses of Theorem \ref{Thm: uniqueness}, there exists an algebra isomorphism between $\mathcal{A}_{\cdot_{\mathsf{Q}}}^W:=\Gamma_{\cdot_{\mathsf{Q}}}^W(\mathcal{A}^W)$ and $\widetilde{\mathcal{A}}_{\cdot_{\mathsf{Q}}}^W:=\widetilde{\Gamma}_{\cdot_{\mathsf{Q}}}^W(\mathcal{A}^W)$.
\end{corollary}

\section{The Multi-Local Algebra $\mathcal{A}_{\bullet_{\mathsf{Q}}}^W$}
\label{Sec: nonlocal deformation algebra}
The procedure outlined in Section \ref{Sec: local deformation}, especially Theorem \ref{Thm: local algebra}, suffices to encompass at the algebraic level the information on the underlying random field. It turn this shall allow us to compute in Section \ref{Sec: Thirring} the expectation value of the perturbative solution to the stochastic Thirring model. Nonetheless, the aforementioned algebraic construct, being intrinsically local, does not suffice to compute multi-local correlation functions of the fields, which play a pivotal r\^ole in a vast range of physical applications. The most natural way of circumventing this hurdle, \textit{i.e.}, deriving explicit expressions of the multi-local correlation functions, is to resort once more to the deformation of a suitable algebra of functionals. The first step in this direction consists of introducing a suitable non local version of the algebraic structures introduced in Sections \ref{Sec: algebras} and \ref{Sec: local deformation}.
\begin{definition}\label{Def: tensor algebras}
	Let $\mathcal{D}'_C(M;\mathsf{Pol}_W)$ be the space of functional-valued vector distributions introduced in Definition \ref{Def: algebra WF} and consider the deformed local algebra $\mathcal{A}^W_{\cdot_{\mathsf{Q}}}$ as per Theorem \ref{Thm: local algebra}. We define
	\begin{equation}\label{Eq: multilocal space}
		(\mathcal{T}^W_C)'(M; \mathsf{Pol}_W) :=W\oplus\bigoplus_{l\geq1}\mathcal{D}_C'(M;\mathsf{Pol}_W)^{\otimes l},
	\end{equation}
where $W$ us defined in Equation \eqref{Eq: W space}. Similarly we introduce the universal tensor algebra
\begin{equation}\label{Eq: universal tensor module}
	\mathcal{T}^W(\mathcal{A}_{\cdot_{\mathsf{Q}}}^W):=\mathcal{E}(M;W)\oplus\bigoplus_{l\geq1}(\mathcal{A}_{\cdot_{\mathsf{Q}}}^W)^{\otimes l}.
\end{equation}
\end{definition}
\begin{remark}
	Exploiting the graded nature of $\mathcal{A}^W$ as per Remark \ref{Rem: M}, we can decompose $\mathcal{T}^W(\mathcal{A}_{\cdot_{\mathsf{Q}}}^W)$ as
	\begin{align}\label{Eq: graded decomposition tensor algebra}
		\mathcal{T}^W(\mathcal{A}_{\cdot_{\mathsf{Q}}}^W)=\mathcal{E}(M;W)\oplus\bigoplus_{l>0}\bigoplus_{r,r'=0}^\infty\bigoplus_{\substack{r_1,\ldots, r_l,r'_1,\ldots, r_l'\\r_1+\ldots+ r_l=r\\ r'_1+\ldots+ r'_l=r'}}\Gamma_{\cdot_{\mathsf{Q}}}^W(\mathcal{M}_{r_1,r_1'})\otimes\ldots\otimes \Gamma_{\cdot_{\mathsf{Q}}}^W(\mathcal{M}_{r_l,r_l'}).
	\end{align}
\end{remark}

Following the line of reasoning adopted in Section \ref{Sec: local deformation} we endow the algebra $\mathcal{T}^W(\mathcal{A}_{\cdot_{\mathsf{Q}}}^W)$ with $\bullet_{\mathsf{Q}}$, a deformation of the tensor product, encoding the information of Equations \eqref{Eq: 2-point function 1} and \eqref{Eq: 2-point function 2}. To this end, for all $u_1\in\mathcal{D}'_C(M^{l_1};\mathsf{Pol}_W)$ and $u_2\in\mathcal{D}'_C(M^{l_2};\mathsf{Pol}_W)$ with $l_1,l_2\in\mathbb{N}$, for any $f_1\in\mathcal{D}(M^{l_1})$, $f_2\in\mathcal{D}(M^{l_2})$ and for all $(\eta,\bar{\eta})\in\Gamma(DM\oplus D^\ast M)$ we set
\begin{align}\label{Eq: bullet product}
	(u_1\bullet_{\mathsf{Q}} u_2)(&f_1\otimes f_2;\eta,\bar{\eta})\nonumber\\
	&:=\sum_{\substack{k\geq0\\k_1+k_2=k}}\frac{1}{k_1! k_2!}\left[(\mathbf{1}_{l_1+l_2}\otimes Q^{\otimes k_1}\otimes \widetilde{Q}^{\otimes k_2})\right]\cdot \left(u_1^{(k_1,k_2)}\tilde{\otimes}u_1^{(k_2,k_1)}\right)(f_1\otimes f_2\otimes \mathbf{1}_{2k};\eta,\bar{\eta}),
\end{align}
where the symbol $\tilde{\otimes}$ has been introduced in Equation \eqref{Eq: expression deformed product}. We have all the ingredients to state the nonlocal counterpart of Proposition \ref{Prop: deformation map}. The proof can be seen both as an adaptation to the case in hand of that of Proposition \ref{Prop: deformation map} or of that in \cite[Thm. 4.4]{DDR20} for the scalar case and therefore we omit it.
\begin{proposition}\label{Prop: nonlocal deformation map}
	Consider the deformed algebra $\mathcal{A}_{\cdot_{\mathsf{Q}}}^W$ introduced in Theorem \ref{Thm: local algebra}. Furthermore, let us denote by $\mathcal{T}^W(\mathcal{A}_{\cdot_{\mathsf{Q}}}^W)$ the universal tensor algebra as per Equation \eqref{Eq: universal tensor module} and by $(\mathcal{T}^W_C)'(M; \mathsf{Pol}_W)$ the space defined in Equation \eqref{Eq: multilocal space}. Then, there exists a linear map
	\begin{equation}
 \Gamma_{\bullet_{\mathsf{Q}}}^W:\mathcal{T}^W(\mathcal{A}_{\cdot_{\mathsf{Q}}}^W)\rightarrow (\mathcal{T}^W_C)'(M; \mathsf{Pol}_W),
	\end{equation}
which satisfies the following conditions:
\begin{enumerate}
	\item for all $n\in\mathbb{N}$ and for any $u_1,\ldots, u_n\in\mathcal{A}_{\cdot_{\mathsf{Q}}}^W$ with either $u_1\in\Gamma_{\cdot_{\mathsf{Q}}}^W(\mathcal{M}_{1,0})$ or $u_1\in\Gamma_{\cdot_{\mathsf{Q}}}^W(\mathcal{M}_{0,1})$,
	see Remark \ref{Rem: M}, it holds that
	\begin{equation}
		\Gamma_{\bullet_{\mathsf{Q}}}^W(u_1\otimes\ldots\otimes u_n)=u_1\bullet_{\mathsf{Q}}\Gamma_{\bullet_{\mathsf{Q}}}^W(u_2\otimes\ldots\otimes u_n),
	\end{equation}
where $\bullet_{\mathsf{Q}}$ has been defined in Equation \eqref{Eq: bullet product}.
\item for all $n\in\mathbb{N}$ and for every $u_1,\ldots, u_n\in\mathcal{A}_{\cdot_{\mathsf{Q}}}^W$, $(\eta,\bar{\eta})\in\Gamma(DM\oplus D^\ast M)$ and $f_1,\ldots, f_n\in\mathcal{D}(M)$ such
that there exists a subset of indices $I\subset\{1,\ldots, n\}$ satisfying
\begin{equation}
	\left(\bigcup_{i\in I} supp(f_i)\right)\cap\left(\bigcup_{i\notin I} supp(f_j)\right)=\emptyset,
\end{equation}
it holds that
\begin{align}
	\Gamma_{\bullet_{\mathsf{Q}}}^W(u_1\otimes\ldots\otimes u_n)(f_1\otimes\ldots\otimes f_n;\eta,\bar{\eta})=\left[\Gamma_{\bullet_{\mathsf{Q}}}^W\Big(\bigotimes_{i\in I}u_i\Big)\bullet_{\mathsf{Q}}\Gamma_{\bullet_{\mathsf{Q}}}^W\Big(\bigotimes_{j\notin I}u_j\Big)\right](f_1\otimes\ldots\otimes f_n;\eta,\bar{\eta}).
\end{align}
\item for any $(\zeta,\bar{\zeta})\in\Gamma(DM\oplus D^\ast M)$, denoting with $\delta_\zeta$ (resp. $\delta_{\bar{\zeta}}$) the functional derivative in the direction $\zeta$ (resp. $\bar{\zeta}$) - see Definition \ref{Def: functional derivatives} - , the following properties hold true:
\begin{align}
	&\Gamma_{\bullet_{\mathsf{Q}}}^W(u)=u,\quad\forall u\in\mathcal{A}_{\cdot_{\mathsf{Q}}}^W,\\
	&\Gamma_{\bullet_{\mathsf{Q}}}^W\circ \delta_\zeta=\delta_\zeta\circ \Gamma_{\bullet_{\mathsf{Q}}}^W,\qquad \Gamma_{\bullet_{\mathsf{Q}}}^W\circ \delta_{\bar{\zeta}}=\delta_{\bar{\zeta}}\circ \Gamma_{\bullet_{\mathsf{Q}}}^W,\\
	&\Gamma_{\bullet_{\mathsf{Q}}}^W(u_1\otimes\ldots\otimes G_\psi\circledast u_l\otimes\ldots\otimes u_n)=(\delta_{\mathrm{Diag}_2}^{l-1}\otimes G_\psi\otimes\delta_{\mathrm{Diag}_2}^{n-l})\circledast\Gamma_{\bullet_{\mathsf{Q}}}^W(u_1\otimes\ldots\otimes u_l\otimes\ldots u_n)\\
	&\Gamma_{\bullet_{\mathsf{Q}}}^W(u_1\otimes\ldots\otimes G_{\bar{\psi}}\circledast u_l\otimes\ldots\otimes u_n)=(\delta_{\mathrm{Diag}_2}^{l-1}\otimes G_{\bar{\psi}}\otimes\delta_{\mathrm{Diag}_2}^{n-l})\circledast\Gamma_{\bullet_{\mathsf{Q}}}^W(u_1\otimes\ldots\otimes u_l\otimes\ldots u_n),
\end{align} 
for any $u_1\ldots, u_n\in\mathcal{A}_{\cdot_{\mathsf{Q}}}^W$, $n\in\mathbb{N}$.
\end{enumerate}
\end{proposition}

The deformation map $\Gamma_{\bullet_{\mathsf{Q}}}^W$ introduced in Proposition \ref{Prop: nonlocal deformation map} is the building block to define a new algebraic structure out of $\mathcal{T}^W(\mathcal{A}_{\cdot_{\mathsf{Q}}}^W)$, as stated in the next theorem.
\begin{theorem}\label{Thm: nonlocal deformed algebra}
	Given a linear map $\Gamma_{\bullet_{\mathsf{Q}}}^W:\mathcal{T}^W(\mathcal{A}_{\cdot_{\mathsf{Q}}}^W)\rightarrow (\mathcal{T}^W_C)'(M;\mathsf{Pol}_W)$ as per Proposition \ref{Prop: nonlocal deformation map} we define
	\begin{equation}
		\mathcal{A}_{\bullet_{\mathsf{Q}}}^W:=\Gamma_{\bullet_{\mathsf{Q}}}^W(\mathcal{T}^W(\mathcal{A}_{\cdot_{\mathsf{Q}}}^W))\subseteq (\mathcal{T}^W_C)'(M;\mathsf{Pol}_W). 
	\end{equation}
In addition, let $\bullet_{\Gamma_{\bullet_{\mathsf{Q}}}^W}:\mathcal{A}_{\bullet_{\mathsf{Q}}}^W\times \mathcal{A}_{\bullet_{\mathsf{Q}}}^W\rightarrow \mathcal{A}_{\bullet_{\mathsf{Q}}}^W$ be such that
\begin{equation}\label{Eq: bullet product 2}
	u_1\bullet_{\Gamma_{\bullet_{\mathsf{Q}}}^W} u_2:=\Gamma_{\bullet_{\mathsf{Q}}}^W\Big({(\Gamma_{\bullet_{\mathsf{Q}}}^W)}^{ -1}(u_1)\otimes{(\Gamma_{\bullet_{\mathsf{Q}}}^W)}^{-1}(u_2)\Big),\qquad \forall u_1,u_2\in\mathcal{A}_{\bullet_{\mathsf{Q}}}^W.
\end{equation}
It descends that $(\mathcal{A}_{\bullet_{\mathsf{Q}}}^W, \bullet_{\Gamma_{\bullet_{\mathsf{Q}}}^W})$ is a unital, associative algebra over $W$.
\end{theorem}
\begin{proof}
	For the sake of brevity, we omit the proof of this theorem. Nonetheless, all the details for the scalar counterpart can be found in \cite[Thm. 4.4]{Dappiaggi:2020gge}. We outline that, being the proof merely algebraic, the
	translation of all the steps to the vector-valued scenario is straightforward.
\end{proof}

\begin{example}
	To convince the reader that the deformed product $\bullet_{\Gamma_{\bullet_{\mathsf{Q}}}^W}$ of Theorem \ref{Thm: nonlocal deformed algebra} codifies the correct information, it is instructive to compute the two-point correlation function between the random fields $\widehat{\psi}$ and $\widehat{\bar{\psi}}$ defined in Equation \eqref{Eq: random fields}. As a preliminary step, recalling that $\Phi \in \mathcal{M}_{1,0} \subset \mathcal{D}'_C(M; \text{Pol}_W)$, $\bar{\Phi} \in \mathcal{M}_{0,1} \subset \mathcal{D}'_C(M; \text{Pol}_{W})$ are the polynomial functional-valued vector distributions introduced in Example \ref{Ex: basic functionals},  item 1. in Proposition \ref{Prop: deformation map} entails that $\Gamma_{\cdot_{\mathsf{Q}}}$ acts as the identity. Hence
	\begin{align}\label{Eq: conto}
		\Gamma^W_{\bullet_{\mathsf{Q}}}((\Gamma^W_{\bullet_{\mathsf{Q}}})^{-1}(\Phi) \otimes (\Gamma^W_{\bullet_{\mathsf{Q}}})^{-1}(\bar{\Phi})) &= \Gamma^W_{\bullet_{\mathsf{Q}}}[(\Gamma^W_{\bullet_{\mathsf{Q}}})^{-1}(\Gamma^W_{\bullet_{\mathsf{Q}}}(\Phi)) \otimes (\Gamma^W_{\bullet_{\mathsf{Q}}})^{-1}(\Gamma^W_{\bullet_{\mathsf{Q}}}(\bar{\Phi}))] \\ \nonumber &= \Gamma^W_{\bullet_{\mathsf{Q}}}(\Phi \otimes \bar{\Phi}) = \Phi \bullet_{\mathsf{Q}} \bar{\Phi},
	\end{align}
where $\bullet_{\mathsf{Q}}$ is the deformation of the tensor product defined in Equation \eqref{Eq: bullet product}. Putting Equations \eqref{Eq: bullet product 2} and \eqref{Eq: conto} together, we obtain
\begin{flalign}\label{exbul}
	\Phi \bullet_{\Gamma^W_{\bullet{\mathsf{Q}}}} \bar{\Phi} (f_1 \otimes f_2; \eta,\bar{\eta}) &= \Phi \bullet_{\mathsf{Q}} \bar{\Phi}(f_1 \otimes f_2; \eta,\bar{\eta}) \\ \notag &= \Phi \otimes \bar{\Phi} (f_1 \otimes f_2; \eta,\bar{\eta}) + Q(f_1 \otimes f_2), 
\end{flalign}
for all $f_1, f_2 \in \mathcal{D}(M)$ and for any field configuration $(\eta,\bar{\eta}) \in \Gamma(DM\oplus D^\ast M)$. As expected, evaluating Equation \eqref{exbul} at the vanishing configuration $\psi=0$ and imposing the constraint $\bar{\psi}=\psi^\dagger\gamma_0$ we recover the two-point function calculated in Equation \eqref{Eq: 2-point function 1}, namely
\begin{equation}
		\Phi \bullet_{\Gamma^W_{\bullet_{\mathsf{Q}}}} \bar{\Phi} (f_1 \otimes f_2; 0, 0)=Q(f_1 \otimes f_2)\equiv\omega_2(f_1 \otimes f_2), \, \, \forall f_1, f_2 \in \mathcal{D}(M).
\end{equation}
 An analogous computation entails 
\begin{equation*}
	\bar{\Phi} \bullet_{\Gamma^W_{\bullet_{\mathsf{Q}}}} \Phi(f_1 \otimes f_2; 0,0) = \widetilde{Q}(f_1 \otimes f_2) \equiv \widetilde{\omega_2}(f_1 \otimes f_2), \, \, \forall f_1, f_2 \in \mathcal{D}(M), 
\end{equation*}
while $\Phi\bullet_{\Gamma^W_{\bullet_{\mathsf{Q}}}}\Phi(f_1 \otimes f_2; 0, 0)=\bar{\Phi}\bullet_{\Gamma^W_{\bullet{\mathsf{Q}}}}\bar{\Phi}(f_1 \otimes f_2; 0, 0)=0$, in agreement with the stochastic nature of the underlying random fields. 
\end{example}

To conclude this section, we discuss uniqueness of the linear map $\Gamma_{\bullet_{\mathsf{Q}}}^W$ and, as a by-product, of $\mathcal{A}_{\bullet_{\mathsf{Q}}}^W$ introduced in Theorem \ref{Thm: nonlocal deformed algebra}. This problem can be tackled with an approach similar to that illustrated in the Section \ref{Sec: local deformation}. Hence, we can state the following result, whose main significance resides in the fact that it codifies the renormalization freedoms in defining the deformed algebra $\mathcal{A}^W_{\bullet_{\mathsf{Q}}}$.

\begin{theorem}\label{Thm: uniqueness bullet}
Assume that $\Gamma^W_{\bullet_{\mathsf{Q}}}, (\Gamma^W_{\bullet_{\mathsf{Q}}})': \mathcal{A}^W_{\bullet_{\mathsf{Q}}} \rightarrow (\mathcal{T}^W_C)'(M; \mathsf{Pol}_W)$ are two distinct linear maps which satisfy the properties enumerated in Theorem \ref{Prop: nonlocal deformation map}. Furthermore, let $\underline{r} := (r_1, ..., r_i, ...) \in \mathbb{N}_0^{\mathbb{N}_0}$ and, analogously, $\underline{r}' := (r'_1, ..., r'_i, ...) \in\mathbb{N}_0^{\mathbb{N}_0}$. Then, there exists a collection of linear maps $\{C_{\underline{r}, \underline{r}'}\}_{\underline{r}, \underline{r}' \in \mathbb{N}_0^{\mathbb{N}_0}}$ with $C_{\underline{r}, \underline{r}'} : \mathcal{T}^W(\mathcal{A}^W_{\cdot_{\mathsf{Q}}}) \rightarrow \mathcal{T}^W(\mathcal{A}^W_{\cdot_{\mathsf{Q}}})$ - where $\mathcal{T}^W(\mathcal{A}^W_{\cdot_{\mathsf{Q}}})$ is characterized in Equation \eqref{Eq: graded decomposition tensor algebra} - which meets the following requirements:
		\begin{itemize}
			\item [1.] for any $j \in \mathbb{N}_0$,
			\begin{equation*}
				C_{\underline{r}, \underline{r}'} [((\mathcal{A}^W_{\cdot_{\mathsf{Q}}})^{\otimes j}] \subseteq \mathcal{M}_{r_1, r'_1} \otimes ..\otimes \mathcal{M}_{r_j, r'_j}, 
			\end{equation*}
			whilst
			\begin{equation*}
				C_{\underline{l}, \underline{l}'} [\mathcal{M}_{r_1, r'_1} \otimes ..\otimes \mathcal{M}_{r_j, r'_j}] = 0, 
			\end{equation*}
			whenever $r_i \le l_i -1$ or $r'_i \le l'_i - 1$ for some index $i \in \{1, ..., j\}$. 
			\item [2.] for every $j \in \mathbb{N}_0$ and for all $u_1, ..., u_j \in \mathcal{A}^W_{\cdot_{\mathsf{Q}}}$, it holds that 
			\begin{flalign*}
				C_{\underline{r}, \underline{r}'}[u_1 \otimes ... &\otimes G_{\psi} \circledast u_k \otimes ... \otimes ... \otimes u_j] =\\ &= (\delta_{\mathrm{Diag}_2}^{\otimes (k-1)} \otimes G_{\psi} \otimes \delta_{\mathrm{Diag}_2}^{\otimes (j-k)}) \circledast C_{\underline{r}, \underline{r}'}[u_1 \otimes ... \otimes u_j], \\
				C_{\underline{r}, \underline{r}'}[u_1 \otimes ... &\otimes G_{\bar{\psi}} \circledast u_k \otimes ... \otimes ... \otimes u_j] = \\ &= (\delta_{\mathrm{Diag}_2}^{\otimes (k-1)} \otimes G_{\bar{\psi}} \otimes \delta_{\mathrm{Diag}_2}^{\otimes (j-k)}) \circledast C_{\underline{r}, \underline{r}'}[u_1 \otimes ... \otimes u_j]. \\
			\end{flalign*}
			In addition,
			\begin{flalign*}
				& \delta_{\zeta} C_{\underline{r}, \underline{r}'} [u_1 \otimes ... \otimes u_j] = \sum_{p=1}^j C_{\underline{r}(p), \underline{r}'} [u_1 \otimes ... \otimes \delta_{\zeta} u_p \otimes ... \otimes u_j],\\
				&\delta_{\bar{\zeta}} C_{\underline{r}, \underline{r}'} [u_1 \otimes ... \otimes u_j] = \sum_{p=1}^j C_{\underline{r}, \underline{r}'(p)} [u_1 \otimes ... \otimes \delta_{\bar{\zeta}} u_p \otimes ... \otimes u_j],
			\end{flalign*}
			where 
			$\delta_{\zeta}$ and $\delta_{\bar{\zeta}}$ denote the functional derivatives in the directions $\zeta \in \Gamma(DM)$ and $\bar{\zeta} \in \Gamma(D^*M)$ respectively, as per Definition \ref{Def: functional derivatives}, whereas 
			\begin{equation*}
				[\underline{r}(p)]_i =
				\begin{cases}
					r_i \, \, \text{if} \, \, i \ne p \\
					r_i - 1 \, \, \text{if} \, \, i = p
				\end{cases}. 
			\end{equation*}
			Furthermore, we recall that $G_{\psi}$ and $G_{\bar{\psi}}$ denote the fundamental solutions of the Dirac operator $\slashed{D}$ and of its adjoint $\slashed{D}^*$. 
			\item [3. ] given two linear maps $\Gamma^W_{\cdot_{\mathsf{Q}}},(\Gamma^W_{\cdot_{\mathsf{Q}}})' : \mathcal{A}^W \rightarrow \mathcal{D}'_C(M; \text{Pol}_W)$ satisfying the requirements of Proposition \ref{Prop: deformation map}, it holds that 
			\begin{flalign*}
				(\Gamma^W_{\bullet_{\mathsf{Q}}})' &([(\Gamma^W_{\cdot_{\mathsf{Q}}})']^{\otimes j} (u_{r_1, r'_1} \otimes ... \otimes u_{r_j, r'_j}))(f_1 \otimes ... \otimes f_j) \\ &= \Gamma^W_{\bullet_{\mathsf{Q}}}([\Gamma^W_{\cdot_{\mathsf{Q}}}]^{\otimes j} (u_{r_1, r'_1} \otimes ... \otimes u_{r_j, r'_j}))(f_1 \otimes ... \otimes f_j) \\
				&+ \sum_{\mathcal{I} \in \mathcal{P}(1,...,j)} \Gamma^W_{\bullet_{\mathsf{Q}}} \left[ ([\Gamma^W_{\cdot_{\mathsf{Q}}}]^{\otimes |\mathcal{I}|} C_{\underline{r}_{\mathcal{I}}, \underline{r}'_{\mathcal{I}}} \left( \bigotimes_{I \in \mathcal{I}} \prod_{i \in I} u_{r_i, r'_i} \right)\right] \left( \bigotimes_{I \in \mathcal{I}} \prod_{i \in I} f_i \right), 
			\end{flalign*}
			for all $u_{r_1, r'_1}, ..., u_{r_j, r'_j} \in \mathcal{A}^W$, $u_{r_i, r'_i} \in \mathcal{M}_{r_i, r'_i}$, $\forall i \in \{1, ..., j\}$ and for any $f_1, ..., f_j \in \mathcal{D}(M)$. In addition, $\mathcal{P}(1,...,j)$ is the family of all admissible partitions of the set $\{1,..., j\}$ into non-empty disjoint subsets, whilst $\underline{r}_{\mathcal{I}} = (\underline{r}_{I})_{I \in \mathcal{I}} := (\sum_{i \in I} r_i)_{I \in \mathcal{I}}$. 
		\end{itemize}
\end{theorem}
\begin{proof}
	The proof is omitted since it descends, barring some minor modifications, from that illustrated in \cite[Thm. 5.6]{Dappiaggi:2020gge}.
\end{proof}

\section{Application to the stochastic Thirring Model}\label{Sec: Thirring}
Having discussed how to implement at the level of algebras of functionals the stochastic information carried by expectation values and correlation functions of the random fields, see Sections \ref{Sec: local deformation} and \ref{Sec: nonlocal deformation algebra} respectively, we apply the framework to the analysis of the stochastic Thirring model on $M=\mathbb{R}^2$ endowed with the Euclidean metric $\delta$. As mentioned in the introduction, we stress once more that the model under scrutiny is not directly connected to the stochastic quantization programme, as considered in \cite{Albeverio,DeVecchi}. Denoting with $\xi^\rho$ and $\bar{\xi}_\rho$ the components of the white noise as per Definition \ref{Def: fermionic white noise}, the stochastic equations of our interest are
\begin{equation}\label{Eq: stochastic thirring model}
	\begin{cases}
		\slashed{D}^{\rho}_{\sigma} \psi^{\sigma} - \lambda (\bar{\psi} \gamma^{\mu} \psi)(\gamma_{\mu})_{\sigma}^{\rho}
		\psi^{\sigma} = \xi^{\rho} \\
		\slashed{D}^{*\, \sigma}_{\rho} \bar{\psi}_{\sigma} - \lambda (\bar{\psi} \gamma^{\mu} \psi) \bar{\psi}_{\sigma} (\gamma_{\mu})_{\rho}^{\sigma} = \bar{\xi}_{\rho} \\
	\end{cases},
\end{equation}
where $\slashed{D} := (i \slashed{\partial} -m)$ and $\slashed{D}^* := (i \slashed{\partial} +m)$ denote the massive Dirac operator and its formal adjoint. In addition $\psi^{\rho}, \bar{\psi}_{\rho}: \mathbb{R}^2 \rightarrow \mathbb{C}^2$ are respectively the spinor and cospinor fields in the two-dimensional scenario, while $\lambda \in \mathbb{R}_{+}$ denotes the coupling constant. 
\begin{remark}\label{Rem: spinor_to_cospinor}
	While in the preceding investigation, the r\^{o}le of the spinor and of the cospinor degrees of freedom were kept independent, in the following we need to identify $\bar{\psi}_\lambda\in\Gamma(D^\ast \mathbb{R}^2)\simeq C^\infty(\mathbb{R}^2;\mathbb{C}^2)$ with  $\bar{\psi}_\lambda=\psi^\dagger_{\lambda'}(\gamma_0)_\lambda^{\lambda'}$ where $\gamma_0$ has been introduce in Section \ref{Sec: introduction}. For consistency such correspondence holds true also at the level of  white noise, namely we require that $\bar{\xi}_\lambda=\xi^\dagger_{\lambda'}(\gamma_0)^{\lambda'}_\lambda$. In view of the defining properties of the centered Gaussian random distributions as per Definition \ref{Def: fermionic white noise}, it follows that $\xi^\lambda$ and $\xi^\dagger_{\lambda'}$ are correlated only when $\lambda\neq \lambda'$. In particular
	\begin{align*}
		&\mathbb{E}[\xi^\lambda(x)\xi^\dagger_{\lambda'}(y)]=a^\lambda_{\lambda'}\delta(x-y),\\
		&\mathbb{E}[\xi^\dagger_{\lambda'}(x)\xi^\lambda(y)]=-a^\lambda_{\lambda'}\delta(x-y),
	\end{align*}
where 
\begin{equation*}
	 a^{\lambda}_{\lambda'} := 1- \delta^{\lambda}_{\lambda'} = 
	\begin{cases}
		1 & \text{if $\lambda \ne \lambda'$} \\
		0 & \text{if $\lambda = \lambda'$}
	\end{cases}.
\end{equation*}
\end{remark}
The first step towards the application of the formalism devised in the preceding sections consists of reformulating Equation \eqref{Eq: stochastic thirring model} in the language of functional-valued vector distributions. This can be achieved by replacing the random distributions $\psi^\lambda$ and $\bar{\psi}_\lambda$ appearing on the left hand-side of Equation \eqref{Eq: stochastic thirring model} with polynomial functional-valued
vector distributions $\Psi^\lambda\in\mathcal{D}'_C(\mathbb{R}^2;Pol_W)$ and $\bar{\Psi}_\lambda\in\mathcal{D}'_C(M;Pol_{W})$. Hence, recalling that $G_\psi$ (resp. $G_{\bar{\psi}}$) denotes the fundamental solution of the Dirac operator (resp. of its formal adjoint), see Equation \eqref{Eq: fundamental solution Dirac}, Equation \eqref{Eq: stochastic thirring model} translates to the following identities at the level of functionals:
\begin{equation}\label{Eq: functional thirring}
	\begin{cases}
		\Psi^{\rho} = \Phi^{\rho} + \lambda (G_{\psi})^{\rho}_{\rho'} \circledast [(\bar{\Psi} \gamma_{\mu} \Psi) (\gamma^{\mu})^{\rho'}_{\rho''} \Psi^{\rho''}]\\
		\bar{\Psi}_{\rho} = \bar{\Phi}_{\rho} + \lambda (G_{\bar{\psi}})_{\rho}^{\rho'} \circledast [(\bar{\Psi} \gamma_{\mu} \Psi) \bar{\Psi}_{\rho''} (\gamma^{\mu})_{\rho'}^{\rho''}]
	\end{cases},
\end{equation}
where $\Phi^\rho\in\mathcal{D}'(\mathbb{R}^2;\mathrm{Pol}_W)$ and $\bar{\Phi}_\rho\in\mathcal{D}'(\mathbb{R}^2;Pol_{W})$ as per Example \ref{Ex: basic functionals} are the functional counterparts of $\widehat{\varphi}^\lambda:=(G_{\psi})^\lambda_{\lambda'}\circledast \xi^{\lambda'}$ and $\widehat{\bar{\varphi}}_\lambda:=(G_{\bar{\psi}})_\lambda^{\lambda'}\circledast \bar{\xi}_{\lambda'}$, respectively. 

In the following we are interested in constructing a perturbative solution of Equation \eqref{Eq: functional thirring}, namely we write 
\begin{equation}
	\label{Eq: perturbative expansion}
	\begin{cases}
		\Psi^{\rho} [[\lambda]] = \sum_{k \ge 0} \lambda^k F^{\rho}_k \in \mathcal{A}^V [[\lambda]], \, \, \, \, \, \, \, \, \, F^{\rho}_k \in \mathcal{A}^V,\\
		\bar{\Psi}_{\rho} [[\lambda]] = \sum_{k \ge 0} \lambda^k (\widetilde{F}_{\rho})_{k} \in \mathcal{A}^{V^*}[[\lambda]], \, \, (\widetilde{F}_{\rho})_{k} \in \mathcal{A}^{V^*},
	\end{cases}
\end{equation}
where $\mathcal{A}^{V},\mathcal{A}^{V^*}$ are as per Definition \ref{Def: A} with $V\simeq V^*\simeq\mathbb{C}^2$.
The coefficients $F^{\rho}_k$ and $(\widetilde{F}_{\rho})_{k}$ are completely determined by an iterative procedure on Equation \eqref{Eq: functional thirring} or, in a slightly more systematic way, as
\begin{equation}
	F^{\rho}_k=\frac{1}{k!}\frac{d^k}{d^k\lambda}\Psi^\rho\Big\vert_{\lambda=0},\qquad (\widetilde{F}_{\rho})_{k}=\frac{1}{k!}\frac{d^k}{d^k\lambda}\bar{\Psi}_\rho\Big\vert_{\lambda=0}.
\end{equation}
The explicit expressions for the lower order coefficients read
\begin{flalign}
	\label{F0}
	&F^{\rho}_0 = \Phi^{\rho}, \, \, \, (\widetilde{F}_{\rho})_0 = \bar{\Phi}_{\rho},\\
	\label{F1}
	&F^{\rho}_1 = (G_{\psi})^{\rho}_{\rho'} \circledast [(\bar{\Phi} \gamma_{\mu} \Phi) (\gamma^{\mu})^{\rho'}_{\rho''} \Phi^{\rho''}], \, \, \, (\widetilde{F}_{\rho})_1 = (G_{\bar{\psi}})_{\rho}^{\rho'} \circledast [(\bar{\Phi} \gamma_{\mu} \Phi) \bar{\Phi}_{\rho''} (\gamma^{\mu})_{\rho'}^{\rho''}],\\
	\label{F2}
	&F^{\rho}_2 = (G_{\psi})^{\rho}_{\rho'} \circledast [(\bar{\Phi} \gamma_{\mu} \Phi) (\gamma^{\mu})^{\rho'}_{\rho''} F^{\rho''}_1 + (\widetilde{F}_1 \gamma_{\mu} \Phi) (\gamma^{\mu})^{\rho'}_{\rho''} \Phi^{\rho''} + (\bar{\Phi} \gamma_{\mu} F_1) (\gamma^{\mu})^{\rho'}_{\rho''} \Phi^{\rho''}], \\ \notag
	&(\widetilde{F}_{\rho})_2 = (G_{\bar{\psi}})_{\rho}^{\rho'} \circledast [(\bar{\Phi} \gamma_{\mu} \Phi) (\widetilde{F}_{\rho''})_1 (\gamma^{\mu})_{\rho'}^{\rho''} + (\widetilde{F}_1 \gamma_{\mu} \Phi) \bar{\Phi}_{\rho''} (\gamma^{\mu})_{\rho'}^{\rho''}  + (\bar{\Phi} \gamma_{\mu} F_1) \bar{\Phi}_{\rho''}(\gamma^{\mu})^{\rho'}_{\rho''}].
\end{flalign}
The higher order coefficients are completely specified by the lower order ones. By an induction argument, we can deduce that the $k$-th order perturbative coefficients are specified by
\begin{flalign}
	\label{Fk}
	F^{\rho}_k &= \sum_{ \substack{k_1, k_2, k_3 \in \mathbb{N}_0 \\ k_1 + k_2 + k_3 = k-1} } (G_{\psi})^{\rho}_{\rho'} \circledast [(\widetilde{F}_{k_1} \gamma_{\mu} F_{k_2}) (\gamma^{\mu})^{\rho'}_{\rho''} F^{\rho''}_{k_3}], \, \, \, \text{for $k \ge 3$}\,, \\
	\label{tildeFk}
	(\widetilde{F}_{\rho})_{k} &= \sum_{ \substack{k_1, k_2, k_3 \in \mathbb{N}_0 \\ k_1 + k_2 + k_3 = k-1} } (G_{\bar{\psi}})_{\rho}^{\rho'} \circledast [(\widetilde{F}_{k_1} \gamma_{\mu} F_{k_2}) (\widetilde{F}_{\rho''})_{k_3} (\gamma^{\mu})^{\rho''}_{\rho'}], \, \, \, \text{for $k \ge 3$}\,. 
\end{flalign} 
Having established the form of the perturbative expansion of the solutions, we are in a position to use the deformation maps $\Gamma_{\cdot_{\mathsf{Q}}}^W$ and $\Gamma_{\bullet_{\mathsf{Q}}}^W$ introduced in Propositions \ref{Prop: deformation map} and \ref{Prop: nonlocal deformation map} to compute the expectation values and correlation functions of the solutions at any order in perturbation theory.

\subsection{Expectation values}\label{Sec: expectation values}
As a warm up we compute the expectation values of $\Psi^{\rho}$ and $\bar{\Psi}_{\rho}$ at first order in the expansion with respect to the coupling constant $\lambda > 0$. To this avail, let us consider the perturbative solution $\Psi^{\rho}[[\lambda]]$ as per Equation \eqref{Eq: perturbative expansion} and, exploiting the linearity of $\Gamma_{\cdot_{\mathsf{Q}}}^W$, see Proposition \ref{Prop: deformation map}, we define
\begin{equation}\label{Eq: deformed solution}
	\Psi^{\rho}_{\cdot_{\mathsf{Q}}} [[\lambda]] := \Gamma^W_{\cdot_{\mathsf{Q}}} (\Psi^{\rho} [[\lambda]]) = \sum_{k} \lambda^{k} \Gamma^W_{\cdot_{\mathsf{Q}}} (F^{\rho}_k) \in \mathcal{A}^W_{\cdot Q} [[\lambda]]\,, 
\end{equation}
being $W: = \bigotimes_{k,k' \ge 0} (V^{\otimes k} \oplus (V^*)^{\otimes k'})$ where $V\equiv\Sigma_2\simeq\mathbb{C}^2$. 
Combining Equations \eqref{Eq: perturbative expansion} and \eqref{F1} with the properties enjoyed by $\Gamma_{\cdot_{\mathsf{Q}}}^W$, see Proposition \ref{Prop: deformation map}, it descends that
\begin{flalign}\label{Eq: expectation value solution}
	\mathbb{E} [\hat{\psi}^{\rho}_{\eta} [|\lambda |] (f)] &= \Gamma^{W}_{\cdot_\mathsf{Q}} (\Psi^{\rho} [[\lambda]]) (f; \eta, \bar{\eta}) \\ \notag &= \Gamma^{W}_{\cdot_\mathsf{Q}} (\Phi^{\rho})(f; \eta, \bar{\eta}) + \lambda \Gamma^{W}_{\cdot_\mathsf{Q}} ((G_{\psi})^{\rho}_{\rho'} \circledast [(\bar{\Phi} \gamma_{\mu} \Phi) (\gamma^{\mu})^{\rho'}_{\rho''} \Phi^{\rho''}]) (f; \eta, \bar{\eta}) + \mathcal{O}(\lambda^2) \\ \notag &
	= \Phi^{\rho} (f; \eta, \bar{\eta}) + \lambda (G_{\psi})^{\rho}_{\rho'} \circledast \Gamma^{W}_{\cdot_\mathsf{Q}} [(\bar{\Phi} \gamma_{\mu} \Phi) (\gamma^{\mu})^{\rho'}_{\rho''} \Phi^{\rho''}] (f; \eta, \bar{\eta}) + \mathcal{O}(\lambda^2). 
\end{flalign}
With $\mathbb{E}[\hat{\psi}^{\rho}_{\eta}[|\lambda |]]$ we denote the $\eta$-shifted expectation of the solution, namely considering an $\eta$-shifted Gaussian random field $\widehat{\varphi}^{\rho}_{\eta} := (G_{\psi})^{\rho}_{\rho'} \circledast \xi^{\rho} + \eta^{\rho}$ as initial condition. To recover the centred character of white noise, it suffices to evaluate $\Gamma^{W}_{\cdot_\mathsf{Q}} (\Psi^{\rho} [[\lambda]]) (f; \eta, \bar{\eta})$ at $\eta=0$. This forces also $\bar{\eta}$ to vanish on account of the constraint $\bar{\eta}_\lambda=\eta^\dagger_{\lambda'}(\gamma_0)_\lambda^{\lambda'}$. It follows that
\begin{flalign*}
	\Gamma^W_{\cdot_\mathsf{Q}} [(\bar{\Phi} \gamma_{\mu} \Phi) &(\gamma^{\mu})^{\rho'}_{\rho''} \Phi^{\rho''}] (f; \eta, \bar{\eta}) = [\Gamma^W_{\cdot_\mathsf{Q}}(\bar{\Phi} \gamma_{\mu} \Phi) \cdot_{\mathsf{Q}} (\gamma^{\mu})^{\rho'}_{\rho''} (\Phi^{\rho''})] (f; \eta, \bar{\eta}) \\ &=  [(\bar{\Phi} \gamma_{\mu} \Phi) (\gamma^{\mu})^{\rho'}_{\rho''} \Phi^{\rho''}] (f; \eta, \bar{\eta}) + (\delta_{\mathrm{Diag}_2} \otimes \widetilde{Q}) \cdot (\delta_{\mathrm{Diag}_2} \Phi \tilde{\otimes} \delta_{\mathrm{Diag}_2})(f \otimes 1_3; \eta, \bar{\eta}) \\ &=  [(\bar{\Phi} \gamma_{\mu} \Phi) (\gamma^{\mu})^{\rho'}_{\rho''} \Phi^{\rho''}] (f; \eta, \bar{\eta}) + \Phi^{\rho'} \widetilde{C}^{\rho}_{\rho'} (f; \eta, \bar{\eta})\,,
\end{flalign*}
being $\widetilde{C} \in \mathcal{E}(\mathbb{R}^2; \mathbb{C}^2 \oplus \mathbb{C}^2)$ a smooth function with integral kernel specified by 
\begin{equation}
	\label{tildeC}
	\widetilde{C} (x) = \chi(x) \, \widetilde{G_{\psi} \cdot G_{\bar{\psi}}} \, (\delta_x \otimes \chi),
\end{equation}
where $\chi \in C^{\infty}_0(\mathbb{R}^2)$ is an arbitrary cutoff function necessary to avoid infrared divergences. At the same time $\widetilde{G_{\psi} \cdot G_{\bar{\psi}}} \in \mathcal{D}'(\mathbb{R}^2 \times \mathbb{R}^2; W)$ is any but fixed extension of the ill-defined product $G_{\psi} \cdot G_{\bar{\psi}} \in \mathcal{D}'(\mathbb{R}^2 \times \mathbb{R}^2 \setminus \mathrm{Diag}_2; W)$ to the thin diagonal of $\mathbb{R}^2$, see the proof of Proposition \ref{Prop: deformation map}. Focusing once more on Equation \eqref{Eq: expectation value solution}, the result reads
\begin{equation}
	\label{Eq: expectation value final}
	\mathbb{E} [\hat{\psi}^{\rho}_{\eta}[|\lambda |](f)] = \Phi^{\rho} (f; \eta, \bar{\eta}) + \lambda (G_{\psi})^{\rho}_{\rho'} \circledast \{[(\bar{\Phi} \gamma_{\mu} \Phi) (\gamma^{\mu})^{\rho'}_{\rho''} \Phi^{\rho''}] + \Phi \widetilde{C} \} (f; \eta, \bar{\eta}) + \mathcal{O}(\lambda^2)\,.
\end{equation}
An identical line of reasoning entails that
\begin{align}\label{Eq: expectation barred solution}
	\mathbb{E} [(\hat{\bar{\psi}}_{\rho})_{\eta} [|\lambda |] (f)] = \bar{\Phi}_{\rho} (f; \eta, \bar{\eta}) + \lambda (G_{\bar{\psi}})_{\rho}^{\rho'} \circledast \{[(\bar{\Phi} \gamma_{\mu} \Phi) \bar{\Phi}_{\rho''}(\gamma^{\mu})_{\rho'}^{\rho''}] + \Phi C \} (f; \eta, \bar{\eta}) + \mathcal{O}(\lambda^2)\,, 
\end{align}
where $C \in \mathcal{E}(\mathbb{R}^2; \mathbb{C}^2 \oplus \mathbb{C}^2)$ is such that
\begin{equation}
	\label{C}
	C(x) := \chi(x) \widetilde{(G_{\bar{\psi}} \cdot G_{\psi})} (\delta_x \otimes \chi),   
\end{equation}
being $\chi \in C^{\infty}_0(\mathbb{R}^2)$ an arbitrary cutoff function. A direct inspection of Equations \eqref{Eq: expectation value final} and \eqref{Eq: expectation barred solution} entails that evaluating at $\eta=\bar{\eta}=0$, it descends that
\begin{equation}\label{Eq: 2nd order 0}
	\begin{cases}
		\mathbb{E} [\widehat{\psi}^{\rho}_0 [[\lambda]] ] = \mathcal{O}(\lambda^2) \\
		\mathbb{E} [(\widehat{\bar{\psi}}_{0})_\rho [[\lambda]] ] = \mathcal{O}(\lambda^2)
	\end{cases}.
\end{equation}
This result is a direct consequence of the cubic form of the interaction in Equation \eqref{Eq: functional thirring}. The following theorem extends the outcome of Equation \eqref{Eq: 2nd order 0} to all perturbative orders. We enunciate it for $\Psi^{\rho}[|\lambda |]$ only, though the same conclusion can be drawn for $\bar{\Psi}_{\rho}[|\lambda |]$.
\begin{theorem}
	\label{expvalue0}
	Consider the perturbative solution $\Psi^{\rho}[|\lambda |] \in \mathcal{A}^W [[\lambda]]$ of Equation \eqref{Eq: functional thirring}. Then, for any $f \in \mathcal{D}(\mathbb{R}^2)$, it descends that
	\begin{equation*}
		\mathbb{E} [\hat{\psi}^{\rho}_0 [[\lambda]] (f) ] := \Gamma^W_{\cdot_{\mathsf{Q}}} (\psi^{\rho} [[\lambda]]) (f; 0, 0) = 0.
	\end{equation*}    
\end{theorem}
\begin{proof}
	On account of Equations \eqref{Eq: deformed product} and \eqref{eq1pr}, the application of $\Gamma^W_{\cdot_{\mathsf{Q}}}$ to a generic element $u \in \mathcal{A}^W$ yields a sum of a finite number of terms with a decreasing polynomial degree in $\Phi^{\rho}$ and $\bar{\Phi}_{\rho}$. To wit, since the computation involves contractions between pair of fields $\Phi^{\rho}$ and $\bar{\Phi}_{\rho}$, the polynomial degree of each addend in the sum is decreased by a multiple of $2$. Therefore, to prove the thesis, it suffices to show that each of the coefficients $F^{\rho}_k$ defined in Equation \eqref{Fk} belongs to $\mathcal{M}_{r,r'}$  with $r+r' = 2n +1$, $n \in \mathbb{N}$ - see Remark \ref{Rem: M}.
	Hence, let us introduce the vector space $\mathsf{O} \subset \mathcal{A}^W$ which contains all polynomial functional-valued vector distributions having an odd polynomial degree in the fields. In other words, $u \in \mathsf{O}$ if the $(k_1, k_2)$-th order derivative of $u$ as per Definition \ref{Def: functional derivatives} is such that  $u^{(k_1, k_2)} (f; 0, 0) = 0$ with $k_1 + k_2 = 2n, n \in \mathbb{N}$, for all $f \in \mathcal{D}(\mathbb{R}^2)$. The space $\mathsf{O}$ enjoys some notable properties which are of straightforward verification, \textit{i.e.}, 
	\begin{itemize}
		\item [a.] $\mathsf{O}$ is closed under the action of the deformation map $\Gamma^W_{\cdot_{\mathsf{Q}}}$, that is the action of $\Gamma^W_{\cdot_{\mathsf{Q}}}$ on any linear combination with matrix coefficients of elements in $\mathsf{O}$ still lies in $\mathsf{O}$,
		\item [b.] for all $u_1, u_2, u_3 \in \mathsf{O}$, also their pointwise product $u_1 u_2 u_3$ belongs to $\mathsf{O}$,
		\item [c.] for any $u \in \mathsf{O}$, then $\bar{u} \in \mathsf{O}$, where the bar entails that all spinor fields are mapped into cospinors and viceversa following Remark \ref{Rem: spinor_to_cospinor}.
	\end{itemize}
	In the case of interest, this translates into proving that $\Psi^{\rho} [[\lambda]] \in \mathsf{O}$. Equation \eqref{Eq: perturbative expansion} entails that this is equivalent to showing that $F^{\rho}_k \in \mathsf{O}$ for all $k \in \mathbb{N}$. As a preliminary step, let us observe that, setting $k=0$, it holds that $F^{\rho}_0 := \Phi^{\rho} \in \mathsf{O}$. 
	By an inductive argument on the index $k$, assuming that $F^{\rho}_j \in \mathsf{O}$ for any $j \le k$, we need to prove that also $F^{\rho}_{k+1} \in \mathsf{O}$. Using Equation \eqref{Fk}, it descends that 
	\begin{equation*}
		F^{\rho}_{k+1} = \sum_{\substack{k_1, k_2, k_3 \in \mathbb{N}_0 \\ k_1 + k_2 + k_3 = k}} (G_{\psi})^{\rho}_{\rho'} \circledast [(\widetilde{F}_{k_1} \gamma_{\mu} F_{k_2}) (\gamma^{\mu})^{\rho'}_{\rho''} F^{\rho''}_{k_3}].
	\end{equation*}
	Being the indices $k_1, k_2, k_3 \le k$, the inductive hypothesis implies that $F_{k_1}, F_{k_2}, F_{k_3} \in \mathsf{O}$. From property c. of $\mathsf{O}$, it follows that also $\widetilde{F}_{k_1} \in \mathsf{O}$, whilst property b. together with property a. entail that $(\widetilde{F}_{k_1} F_{k_2}) (\gamma^{\mu})^{\rho'}_{\rho''} F^{\rho''}_{k_3} \in \mathsf{O}$. The thesis is thus proven. 
\end{proof}

\subsection{Two-point correlation functions}\label{Sec: Two-point correlation functions}
Having characterized the expectation values of the perturbative solution of Equation \eqref{Eq: stochastic thirring model}, we investigate the structure of the correlation functions. Following the approach adopted in Section \ref{Sec: expectation values}, the deformed algebra of functionals  $\mathcal{A}_{\bullet_{\mathsf{Q}}}^W:=\Gamma_{\bullet_{\mathsf{Q}}}^W(\mathcal{T}^W(\mathcal{A}_{\cdot_{\mathsf{Q}}}^W))$ as per Theorem \ref{Thm: nonlocal deformed algebra} comes into play. Henceforth, to lighten the notation, we will omit the spinor and cospinor indices in the functional expressions of the fields. Hence, for any $f_1, f_2 \in \mathcal{D}(\mathbb{R}^2)$, we can write the two-point correlation function of the $\eta-$shifted random fields $\widehat{\psi}^{\rho}_{\eta}$ and $(\widehat{\bar{\psi}}_{\rho})_{\eta}$ with $\eta \in \Gamma(D \mathbb{R}^2) \simeq C^{\infty}(\mathbb{R}^2; \mathbb{C}^2)$ as follows
\begin{flalign}
	\label{2pointeta}
	\mathbb{E}[\widehat{\psi}^{\rho}_{\eta}[[\lambda]](f_1) (\widehat{\bar{\psi}}_{\rho})_{\eta}[[\lambda]](f_2)] &= (\Psi_{\cdot_{\mathsf{Q}}} [[\lambda]] \bullet_{\Gamma^W_{\bullet_{\mathsf{Q}}}} \bar{\Psi}_{\cdot_{\mathsf{Q}}} [[\lambda]])(f_1 \otimes f_2; \eta, \bar{\eta})\\ \notag & = \Gamma^W_{\bullet_{\mathsf{Q}}}[\Gamma^W_{\cdot_{\mathsf{Q}}}(\Psi [[\lambda]]) \otimes \Gamma^W_{\cdot_{\mathsf{Q}}}(\bar{\Psi} [[\lambda]])] (f_1 \otimes f_2; \eta, \bar{\eta}) \\ \notag &= \sum_{k \ge 0} \lambda^k \sum_{\substack{k_1, k_2 \ge 0 \\ k_1 + k_2 = k}} \Gamma^W_{\bullet_{\mathsf{Q}}}[\Gamma^W_{\cdot_{\mathsf{Q}}}(F_{k_1}) \otimes \Gamma^W_{\cdot_{\mathsf{Q}}}(\widetilde{F}_{k_2})] (f_1 \otimes f_2; \eta, \bar{\eta}),
\end{flalign}
where we employed the definition of $\bullet_{\Gamma_{\bullet_{\mathsf{Q}}}^W}$ as per Equation \eqref{Eq: bullet product}, as well as the perturbative decomposition of $\Psi\llbracket\lambda\rrbracket$ and $\bar{\Psi}\llbracket\lambda\rrbracket$, see Equation \eqref{Eq: deformed solution}. Up to the second order in $\lambda$, Equation \eqref{2pointeta} reads 
\begin{flalign}\label{Eq: omega2}
	\mathbb{E}[\widehat{\psi}^{\rho}_{\eta}[[\lambda]](f_1) (\widehat{\bar{\psi}}_{\rho})_{\eta}[[\lambda]](f_2)] &= \underbrace{\Gamma^W_{\bullet_{\mathsf{Q}}}[\Gamma^W_{\cdot_{\mathsf{Q}}}(\Phi) \otimes \Gamma^W_{\cdot_{\mathsf{Q}}}(\bar{\Phi})] (f_1 \otimes f_2; \eta, \bar{\eta})}_{\text{(a)}} \\ \notag &+ \lambda \underbrace{\Gamma^W_{\bullet_{\mathsf{Q}}}\{\Gamma^W_{\cdot_{\mathsf{Q}}}[G_{\psi} \circledast ((\bar{\Phi}\gamma_{\mu} \Phi) \gamma^{\mu} \Phi)] \otimes \Gamma^W_{\cdot_{\mathsf{Q}}}(\bar{\Phi})\} (f_1 \otimes f_2; \eta, \bar{\eta})}_{\text{(b)}} \\ \notag &+ \lambda \underbrace{\Gamma^W_{\bullet_{\mathsf{Q}}} \{ \Gamma^W_{\cdot_{\mathsf{Q}}}(\Phi) \otimes \Gamma^W_{\cdot_{\mathsf{Q}}}[G_{\bar{\psi}} \circledast ((\bar{\Phi}\gamma_{\mu} \Phi) \bar{\Phi} \gamma^{\mu})] \} (f_1 \otimes f_2; \eta, \bar{\eta})}_{\text{(c)}}+\mathcal{O}(\lambda^2),
\end{flalign} 
We observe that, since $\Phi \in \mathcal{M}_{1,0}$ and $\bar{\Phi} \in \mathcal{M}_{0,1}$ - see Remark \ref{Rem: M} -, item 1. of Proposition \ref{Prop: deformation map} entails that $\Gamma^W_{\cdot_{\mathsf{Q}}} (\Phi) = \Phi$ and $\Gamma^W_{\cdot_{\mathsf{Q}}} (\bar{\Phi}) = \bar{\Phi}$. Recalling the properties enjoyed by the map $\Gamma_{\bullet_{\mathsf{Q}}}^W$ in Proposition \ref{Prop: nonlocal deformation map}, we can compute explicitly all contributions
\begin{flalign*}
	\text{(a)} &:= \Gamma^W_{\bullet_{\mathsf{Q}}}(\Phi \otimes \bar{\Phi}) (f_1 \otimes f_2; \eta, \bar{\eta}) = (\Phi \bullet_{\mathsf{Q}} \bar{\Phi}) (f_1 \otimes f_2; \eta, \bar{\eta}) \\ &= \underbrace{(\Phi \otimes \bar{\Phi}) (f_1 \otimes f_2; \eta, \bar{\eta})}_{\text{(a1)}} + \underbrace{Q(f_1 \otimes f_2)}_{\text{(a2)}},
\end{flalign*} 
\begin{flalign*}
	\text{(b)} &:= \Gamma^W_{\bullet_{\mathsf{Q}}}[\Gamma^W_{\cdot_{\mathsf{Q}}}(G_{\psi} \circledast ((\bar{\Phi}\gamma_{\mu} \Phi) \gamma^{\mu} \Phi))] \otimes \bar{\Phi}] (f_1 \otimes f_2; \eta, \bar{\eta}) \\ &=  [(G_{\psi} \circledast ((\bar{\Phi} \gamma_{\mu} \Phi) \gamma^{\mu} \Phi) + G_{\psi} \circledast \Phi \widetilde{C}))  \bullet_{\mathsf{Q}} \bar{\Phi}] (f_1 \otimes f_2; \eta, \bar{\eta}) \\&= [(G_{\psi} \circledast ((\bar{\Phi} \gamma_{\mu} \Phi) \gamma^{\mu} \Phi + \Phi \widetilde{C})) \otimes \bar{\Phi}] (f_1 \otimes f_2; \eta, \bar{\eta}) \\ &+ Q \cdot (G_{\psi} \circledast (\bar{\Phi} \gamma_{\mu} \Phi) \gamma^{\mu} \otimes 1)(f_1 \otimes f_2; \eta, \bar{\eta}) \\ &+   2Q \cdot (G_{\psi} \circledast \bar{\Phi} \Phi \otimes 1) + Q \cdot (G_{\psi} \circledast \widetilde{C}1 \otimes 1) (f_1 \otimes f_2; \eta, \bar{\eta}),
\end{flalign*}
where we used the relation $\gamma^{\mu} \gamma_{\mu} = \delta^{\mu}_{\mu} = 2$ which holds true in the two-dimensional Euclidean scenario, being $\delta$ the Euclidean metric. Eventually the last term reads
\begin{flalign*}
	\text{(c)} &:= \Gamma^W_{\bullet_{\mathsf{Q}}}[\Phi \otimes \Gamma^W_{\cdot_{\mathsf{Q}}}(G_{\bar{\psi}} \circledast ((\bar{\Phi}\gamma_{\mu} \Phi) \bar{\Phi} \gamma^{\mu}))] (f_1 \otimes f_2; \eta, \bar{\eta}) \\ &=  [\Phi \bullet_{\mathsf{Q}} (G_{\bar{\psi}} \circledast ((\bar{\Phi} \gamma_{\mu} \Phi) \bar{\Phi} \gamma^{\mu}) + C \bar{\Phi})] (f_1 \otimes f_2; \eta, \bar{\eta}) \\&= [\Phi \otimes (G_{\bar{\psi}} \circledast ((\bar{\Phi} \gamma_{\mu} \Phi) \bar{\Phi} \gamma^{\mu} + C \bar{\Phi}))] (f_1 \otimes f_2; \eta, \bar{\eta}) \\ &+ Q \cdot (1 \otimes G_{\bar{\psi}} \circledast (\bar{\Phi} \gamma_{\mu} \Phi) \gamma^{\mu})(f_1 \otimes f_2; \eta, \bar{\eta}) \\&+ Q \cdot (1 \otimes G_{\bar{\psi}} \circledast (\gamma_{\mu} \Phi \bar{\Phi} \gamma^{\mu}))(f_1 \otimes f_2; \eta, \bar{\eta}) \\ &+ Q \cdot (1 \otimes G_{\psi} \circledast \widetilde{C}1) (f_1 \otimes f_2; \eta, \bar{\eta}).
\end{flalign*}
As a result, evaluating at $\eta=0$ and thus also at $\bar{\eta}=0$, the two-point correlation function reads
\begin{flalign}\label{Eq: omega20}
	\mathbb{E}&[\widehat{\psi}^{\rho}_{0}[[\lambda]](f_1) (\widehat{\bar{\psi}}_{\rho})_{0}[[\lambda]](f_2)] = \\ \notag &= Q(f_1 \otimes f_2) + \lambda Q \cdot (G_{\psi} \circledast \widetilde{C}1 \otimes 1) (f_1 \otimes f_2) + \lambda Q (G_{\bar{\psi}} \circledast C1 \otimes 1)(f_1 \otimes f_2) + \mathcal{O}(\lambda^2).
\end{flalign}
Analogously, if we consider
\begin{flalign}
	\label{2pointeta1}
	\mathbb{E}[ (\widehat{\bar{\psi}}_{\rho})_{\eta}[[\lambda]](f_1)\widehat{\psi}^{\rho}_{\eta}[[\lambda]](f_2)] &= (\bar{\Psi}_{\cdot_{\mathsf{Q}}} [[\lambda]] \bullet_{\Gamma^W_{\bullet_{\mathsf{Q}}}} \Psi_{\cdot_{\mathsf{Q}}} [[\lambda]])(f_1 \otimes f_2; \eta, \bar{\eta}), 
\end{flalign}
an identical procedure yields
\begin{flalign}
	\label{omega201}
	\mathbb{E}[& (\widehat{\bar{\psi}}_{\rho})_{0}[[\lambda]](f_1)\widehat{\psi}^{\rho}_{0}[[\lambda]](f_2)] = \\ \notag &= \widetilde{Q}(f_1 \otimes f_2) + \lambda \widetilde{Q} \cdot (1 \otimes G_{\psi} \circledast \widetilde{C}1) (f_1 \otimes f_2) + \lambda \widetilde{Q} (G_{\bar{\psi}} \circledast C1 \otimes 1)(f_1 \otimes f_2) + \mathcal{O}(\lambda^2)\,.
\end{flalign}
\begin{remark}
	When computing the two-point correlation functions of the solutions to Equation \eqref{Eq: stochastic thirring model}, resorting to a procedure analogous to the one adopted in the proof of Theorem \ref{expvalue0}, one can show that the cubic nature of the interaction term forces $(\Psi_{\cdot_{\mathsf{Q}}} [[\lambda]] \bullet_{\Gamma^W_{\bullet_{\mathsf{Q}}}} \Psi_{\cdot_{\mathsf{Q}}} [[\lambda]])(f_1 \otimes f_2; 0, 0)=(\bar{\Psi}_{\cdot_{\mathsf{Q}}} [[\lambda]] \bullet_{\Gamma^W_{\bullet_{\mathsf{Q}}}} \bar{\Psi}_{\cdot_{\mathsf{Q}}} [[\lambda]])(f_1 \otimes f_2; 0, 0)=0$. 
\end{remark}

\subsection{Renormalized Equation}\label{Sec: Renormalized Equation}
In the preceding sections we showed how the deformation maps $\Gamma_{\cdot_{\mathsf{Q}}}^W$ and $\Gamma_{\bullet_{\mathsf{Q}}}^W$ play a key r\^ole in implementing the stochastic properties of the underlying model at the level of functionals. To wit, the expectation value of the perturbative solution coincides with the action of the deformation map $\Gamma^W_{\cdot_{\mathsf{Q}}}$ introduced in Proposition \ref{Prop: deformation map} on Equation \eqref{Eq: functional thirring}. The explicit expression reads
\begin{flalign}\label{Eq: soleqdiracdef}
	\Psi^{\rho}_{\cdot_{\mathsf{Q}}} [[\lambda]] &= \Gamma^W_{\cdot_{\mathsf{Q}}} (\Psi^{\rho} [[\lambda]]) = \Gamma^W_{\cdot_{\mathsf{Q}}} (\Phi^{\rho} + \lambda (G_{\psi})^{\rho}_{\rho'} \circledast [(\bar{\Psi} \gamma^{\mu} \Phi)(\gamma_{\mu})^{\rho'}_{\rho''} \Psi^{\rho''}]) \\ \notag &= \Phi^{\rho} + \lambda (G_{\psi})^{\rho}_{\rho'} \circledast [(\bar{\Psi}_{\cdot_{\mathsf{Q}}} \gamma^{\mu} \cdot_{\mathsf{Q}} \Psi_{\cdot_{\mathsf{Q}}})(\gamma_{\mu})^{\rho'}_{\rho''} \cdot_{\mathsf{Q}} \Psi^{\rho''}_{\cdot_{\mathsf{Q}}}],
\end{flalign} 
where $\Psi^{\rho}_{\cdot_{\mathsf{Q}}} = \Gamma^W_{\cdot_{\mathsf{Q}}}(\Psi^{\rho})$ and $(\bar{\Psi}_{\rho})_{\cdot_{\mathsf{Q}}} = \Gamma^W_{\cdot_{\mathsf{Q}}}(\bar{\Psi}_{\rho})$. Despite being well-defined, the expression in Equation \eqref{Eq: soleqdiracdef} is hardly manageable since it involves the deformed product $\cdot_{\mathsf{Q}}$ defined in Equation \eqref{Eq: deformed product}. Hence, it seems natural to wonder whether it is possible to bypass this hurdle, namely reformulating Equation \eqref{Eq: soleqdiracdef} in an equivalent way in terms of the lone pointwise tensor product of functional-valued vector distributions. The answer is positive and it requires the above deformed equation to undergo a suitable renormalization procedure.
This leads to a \textit{renormalised equation} which contains additional contributions known as \emph{counter-terms} in the physics literature. 
\begin{theorem}\label{renequcoeff}    
	On account of Equation \eqref{Eq: functional thirring}, denoting by $\Psi^{\rho} [[\lambda]] \in \mathcal{A}^W [[\lambda]]$ its perturbative solution - whose counterpart in the deformed algebra $\mathcal{A}^W_{\cdot_{\mathsf{Q}}}$ is given by $\Psi^{\rho}_{\cdot_{\mathsf{Q}}} := \Gamma^W_{\cdot_{\mathsf{Q}}} (\Psi^{\rho})$ - there exists a collection of linear operator-valued vector functionals $\{(H_{k})^{\rho}_{\rho'}\}_{k \in \mathbb{N}}$ satisfying the following properties: 
	\begin{itemize}
		\item [1.] for all $k \in \mathbb{N}$ and for any field configuration $(\eta, \bar{\eta}) \in \Gamma(D\mathbb{R}^2 \oplus D^* \mathbb{R}^2)$, it holds that
		\begin{equation}
			\label{Hkmap}
			H_k(\eta, \bar{\eta}) : \Gamma(D\mathbb{R}^2 \oplus D^* \mathbb{R}^2) \rightarrow \Gamma(D\mathbb{R}^2 \oplus D^* \mathbb{R}^2), 
		\end{equation}
		where we omitted the spinor indices for the sake of conciseness. 
		\item [2.] for all $k \in \mathbb{N}$, every $H_k$ depends polynomially on the fields $\eta, \bar{\eta}$ and, for all $i,j, k \in \mathbb{N}$ with $i+j = 2k+1$, one has that
		\begin{equation}
			\label{Hkder}
			H_k^{(i,j)} (0) = 0. 
		\end{equation}
	\end{itemize}
	Furthermore, if we define 
	\begin{equation}
		\label{H}
		H := \sum_{k=1}^{\infty} \lambda^k H_k,
	\end{equation}
	it descends that  
	\begin{equation}
		\label{psiHk}
		\Psi^{\rho}_{\cdot_{\mathsf{Q}}} = \Phi^{\rho} + \lambda (G_{\psi})^{\rho}_{\rho'} \circledast ((\bar{\Psi}_{\cdot_{\mathsf{Q}}} \gamma_{\mu} \Psi_{\cdot_{\mathsf{Q}}}) (\gamma^{\mu})^{\rho'}_{\rho''} \Psi^{\rho''}_{\cdot_{\mathsf{Q}}})+ (G_{\psi})^{\rho}_{\rho'} \circledast (H^{\rho'}_{\rho''} \Psi^{\rho''}_{\cdot_{\mathsf{Q}}}).
	\end{equation}
\end{theorem}
\begin{proof} 
	The strategy of the proof is to exhibit a suitable expression for $H_k$ in Equation \eqref{H}, satisfying all the listed requirements, at every step of the induction procedure. Let us begin by studying the case $k=0$. At the zeroth order in $k$
	\begin{equation*}
		\Psi^{\rho}_{\cdot_{\mathsf{Q}}} = \Phi^{\rho} + \mathcal{O}(\lambda),
	\end{equation*}
	which does not require any correction. At order $k=1$, we can exploit Equation \eqref{F1} to frame $H_1$. More specifically, we compare the left hand-side of Equation \eqref{psiHk} computed by means of the relation $ \Psi^{\rho}_{\cdot_{\mathsf{Q}}} =  \Gamma^{W}_{\cdot_{\mathsf{Q}}}( \Psi^{\rho})$, with the right hand-side, obtaining 
	\begin{flalign*}
		\text{LHS} &= \Gamma^W_{\cdot_{\mathsf{Q}}} (\Psi^{\rho}) = \Gamma^W_{\cdot_{\mathsf{Q}}} [\Phi^{\rho} + \lambda (G_{\psi})^{\rho}_{\rho'} \circledast ((\bar{\Phi} \gamma_{\mu} \Phi) (\gamma^{\mu})^{\rho'}_{\rho''} \Phi^{\rho''})] + \mathcal{O}(\lambda^2) \\ &= \Phi^{\rho} + \lambda (G_{\psi})^{\rho}_{\rho'} \circledast ((\bar{\Phi} \gamma_{\mu} \Phi) (\gamma^{\mu})^{\rho'}_{\rho''} \Phi^{\rho''}) + \lambda (G_{\psi})^{\rho}_{\rho'} \circledast \widetilde{C}^{\rho'}_{\rho''} \Phi^{\rho''} + \mathcal{O}(\lambda^2) \\
		\text{RHS} &= \Phi^{\rho} + \lambda (G_{\psi})^{\rho}_{\rho'} \circledast ((\bar{\Phi} \gamma_{\mu} \Phi) (\gamma^{\mu})^{\rho'}_{\rho''} \Phi^{\rho''}) + \lambda (G_{\psi})^{\rho}_{\rho'} \circledast H^{\rho'}_{\rho''} \Phi^{\rho''} + \mathcal{O}(\lambda^2),
	\end{flalign*}
	being $\widetilde{C} \in \mathcal{E}(\mathbb{R}^2; \mathbb{C}^2 \oplus \mathbb{C}^2)$ as defined in the previous section. Hence, a direct inspection entails that $H_1 = \widetilde{C}$ satisfies all properties listed in the statement of the theorem up to the order $\lambda^2$. \\
	\noindent Let us assume that $H_{j}$ has been consistently assigned for any $j \le k -1$ with $k \ge 1$. We show that we can find a suitable candidate for $H_{k}$, abiding by Equations \eqref{Hkmap}, \eqref{Hkder} and \eqref{psiHk} up to order $\mathcal{O}(\lambda^{k+1})$. Similarly to the case $k=1$, we can rewrite the left hand-side of Equation \eqref{psiHk} as follows 
	\begin{flalign*}
		\text{LHS} &:= T^{\rho}_{k-1} + \Gamma^{W}_{\cdot_{\mathsf{Q}}} \{\lambda^{k} \sum_{k_1 + k_2 + k_3 = k-1} (G_{\psi})^{\rho}_{\rho'} \circledast [(\widetilde{F}_{k_1} \gamma_{\mu} F_{k_2}) (\gamma^{\mu})^{\rho'}_{\rho''} F^{\rho''}_{k_3}]\} + \mathcal{O}(\lambda^{k+1})\\ &= T^{\rho}_{k-1} + \lambda^k \sum_{k_1 + k_2 + k_3 = k-1} \Gamma^{W}_{\cdot_{\mathsf{Q}}}\{(G_{\psi})^{\rho}_{\rho'} \circledast [(\widetilde{F}_{k_1} \gamma_{\mu} F_{k_2}) (\gamma^{\mu})^{\rho'}_{\rho''} F^{\rho''}_{k_3}]\} + \mathcal{O}(\lambda^{k+1})\\
		&= T^{\rho}_{k-1} + \lambda^k \sum_{k_1 + k_2 + k_3 = k-1} (G_{\psi})^{\rho}_{\rho'} \circledast \Gamma^{W}_{\cdot_{\mathsf{Q}}}[(\widetilde{F}_{k_1} \gamma_{\mu} F_{k_2}) (\gamma^{\mu})^{\rho'}_{\rho''} F^{\rho''}_{k_3}] + \mathcal{O}(\lambda^{k+1}), 
	\end{flalign*}
	where we denote by $T_{k-1}$ all the contributions up to the order $\mathcal{O}(\lambda^{k-1})$ and where we exploited Equation \eqref{Eq: perturbative expansion} as well as Equations \eqref{Fk} and \eqref{tildeFk}. For what concerns the right hand-side, one has that 
	\begin{flalign*}
		\text{RHS} &:= T^{\rho}_{k-1} + \lambda^k \sum_{k_1 + k_2 + k_3 = k-1} (G_{\psi})^{\rho}_{\rho'} \circledast [\Gamma^{W}_{\cdot_{\mathsf{Q}}}(\widetilde{F}_{k_1}) \gamma_{\mu} \Gamma^{W}_{\cdot_{\mathsf{Q}}}(F_{k_2})) (\gamma^{\mu})^{\rho'}_{\rho''} \Gamma^{W}_{\cdot_{\mathsf{Q}}}(F^{\rho''}_{k_3})] \\& + \lambda^k \sum_{\substack{k_1 + k_2 = k \\ k_2 > 0}} (G_{\psi})^{\rho}_{\rho'} \circledast [(H_{k_1})^{\rho'}_{\rho''} \Gamma^{W}_{\cdot_{\mathsf{Q}}}(F^{\rho''}_{k_2})] + \lambda^k (G_{\psi})^{\rho}_{\rho'} \circledast [(H_k)^{\rho'}_{\rho''} \Phi^{\rho''}] + \mathcal{O}(\lambda^{k+1}).
	\end{flalign*}
	Since the inductive hypothesis implies that $T^{\rho}_{k-1}$ satisfies Equation \eqref{psiHk} up to order $\mathcal{O}(\lambda^k)$, it does not contribute to the expression for $H_k$ and we can neglect it in the ensuing analysis. \\
	\noindent In order to define the sought linear operator $H_k$ we recall that, in the proof of Theorem \ref{expvalue0}, we introduced the space of vector-valued distributions with an odd polynomial degree in the fields and we denoted it by $\mathsf{O} \subset \mathcal{A}^W$. From the inductive hypothesis it descends that $H_j$ for $j \le k-1$ has an even polynomial degree in $\Phi^{\rho}$ and $\bar{\Phi}_{\rho}$. Furthermore, the specific form of the interaction term $V^{\rho}(\psi, \bar{\psi}):= (\bar{\psi} \gamma_{\mu} \psi) (\gamma^{\mu})^{\rho}_{\rho'} \psi^{\rho'}$ appearing in Equation \eqref{Eq: functional thirring} entails that every perturbative coefficient $F_j$ with $j \in \mathbb{N}$ - see Equation \eqref{Fk} - depends exactly on $n$ cospinor fields $\bar{\Phi}_{\rho}$ and on $n+1$ spinor fields $\Phi^{\rho}$ with $n \in \mathbb{N}, n \le j$. Moreover, since the action of the deformation map $\Gamma^W_{\cdot_{\mathsf{Q}}}$ corresponds to the contraction of a pair of fields $\Phi^{\rho}$ and $\bar{\Phi}_{\rho}$, this operation always leaves a surviving unpaired spinor field. Hence, the second and third contributions in the expansion of the right hand-side of Equation \eqref{psiHk} are of the form $G_{\psi} \circledast u$, with $u \in \mathsf{O}$. More precisely $u:= K\Phi$ where $K$ is a linear operator-valued vector functional which satisfies all the requirements imposed on the renormalization counterterms in the statement of the theorem. Therefore, it holds that 
	\begin{flalign*}
		\Psi^{\rho}_{\cdot_{\mathsf{Q}}} &- \Phi^{\rho} - \lambda (G_{\psi})^{\rho}_{\rho'} \circledast ((\bar{\Psi}_{\cdot_{\mathsf{Q}}} \gamma_{\mu} \Psi_{\cdot_{\mathsf{Q}}}) (\gamma^{\mu})^{\rho'}_{\rho''} \Psi^{\rho''}_{\cdot_{\mathsf{Q}}}) -  (G_{\psi})^{\rho}_{\rho'} \circledast (H^{\rho'}_{\rho''} \Psi^{\rho''}_{\cdot_{\mathsf{Q}}}) \\&= \lambda^k (G_{\psi})^{\rho}_{\rho'} \circledast [(K - H_k)^{\rho'}_{\rho''} \Phi^{\rho''}] + \mathcal{O}(\lambda^{k+1}). 
	\end{flalign*}
	Setting $H_k := K$, all properties listed in the statement of the theorem are fulfilled. Observing that the previous equation holds true up to order $\mathcal{O}(\lambda^{k+1})$, one obtains the sought result. 
\end{proof}

\subsection{On the renormalizability of the Thirring Model}\label{Sec: Renormalizability}
Since renormalization plays a pivotal r\^{o}le in the analysis of the stochastic Thirring model, it is natural to wonder whether, in computing expectation values and correlation functions of the solution, only a finite number of distributions requires a renormalization procedure. This goes under the name of \textit{subcritical regime}. To this end, in this section we shall work on $M \equiv \mathbb{R}^d$ with $d \ge 1$, equipped with the Euclidean metric $\delta$. Yet, we highlight that our results will only depend on the underlying dimension $d$ and therefore, mutatis mutandis, they hold true if we replace $\mathbb{R}^d$ with a Riemannian manifold $(M,g)$.

In the ensuing analysis, we shall generalize the strategy illustrated in \cite[Sec. 5.2]{BDR23} for the complex scalar field to the Fermionic scenario. The problem of studying the regime of sub-criticality of the Thirring model described in Section \ref{Sec: on the model} can be conveniently tackled by resorting to a graphical approach. In other words, every contribution in the perturbative expansion of the solution of Equation \eqref{Eq: functional thirring} can be depicted as a graph, abiding by the following rules: 

\begin{itemize}
	\item the symbol $\fiammifero$ represents the occurrence of the spinor field $\Phi^{\rho} \in \mathcal{A}^W$, whilst $\fiammiferoCC$ corresponds to the cospinor field $\bar{\Phi}_{\rho} \in \mathcal{A}^W$ with $\bar{\Phi}_{\rho} := \Phi^{\dagger}_{\rho'} (\gamma_0)^{\rho'}_{\rho}$. 
	\item the convolution with the fundamental solution of the Dirac operator $G_{\psi}$ is graphically represented by the segment $\propagatore$ , while that with the fundamental solution of its adjoint $G_{\bar{\psi}}$ is denoted by $\propagatoreCC$ .
	\item to pictorially represent the pointwise product between two elements lying in the vector algebra $\mathcal{A}^W$, we join the roots of their respective graphs in a vertex. 
\end{itemize}
\begin{remark}\label{gammamatrixgraph}
	As we shall see in the following, the occurrence of the gamma matrices $\gamma^{\mu}, \mu \in \{0,1\}$ is irrelevant for the ensuing analysis, since the deformation map $\Gamma^W_{\cdot_{\mathsf{Q}}}$ does not act on the underlying matrix structure. Therefore we will omit any reference to all kind of indices in the graphical representation. However, one should bear in mind that each vertex in the aforementioned graph structure comes with a pair of suitably contracted spinor indices.   
\end{remark}

\begin{example}\label{exgraph}
	To make the reader acquainted with the pictorial representation of elements of the vector algebra $\mathcal{A}^W$, let us consider a few representatives which are involved in the perturbative expansion of the solution of Equation \eqref{Eq: functional thirring}, \textit{i.e.}, 
	\begin{equation*}
		\Phi^{\rho} (\gamma^{\mu})^{\rho'}_{\rho} \bar{\Phi}_{\rho'}=
		\begin{tikzpicture}[thick,scale=1.2]
			\draw[red] (0,0) -- (0.2,0.3);
			\filldraw[red] (0.2,0.3)circle (1pt);
			\draw (0,0) -- (-0.2,0.3);
			\filldraw (-0.2,0.3)circle (1pt);,
		\end{tikzpicture}
	\end{equation*}
	\begin{equation*}
		(G_{\psi})^{\rho}_{\rho'}\circledast [(\bar{\Phi} \gamma_{\mu} \Phi) (\gamma^{\mu})^{\rho'}_{\rho''} \Phi^{\rho''}]=
		\begin{tikzpicture}[thick,scale=1.2]
			\draw (0,0) -- (0,0.3);
			\draw[red] (0,0.3) -- (0.3,0.6);
			\filldraw[red] (0.3,0.6)circle (1pt);
			\draw (0,0.3) -- (0,0.7);
			\filldraw (0,0.7)circle (1pt);
			\draw (0,0.3) -- (-0.3,0.6);
			\filldraw (-0.3,0.6)circle (1pt);,
		\end{tikzpicture}
	\end{equation*}
	\begin{equation*}
		(G_{\bar{\psi}})^{\rho}_{\rho'}\circledast [(\bar{\Phi} \gamma_{\mu} \Phi) \bar{\Phi}_{\rho''}(\gamma^{\mu})^{\rho''}_{\rho'}]= \begin{tikzpicture}[thick,scale=1.2]
			\draw[red] (0,0) -- (0,0.3);
			\draw[red] (0,0.3) -- (0.3,0.6);
			\filldraw[red] (0.3,0.6)circle (1pt);
			\draw[red] (0,0.3) -- (0,0.7);
			\filldraw[red] (0,0.7)circle (1pt);
			\draw (0,0.3) -- (-0.3,0.6);
			\filldraw (-0.3,0.6)circle (1pt);
		\end{tikzpicture}.
	\end{equation*}
	Observe that a contraction of the roots of each contribution requires a gamma matrix - see Remark \ref{gammamatrixgraph}. The matrix index $\mu \in \{0,1\}$ plays no r\^ole in the following construction. 
\end{example}
As illustrated in Example \ref{exgraph}, the  graphical rules listed above suffice to represent the perturbative expansion of the solution as per Equation \eqref{Eq: perturbative expansion}. Notwithstanding, the necessity of a renormalization procedure arises only after acting on each contribution with the deformation map $\Gamma^W_{\cdot_{\mathsf{Q}}}$ - see Section \ref{Sec: local deformation} for additional details. Heuristically speaking, the action of $\Gamma^W_{\cdot_{\mathsf{Q}}}$ coincides with a series of contractions between the spinor and cospinor fields $\Phi^{\rho}$ and $\bar{\Phi}_{\rho}$, which are closely linked to the occurrence of powers of $Q := Q_{\psi \bar{\psi}}$ and $\widetilde{Q} := Q_{\bar{\psi} \psi}$. From a graphical point of view, this operation is tantamount to performing a progressive collapse of pairs of leaves into a loop. In addition, we remark that from the properties of the vector-valued Gaussian white noise - see Definition \ref{Def: fermionic white noise} - it descends that the only admissible contractions are between leaves of different colour. The following example illustrates how these concepts can be applied to a simplified case. 
\begin{example}
	For all $f \in \mathcal{D}(\mathbb{R}^d)$ and for any field configuration $\psi \in \Gamma(D \mathbb{R}^d) \simeq C^{\infty}(\mathbb{R}^d; \mathbb{C}^d)$, the diagrammatic counterpart of the following expression, \textit{i.e.},
	\begin{equation*}
		\Gamma^W_{\cdot_{\mathsf{Q}}}(\bar{\Phi}_{\rho'}(\gamma^{\mu})^{\rho'}_{\rho}\Phi^{\rho} )(f; \psi, \bar{\psi}) = \Phi^{\rho} (\gamma^{\mu})^{\rho'}_{\rho} \bar{\Phi}_{\rho'}(f; \psi, \bar{\psi}) + \widetilde{C} (f; \psi, \bar{\psi}) 
	\end{equation*}
	is given by 
	\begin{equation*}
		\Gamma^W_{\cdot_{\mathsf{Q}}}(\bar{\Phi}_{\rho'}(\gamma^{\mu})^{\rho'}_{\rho}\Phi^{\rho})=
		\begin{tikzpicture}[thick,scale=1.2]
			\draw[red] (0,0) -- (0.2,0.3);
			\filldraw[red] (0.2,0.3)circle (1pt);
			\draw (0,0) -- (-0.2,0.3);
			\filldraw (-0.2,0.3)circle (1pt);
		\end{tikzpicture}
		+\fish,
	\end{equation*}
	where the second contribution corresponds to the contraction of the spinor and cospinor fields codified in the action of the deformation map. 
\end{example}
At this stage, we need to isolate all contributions occurring in the perturbative expansion of the solution of the stochastic nonlinear Dirac equation. From the explicit form of Equations \eqref{Fk} and \eqref{tildeFk}, we deduce that each perturbative coefficient in the expansion is built out of elements of the form $G_{\psi} \circledast [(\widetilde{F}_{k1} \gamma_{\mu} F_{k_2}) (\gamma^{\mu})^{\rho'}_{\rho''} F^{\rho''}_{k_3}]$ or $G_{\bar{\psi}} \circledast [(\widetilde{F}_{k1} \gamma_{\mu} F_{k_2}) (\widetilde{F}_{k_3})_{\rho''}(\gamma^{\mu})^{\rho''}_{\rho'}]$. Hence, an inductive reasoning implies that their graphical counterparts must exhibit a nested tree structure with branches of the form 
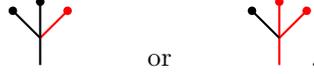
\begin{equation*}
	\begin{tikzpicture}[thick,scale=1.2]
		\draw (0,0) -- (0,0.3);
		\draw[red] (0,0.3) -- (0.3,0.6);
		\filldraw[red] (0.3,0.6)circle (1pt);
		\draw (0,0.3) -- (0,0.7);
		\filldraw (0,0.7)circle (1pt);
		\draw (0,0.3) -- (-0.3,0.6);
		\filldraw (-0.3,0.6)circle (1pt);
	\end{tikzpicture}
	\hspace{1cm}\text{or}\hspace{1cm}
	\begin{tikzpicture}[thick,scale=1.2]
		\draw[red] (0,0) -- (0,0.3);
		\draw[red] (0,0.3) -- (0.3,0.6);
		\filldraw[red] (0.3,0.6)circle (1pt);
		\draw[red] (0,0.3) -- (0,0.7);
		\filldraw[red] (0,0.7)circle (1pt);
		\draw (0,0.3) -- (-0.3,0.6);
		\filldraw (-0.3,0.6)circle (1pt);
	\end{tikzpicture}.
\end{equation*}
Furthermore, we outline that, if in the computations we increase by one the perturbative order in the coupling constant $\lambda$, at a graphical level we need to add to the existing diagram a vertex of the following form:
\begin{equation*}
	\begin{tikzpicture}[thick,scale=1.2]
		\draw[red] (0,0) -- (0.3,0.3);
		\filldraw[red] (0.3,0.3)circle (1pt);
		\draw (0,0) -- (0,0.4);
		\filldraw (0,0.4)circle (1pt);
		\draw (0,0) -- (-0.3,0.3);
		\filldraw (-0.3,0.3)circle (1pt);
	\end{tikzpicture}
	\hspace{1cm}\text{or}\hspace{1cm}
	\begin{tikzpicture}[thick,scale=1.2]
		\draw[red] (0,0) -- (0.3,0.3);
		\filldraw[red] (0.3,0.3)circle (1pt);
		\draw[red] (0,0) -- (0,0.4);
		\filldraw[red] (0,0.4)circle (1pt);
		\draw (0,0) -- (-0.3,0.3);
		\filldraw (-0.3,0.3)circle (1pt);
	\end{tikzpicture}\,.
\end{equation*}
\noindent This observation allows us to restrict the pool of graphs contributing to the construction of the perturbative solution of Equation \eqref{Eq: functional thirring}. 

The strategy we adopt to estimate the number of graphs in the expansion of the perturbative solution, which yield ill-defined structures, relies on the computation of their degree of divergence, see Equation \eqref{Eq: degree of divergence}. Yet, the presence of non-contracted leaves does not contribute to the degree of divergence of a given graph. Indeed, at a distributional level, these are equivalent to a multiplication by a smooth function $\eta \in \Gamma(D \mathbb{R}^d) \simeq C^{\infty}(\mathbb{R}^d; \mathbb{C}^d)$, which does not affect the overall divergence of the diagram. Thus, it is legit to focus on the analysis of the lone \emph{maximally contracted graphs}, namely those in which all leaves are contracted. The need to select the sole maximally pathological graphs prompts the following definition.

\begin{definition}\label{Def: admissible graphs}
	We call a given graph $\mathcal{G}$ \emph{admissible} if it can be obtained from a maximally contracted diagram associated to any perturbative coefficient $F_k$ as per Equation \eqref{Fk}. 
\end{definition}
\begin{definition}\label{Def: valence}
	Given an admissible graph $\mathcal{G}$, we call \emph{valence} of one of the vertices of $\mathcal{G}$ the number of edges connected to it. 
\end{definition}
Denoting with $u_{\mathcal{G}}$ the distribution associated to the admissible graph $\mathcal{G}$, we have all the tools to delve into the study of the degree of divergence of the perturbative coefficients appearing in Equation \eqref{Eq: perturbative expansion}. They are completely characterized by the following properties:
\begin{itemize}
	\item the graph $\mathcal{G}$ has $N$ vertices of valence $1$ or $4$ as per Definition \ref{Def: valence} and $L$ edges. We pinpoint that, when counting the number of vertices, we include all internal points on which we perform the contractions but not the root of the diagram.  
	\item denoting by $s(e)$ the origin of an edge $e$ - where the letter $s$ stands for ``source'' - and by $t(e)$ its target, the graphical counterpart of each propagator of the form $G_{\psi}(x_{s(e)}, x_{t(e)})$ or of the form $G_{\bar{\psi}}(x_{s(e)}, x_{t(e)})$ is an edge $e$.
	\item let $E_{\mathcal{G}} := \{e_i \, | \, i =1, ..., N\}$ be the collection of edges of the graph $\mathcal{G}$. Then, the distribution $u_{\mathcal{G}}$ has the following associated integral kernel:
	\begin{equation*}
		u_{\mathcal{G}}(x_1, ..., x_N) := \prod_{e \in E_{\mathcal{G}}} G(x_{s(e)}, x_{t(e)}),
	\end{equation*}
	where each $G$ appearing in the above product coincides either with $G_{\psi}$ or $G_{\bar{\psi}}$, depending on the contribution under scrutiny. 
\end{itemize}
The action of $\Gamma^W_{\cdot_{\mathsf{Q}}}$ on the elements occurring in the perturbative expansion in Equation \eqref{Eq: perturbative expansion} yields distributions $u\in\mathcal{D}'(U)$, with $U \subseteq \mathbb{R}^{Nd}$. As already stressed in the proof of Proposition \ref{Prop: deformation map}, the singular support of $u$ is the thin diagonal of $\mathbb{R}^{Nd}$, \textit{i.e.}, 
\begin{equation*}
	\mathrm{Diag}_{Nd}(\mathbb{R}^{Nd}) := \{(x_1, ..., x_{Nd}) \in \mathbb{R}^{Nd} \, | \, x_1 = ... = x_{Nd}\}.
\end{equation*}
In order to tackle the problem of renormalizability, we need to resort to microlocal analytical techniques, specifically to the computation of the degree of divergence of $u$ with respect to the thin diagonal of $\mathbb{R}^{Nd}$. Since the dimension of the manifold of interest is $\dim\mathbb{R}^{Nd} = N d$, whereas that of the singular submanifold $\mathrm{Diag}_{Nd} (\mathbb{R}^{Nd})$ is equal to $d$, a direct computation yields
\begin{equation*}
	\textrm{codim} \, (\mathrm{Diag}_{Nd} (\mathbb{R}^{Nd})) = N d - \dim \, (\mathrm{Diag}_{Nd} (\mathbb{R}^{Nd})) = Nd - d = (N-1)d\,.
\end{equation*}
In the following the scaling degree of the fundamental solution of the Dirac operator and of its formal adjoint at the thin diagonal of $\mathbb{R}^d$ will play a distinguished r\^{o}le. 
Hence it is worth recalling that, on account of Equation \eqref{Eq: scaling degree fundamental solutions} and of the following discussion
	\begin{equation*}
		\text{\textrm{sd}}_{\mathrm{Diag}_2} \, (G_{\psi}) = \text{\textrm{sd}}_{\mathrm{Diag}_2} \, (G_{\bar{\psi}}) = d-1.
	\end{equation*}
For any fixed graph $\mathcal{G}$ with $L$ edges and $N$ vertices, being every convolution with $G_{\psi}$ or $G_{\bar{\psi}}$ represented by an edge, we have that the overall scaling degree of $u_{\mathcal{G}}$ is $L (d-1)$. Thus the degree of divergence associated to a maximally contracted graph $\mathcal{G}$ is
\begin{equation}\label{Eq: divgraph}
	\rho_{\mathrm{Diag}_{Nd} (\mathbb{R}^{Nd})} = L (d-1) - (N-1)d. 
\end{equation}

Our goal is thus to rewrite Equation \eqref{Eq: divgraph} in terms of the order $k$ of the perturbative expansion in the coupling constant, to find a region of the parameter space in which the degree of divergence becomes negative as $k$ grows. Since the structure of the graphs associated to the perturbative solution of the Thirring model as per Equation \eqref{Eq: functional thirring} coincides structurally to that of the stochastic nonlinear Schr\"odinger equation, studied in \cite{BDR23}, we refer to this paper for a proof of the following lemmas. 
\begin{lemma}\label{Lem: vertices}
	The number of fixed vertices of valence $1$ of an admissible graph $\mathcal{G}$ at a perturbative order $k$ is $N_1(k)=2k+1$, whilst that for those of valence $4$ is $N_4(k)=k$.
\end{lemma}
The results in Lemma \ref{Lem: vertices} must be refined since each contraction affects a pair of fields. Hence there will always remain a non-contracted leaf in every branch. As a result, at any perturbative order, we need to remove $\frac{1}{2}(2k+1-1) = k$ vertices, corresponding to the maximum number of contractions we can perform on the given graph. To wit, for a maximally contracted graph, the total number of free vertices amounts to 
\begin{equation}
	\label{Nkc}
	N(k) = 3k+1 - k = 2k+1. 
\end{equation}
\begin{lemma}
	The number of edges of an admissible graph $\mathcal{G}$ does not depend on the contractions and it is related to the perturbative order $k$ as
	\begin{equation}
		L(k)=3k+1.
	\end{equation}
\end{lemma}
\noindent Inverting Equation \eqref{Nkc}, it is possible to express $k$ as a function of $N$, \textit{i.e.},
\begin{equation}\label{Eq: edges}
	k = \frac{N-1}{2}.
\end{equation}
Moreover, since Equation \eqref{Nkc} implies that every admissible configuration is characterized by an odd number of vertices, $N-1$ is even and, thus, $k(N)$ is an integer. Plugging this relation into Equation \eqref{Eq: edges}, one obtains 
\begin{equation}
	\label{Lk1}
	L(k) = 3k + 1 = 3 \left( \frac{N-1}{2} \right) + 1 = \frac{3}{2}N - \frac{1}{2}.
\end{equation}
Since we are interested in the sub-critical regime, we shall consider the case in which only a finite number of admissible graphs exhibits a singular behaviour. This  is tantamount to requiring that, if we consider a sufficiently large number of vertices $N$, the degree of divergence of the admissible graphs as per Definition \ref{Def: admissible graphs} becomes negative. 
\begin{theorem}\label{Thm: subcritical dimension}
	If the dimension $d \le 2$, the number of admissible graphs, which need to be renormalized, is finite.
\end{theorem}
\begin{proof}
	On account of Equation \eqref{Eq: divgraph} as well as of Equations \eqref{Eq: edges} and \eqref{Lk1}, it holds that
	\begin{flalign}
		\rho_{\mathrm{Diag}_{Nd}} (u_{\mathcal{G}}) &= \left[3 \left(\frac{N-1}{2}\right)+1\right] (d-1) - (N-1)d = \frac{3}{2} N d - \frac{d}{2} - \frac{3}{2} N +\frac{1}{2}- Nd + d \notag \\
		&= \frac{1}{2}(d-3) N + \frac{1}{2}(d+1).\label{Eq: critical dimension}
	\end{flalign}
	In order for the degree of divergence $\rho_{\mathrm{Diag}_{Nd}} (u_{\mathcal{G}})$ to be negative for $N$ large enough, we require the coefficient of the term depending on $N$ to be negative. As a consequence, 
	\begin{equation*}
		\frac{1}{2}(d-3) < 0 \Rightarrow d < 3.
	\end{equation*}
	We can thus conclude that for $d < 3$, the model is subcritical - \textit{i.e.}, $\rho_{\mathrm{Diag}_{Nd}}(u_{\mathcal{G}}) < 0$ for $N$ large enough.     
\end{proof}

\begin{remark}
 Observe that, as a consequence of Theorem \ref{Thm: subcritical dimension}, we can also infer that, if $d = 3$, $\rho_{\mathrm{Diag}_{Nd}}(u_{\mathcal{G}})$ as per Equation \eqref{Eq: critical dimension} is independent of $N$. This entails that the model under scrutiny lies in the so called critical regime and the tools used above cannot be applied to dray any a priori conclusion on the renormalizability of the model. 
\end{remark}

\vskip.2cm

\noindent{\bf Data Availability Statement:} We declare that the manuscript has no associated data and hence
no data set has been used in the realization of this work.

\vskip.2cm

\noindent{\bf Conflict of Interests:} There are no financial and non-financial conflict of interests.

\vskip.2cm

\noindent{\bf Ethical Approval \& Informed Consent:} Not applicable.

\paragraph{Acknowledgements.}
A.B. is supported by a PhD fellowship of the University of Pavia and partly by the GNFM-Indam Progetto Giovani {\em Feynman propagator for Dirac fields: a microlocal analytic approach}, CUP E53C22001930001, whose support is gratefully acknowledged. P.R by a postdoc fellowship of the Institute for Applied Mathematics of the University of Bonn. This work is based partly on the thesis of B.C., titled {\em A Microlocal Approach to the Study of the Nonlinear Stochastic Dirac Equation} submitted to the University of Pavia in partial fulfilment of the requirements for completing the master degree program in physics.

\appendix

\section{Spinor white noise}
In this work, randomness is encoded in the information carried by an (additive) Gaussian \emph{white noise}, which plays the r\^ole of a random source term in the associated stochastic partial differential equation. Since we are interested in a dynamical model for spinors, we devote this appendix to the characterization of a suitable notion of random distribution with the desired properties. As a first step in this direction, we discuss the notion of vector-valued random distribution.
\begin{definition}\label{randist}
	Let $V$ be a finite-dimensional vector space over the real or complex field with $dim \, V = m$, $m\geq 1$, endowed with the Euclidean inner, real or sesquilinear, product. 
	Given a probability space $(\mathfrak{W}, \mathcal{F}, \mathbb{P})$, where $\mathfrak{W}$ is a set, $\mathcal{F}$ is a $\sigma-$algebra consisting of subsets of $\mathfrak{W}$ while $\mathbb{P}$ denotes the probability measure over $\mathcal{F}$, we denote by $L^2((\mathfrak{W}, \mathcal{F}, \mathbb{P});V)$ the space of square integrable vector random variables on $(\mathfrak{W}, \mathcal{F}, \mathbb{P})$. We call \textit{vector-valued random distribution} a continuous, linear map $\eta: \mathcal{D}(\mathbb{R}^d) \rightarrow L^2((\mathfrak{W}, \mathcal{F}, \mathbb{P}); V)$. 
\end{definition}
\begin{remark}
	Definition \ref{randist} suggests that a \emph{vector-valued random distribution} $\eta$ can also be interpreted as a vector-valued distribution $\eta \in \mathcal{D}'(\mathbb{R}^d; L^2((\mathfrak{W}, \mathcal{F}, \mathbb{P}); V))$, hence taking values in the space of square integrable vector random variables. Furthermore, endowing $V$ with the standard basis $\{e_i\}_{i=1, ..., m}$, we can define the $i-th$ components of $\eta$ to be a random distribution $\eta_i \in \mathcal{D}'(\mathbb{R}^d; L^2(\mathfrak{W}, \mathcal{F}, \mathbb{P}))$ with $\eta^i(f) := \langle \eta(f), e_i^* \rangle$, where $\{e_i^*\}_{i=1,...,m}$ denotes the dual basis of $V^*$, for all $i=1, ..., m$ and for all $f \in \mathcal{D}(\mathbb{R}^d)$.
\end{remark}



\noindent In the following we give a definition of white noise tailored to spinor fields, focusing only on the case investigated in this paper. Since this is a complex valued distribution, we need to take into account that a complex normal distribution is completely specified not only by its mean and covariance but also by its relation matrix.  
\begin{definition}\label{Def: fermionic white noise}
	Let 
	\begin{equation*}
		\underline{\xi} := \begin{pmatrix}
			\xi \\
			\bar{\xi}
		\end{pmatrix}
		\in \mathcal{D}'(\mathbb{R}^d; L^2((\mathfrak{W}, \mathcal{F}, \mathbb{P}); \mathbb{C}^4))
	\end{equation*}
	be a vector-valued, complex random distribution as per Definition \ref{randist}. 
	We say that $\underline{\xi}$ is a \emph{vector-valued spinor white noise} if it is a vector-valued centered Gaussian random distribution whose covariance and relation is given component by component by
	\begin{flalign} 
		&\mathbb{E}[\xi^i(f) \xi^j(g)] = 0, \\
		&\mathbb{E}[\xi^i(f) \bar{\xi}_j^* (g)] = \langle \delta^{i}_j \delta_{\mathrm{Diag}_2} \ast f, g \rangle_{L^2(\mathbb{R}^d)}=-\mathbb{E}[\bar{\xi}_j^* (f)\xi^i(g) ],\label{Eq: noncommutative white noise}
	\end{flalign}
	$\forall f, g \in \mathcal{D}(\mathbb{R}^d), \forall i,j \in \{1,...,m\}$, where $\delta_{\mathrm{Diag}_2}$ denotes the Dirac delta distribution supported on the diagonal of $\mathbb{R}^2$.
\end{definition}
We stress that Equation \eqref{Eq: noncommutative white noise} can be read as a form of anticommutativity of the noise. This translates at the level of stochastic partial differential equations the prototypical behaviour of Fermionic fields which, in quantum field theory, are built in terms of spinors.

\section{Clifford Algebras}\label{App: Clifford Algebras}

In this section we report some basic notions on Clifford algebras referring to \cite{Lawson} for a complete survey and for all proofs of the statements present in this section. Our goal is to keep the paper self-consistent especially for those readers who are not familiar with the algebraic structures at the heart of spinor fields. In the following we consider only $\mathbb{R}^d$, $d\geq 1$ as an admissible background.

\begin{definition}[Clifford algebra]\label{Def: Clifford algebra}
	Given $\mathbb{R}^d$ endowed with the Euclidean metric $\delta$, we call \textbf{Clifford algebra} $Cl(d)$ of $\mathbb{R}^d$ the real associative algebra generated by $\mathbf{1}$ and by an orthonormal basis of $\mathbb{R}^d$, whose elements $\gamma_i$, $i=1,\ldots,d$ satisfy the relation
	\begin{equation*}
		\gamma_i\cdot\gamma_j+\gamma_j\cdot\gamma_i=2 \delta_{ij}\mathbf{1}\,.
	\end{equation*}
\end{definition}

\begin{remark}
	In Definition \ref{Def: Clifford algebra} one can replace $\delta$ with any Riemannian metric $g$ constructing an algebra isomorphic to $Cl(d)$. This can be realized by direct inspection of the following alternative, albeit equivalent, definition of Clifford algebra. Let $V$ be any finite-dimensional, real vector space endowed with a symmetric bilinear form $\mu$ and let $\mathcal{T}(V):=\bigoplus_{k\geq0} V^{\otimes k}$ be the associated universal tensor algebra where $V^{\otimes 0}:=\mathbb{R}$. Denoting by $\mathcal{I}_\mu(V)$ the ideal of $\mathcal{T}(V)$ generated by $v\cdot v+\mu(v)\mathbf{1}$ for all $v\in V$, we set
	\begin{equation*}
		Cl(V,\mu):=\frac{\mathcal{T}(V)}{\mathcal{I}_\mu(V)}.
	\end{equation*}
	This is isomorphic as an algebra to $Cl(d)$ with $d=\dim V$. 
\end{remark}


\begin{definition}[Spin group]\label{Def: Spin Group}
	Let $\mathbb{R}^d$ be endowed with the standard Euclidean metric $\delta$. The \textbf{Spin group} is the subset of $Cl(d)$ as per Definition \ref{Def: Clifford algebra}, such that
	\begin{equation*}
		Spin(d):=\{v_1\cdot\ldots\cdot v_{2k}\in Cl(d)\,\vert\,v_i\in\mathbb{R}^d,\, \delta(v_i,v_i)=1,\; k\in\mathbb{N}\}.
	\end{equation*}
In addition we can introduce a degree map $\mathrm{deg}:Spin(d)\to\mathbb{N}$ such that $s\mapsto\mathrm{deg}(s)=k$ if $s=v_1\cdot\ldots\cdot v_{2k}$.
\end{definition}
	
It is important to stress that $Spin(d)$ comes endowed with a natural structure as Lie group and that one can replace $\delta$ in Definition \ref{Def: Spin Group} with a generic Riemannian metric obtaining an equivalent group. 	
	
	\begin{remark}
		Identifying $\mathbb{R}^d$ as a subspace of $Cl(d)$ via the embedding map $\gamma:\mathbb{R}^d\rightarrow Cl(d)$, $v\mapsto\gamma(v)=\sum\limits_{i=1}^n\gamma_i v_i$ see Definition \ref{Def: Clifford algebra}, we can define an homomorphism of Lie groups 
		\begin{align*}
			\lambda: Spin&(d)\rightarrow SO(d)\\
			&s\mapsto \mathcal{R}_s=\mathcal{R}(s,\cdot),
		\end{align*}
		where $R$ is the action of the spin group on $\mathbb{R}^d$, namely
		\begin{align*}
			\mathcal{R}:Spin(d)\times&\mathbb{R}^d\rightarrow\mathbb{R}^d\\
			&(s,v)\mapsto (-1)^{\text{deg}(s)}s\cdot x\cdot s^{-1}.
		\end{align*}
	Observe that, as a consequence the map $\lambda$ is a double covering map. For a detailed proof, see \cite[Thm. 2.10]{Lawson}.
	\end{remark}

In view of the preceding discussion, particularly Definition \ref{Def: Clifford algebra}, an orthonormal basis $\{e_i\}_{i=1}^d$ of $\mathbb{R}^d$ can also be read as a basis of $\{\gamma(e_i)\}_{i=1}^d$ of $Cl(d)$. Associated to it we can define the renown {\bf gamma matrices} $\gamma_i$, $i=1,\dots, d$
\begin{equation}\label{Eq: Gamma Matrices}
	\gamma_i=\sigma_d\circ\gamma(e_i)\in End(\Sigma_d),
\end{equation}
where $\sigma_d$ is the spin representation as per Definition \ref{Def: spinor space}. 




\end{document}